\documentclass[10pt, reqno]{amsart}
\usepackage{graphicx,amsmath,amssymb,epsfig}
\long\def\comment#1{}
\long\def\comment#1{}
\newtheorem{theorem}{Theorem}
\newtheorem{algorithm}{Algorithm}[section]

\newtheorem{corollary}{Corollary}
\newtheorem{lemma}{Lemma}

\theoremstyle{definition}
\newtheorem{definition}{Definition}

\newtheorem{remark}{Remark}[section]
\numberwithin{remark}{section}

\newcommand{\be}{\begin{eqnarray}}
\newcommand{\ee}{\end{eqnarray}}

\newcommand{\NN}{\mathbf{N}}

\newcommand{\ba}{\begin{array}}
\newcommand{\ea}{\end{array}}
\newcommand{\bs}{\begin{align}\begin{split}\nonumber}
\newcommand{\bsnumber}{\begin{align}\begin{split}}
\newcommand{\es}{\end{split}\end{align}}

\renewcommand{\[}{\left[}
\renewcommand{\]}{\right]}
\renewcommand{\hat}{\widehat}
\renewcommand{\qed}{\hfill{\tiny \ensuremath{\blacksquare} }}%

\newcommand{\cc}{\mathbf{c}}

\newcommand{\Gn}{\mathbb{G}_n}

\newcommand{\Ep}{{\mathrm{E}}}
\newcommand{\barEp}{\bar \Ep}
\newcommand{\En}{{\mathbb{E}_n}}

\renewcommand{\Pr}{{\mathrm{P}}}

\def\RR{ {\mathbb{R}}}

\def\supp{{\rm support}}

\newcommand{\semin}[1]{\phi_{{\rm min}}(#1)}
\newcommand{\semax}[1]{\phi_{{\rm max}}(#1)}
\renewcommand{\hat}{\widehat}
\renewcommand{\leq}{\leqslant}
\renewcommand{\geq}{\geqslant}
\newcommand{\sign}{ {\rm sign}}

\newcommand{\hvpi}{{\hat \varphi_{i}}}
\newcommand{\vpi}{{\varphi_{i}}}
\newcommand{\diag}{{\rm diag}}

\newcommand{\G}{{G}}
\newcommand{\CR}{{\mathcal{CR}}}

\begin{document}

\title[Post-Selection Inference for GLMs]
{Post-Selection Inference for Generalized Linear Models with Many Controls}

\author{Alexandre Belloni, Victor Chernozhukov and Ying Wei}
%\author[Belloni]{Alexandre Belloni$^*$}\thanks{$^*$Duke University, e-mail:abn5@duke.edu.}
%\author[Chernozhukov]{Victor Chernozhukov$^\dag$}\thanks{$^\dag$Massachusetts Institute of Technology, e-mail:vchern@mit.edu}
%\author[Wei]{Ying Wei$^\ddag$}\thanks{$^\ddag$Columbia University, e-mail:yw2148@columbia.edu}

\date{First version: November 2012,  this version \today.}

\begin{abstract}
This paper considers generalized linear models in the presence of many controls. We lay out a general methodology to estimate an effect of interest based on the construction of an instrument that immunize against model selection mistakes and apply it to the case of logistic binary choice model. More specifically we propose new methods for estimating and constructing confidence regions for a regression parameter of primary interest $\alpha_0$, a parameter in front of the regressor of interest, such as the treatment variable or a policy variable.  These methods allow to estimate $\alpha_0$ at the root-$n$ rate when the total number $p$ of other regressors, called controls, potentially  exceed the sample size $n$ using  sparsity assumptions.  The sparsity assumption means that there is a subset of $s<n$ controls which suffices to accurately approximate the nuisance part of the regression function. Importantly, the estimators and these resulting confidence regions are valid uniformly over $s$-sparse models satisfying $s^2\log^2 p = o(n)$ and other technical conditions. These procedures do not rely on traditional consistent model selection arguments for their validity. In fact, they are robust with respect to moderate model selection mistakes in variable selection. Under suitable conditions, the estimators are semi-parametrically efficient in the sense of attaining the semi-parametric efficiency bounds for the class of models in this paper.
~\\

\end{abstract}

\maketitle
%\section*{Contributions}
%
%Main contribution: uniform inference for $\alpha_0$ under $p\gg n$.
%\begin{itemize}
%\item $\sqrt{n}$-asymptotic normality
%\item inverse chi-square test
%\item specialized region that under $H_0:\alpha_0=0$ is valid under a weaker requirement on $s$ by exploiting the heteroscedasticity wrt $x_i$ only
%\end{itemize}
%Secondary contributions:
%\begin{itemize}
%\item new identification condition for Logistic Regression that builds upon self-concordance result in Bach (2010) but seems to relax some technical conditions
%\item sparsity bound that does not rely on irrepresentability condition
%\item post-Lasso-Logistic results
%\end{itemize}
%
%\pagebreak

\section{Introduction}

The literature on high-dimensional generalized linear models has experienced rapid development \cite{vdGeer,Unified2012}.  As in the case of linear mean regression models, a striking result of this literature is to achieve consistency when the total number of covariates $p$ is potentially much larger than the sample size $n$. The main underlying assumption for achieving consistency is sparsity, namely that the number of relevant controls is at most $s$, which is much smaller than $n$.  Much of the interest focuses on $\ell_1$-penalized estimators that achieve desirable theoretical and computational properties, at least when the log-likelihood functions are concave. The theoretical properties are analogous to those of the corresponding $\ell_1$-penalized least squares estimator for linear mean regression models called Lasso \cite{T1996,BickelRitovTsybakov2009}. Results include prediction error consistency, consistency of the parameter estimates in $\ell_k$-norms, variable selection consistency, and minimax-optimal rates.

Several papers have focused on high-dimensional logistic binary choice models, trying to exploit their structure in detail. $\ell_1$-penalized logistic regressions models were studied in \cite{Bunea2008}, \cite{Bach2010}, and \cite{Kwemou2012}. Group logistic regression were studied in \cite{GroupLogistic2008} and \cite{Kwemou2012} to exploit addition sparsity patterns. Ising models were considered in \cite{Ising2010} and connections with robust $1$-bit recovery were derived in \cite{PlanVershynin2012}. These works also derive rates of convergence for the coefficients, prediction error consistency, and variable selection consistency under various conditions.

This paper attacks the problem of estimation and inference on a regression coefficient of interest in generalized linear models allowing for the total number of covariates $p$ to be much larger than the sample size $n$. Specifically, we construct $\sqrt{n}$-consistent estimators and confidence regions for a parameter of interest $\alpha_0$, which measures the  impact of a regressor of interest -- typically a ``policy variable"-- on the regression function. Importantly, we show that the estimator is $\sqrt{n}$-consistent and the confidence regions achieve the required asymptotic coverage uniformly over many data-generating processes.  We discuss the model framework for generalized linear models and then provide estimators to the logistic regression case to illustrate the results.

It is important to note that our estimation and inferential results are valid  \textit{without assuming}  the conventional ``separation condition" -- namely, without assuming that all the non-zero coefficients are sufficiently separated from zero.  Although the  separation condition is commonly used and might be appealing in some technometric applications (e.g. signal processing), it is often unrealistic and not credible in econometric, biometric, and many other applications. Even if applicable, it might not lead to  accurate approximations of the finite sample behavior of estimation and inference procedures. Our procedures are robust to  violation of the separation condition, and, thus, are robust to moderate model selection mistakes which inevitably occur in many applications (mistakes are very likely to occur when some coefficients are at the range of $O(n^{-1/2})$ which is typically not distinguishable from zero).

Our work contributes to a growing literature that avoids imposing  separation conditions. In the context of instrumental regression, \cite{BellChernHans:Gauss} and \cite{BellChenChernHans:nonGauss} provide uniformly valid estimation and inference methods for instrumental variable models, using either  post-selection or $\ell_1$-regularization methods to estimate ``optimal instruments."  They provide a $\sqrt{n}$-consistent, semi-parametrically efficient estimator of the main low-dimensional structural parameter. In the context of the linear mean regression model, \cite{BelloniChernozhukovHansenVolume,BelloniChernozhukovHansen2011} proposed a ``double selection" approach to constructing uniformly valid estimation and inference methods, and \cite{c.h.zhang:s.zhang} used one-step corrections to $\ell_1$-regularized estimators. In either case a $\sqrt{n}$-consistent, semi-parametrically efficient estimator of the  low-dimensional regression parameter of interest is provided.  In the the case of  linear quantile regression models,
\cite{BelloniChernozhukovKato2013a,BelloniChernozhukovKato2013b} provide uniformly $\sqrt{n}$-consistent estimators and uniformly valid inference methods for least absolute deviations and quantile regressions.  In an independent and contemporaneous work, \cite{vandeGeerBuhlmannRitov2013}  propose an approach to inference in generalized linear models, based upon the one-step correction of $\ell_1$-penalized estimator, where the pieces of the corrections are estimated via (approximate) Lasso inversion of  the sample information matrix; they also provide theoretical analysis under high-level conditions.  The approach taken in the present paper is an independent proposal, and relies instead on either optimal instrument strategy or the double selection strategy which is related to Neyman's approach to dealing with nuisance parameters. %Both strategies  build upon our own previous contributions (in the context of linear instrumental and mean regression) referenced above.

The aforementioned works as well as the current approach deviate substantially from the traditional approach of performing inference based upon perfect model selection results. \cite{leeb:potscher:hodges,potscher:leeb:dpe,potscher,LeebPotscher2005} have shown that such traditional/naive inference approach is not robust to violations of the separation condition, which bounds the magnitude of the non-zero coefficients away from zero. The naive post selection estimators and inference based upon them break down in the sense of failing to achieve $\sqrt{n}$-consistency and asymptotic normality when the separation condition is violated.  We shall confirm the failure of such naive post selection procedures in Monte-Carlo experiments.  In sharp contrast our procedure, by construction, is robust to violation of such assumptions.
We shall demonstrate this via theoretical results as well as via Monte-Carlo experiments.  The theoretical results hold uniformly in the class of $s$-sparse models and can be shown also to hold over approximately sparse models, using arguments similar to those used for linear mean and quantile models in \cite{BelloniChernozhukovHansen2011} and \cite{BelloniChernozhukovKato2013a,BelloniChernozhukovKato2013b}.

We construct our estimators and confidence regions via three steps. The first step use post-model selection methods to estimate the nuisance part of the regression -- the part of the regression function associated to controls (i.e., non-main regressors).  The second step uses post-model selection to estimate an optimal instrument. The third step suitably combines these estimates to form estimating equations that are immunized against crude estimation of the nuisance functions. Solutions of these equations lead to our proposed estimators and confidence regions. The framework allows for different methods to be used on each step leading to different estimators for generalized linear models. For the case of logit link function, we propose one estimator based upon instrumental logistic regression with optimal instrument and another estimator based upon double selection logistic regression. We verify the uniform validity of these procedures and demonstrate their good properties in a wide variety of experiments. While both implementations perform well, the double selection procedure emerged as the clear winner in these experiments.  Our results and proofs reveal that many different estimators can be used as ingredient in the three steps of the algorithm, as long as a required sparsity and rates for estimating nuisance functions are achieved. For example, the first and second steps can be based not only on post-selection estimators but also on $\ell_1$-regularized estimators, while the third step can be alternatively approximated by a one-step correction from an initial value. Therefore several implementations having the same asymptotic properties are possible. We narrowed down our formal theoretical analysis to the set of procedures that exhibited the best performance in  Monte-Carlo experiments (for example, Lasso methods performed worse than post-Lasso methods for estimating the nuisance parts, and one-step corrections performed worse than the exact solution of the estimating equation). One of the main results is to establish $\sqrt{n}$-consistency and asymptotic normality of estimators for generalized linear modes under high-level conditions on nuisance parameters.

Our constructions of the final estimators and confidence regions  mainly make use of the post-model selection estimators in estimating the nuisance part of the regression function as well as the optimal instrument.  As mentioned earlier, we focus on using selection as a means of regularization (which is necessary when $p>n$), mainly because compared to other methods of regularization, such as $\ell_1$-penalized maximum likelihood, they performed best in  a wide set of experiments. In order to develop sharp results for these estimators we must control sparsity effectively. We therefore provide sparsity bounds for $\ell_1$-penalized logistic maximum likelihood estimators, which is used for selection, and also derive the rates of convergence for the post-model selection logistic maximum likelihood estimator. These results  are of independent interest.   In the estimation of optimal instruments, which we  use as an ingredient in building the optimal estimating equation  to create immunization property, we rely on post-selection least squares estimator with data dependent weights.  The presence of data-dependent weights creates several interesting technical challenges. Finally,  to obtain the asymptotic approximations to the estimators of regression coefficients of interest we rely on empirical process methods, using self-normalized maximal inequalities and entropy calculations that rely on the sparsity of the models selected via data-driven procedures. These proofs are of independent interests in other types of generalized linear models.

We organize the remainder of the paper as follows. In Section \ref{Sec:Model}, we present the framework for generalized linear models and the proposed estimators specialized to the logistic link function case. In Section \ref{Sec:Main} we provide the statements of our main results on the uniform validity of the estimators and confidence regions. We present primitive conditions for the logistic case and high-level conditions for generalized linear models. Section \ref{SecMC} contains a Monte-Carlo experiment. We present the proofs of these results in Appendix \ref{Sec:ProofMain}.  In Appendix \ref{Sec:AnalysisAux} we collect results on Lasso and Post-Lasso with estimated weights (Appendix \ref{Sec:EstLasso}) as well
as results on $\ell_1$-penalized Logistic regression and post model selection Logistic regression (Appendix \ref{Sec:Step1}). In Appendix \ref{Sec:AuxiliaryInequalities} we present auxiliary inequalities.

\subsection{Notation}
Denote by $(\Omega,\Pr)$ the underlying probability space.
The notation $\En[\cdot]$ denotes the
average over index $1 \leq i \leq n$, i.e., it simply abbreviates the notation $n^{-1} \sum_{i=1}^n[\cdot]$. For example, $\En[x_{ij}^{2}] = n^{-1} \sum_{i=1}^{n}x_{ij}^{2}$. Moreover, we use the notation $\barEp[\cdot]=\En[\Ep[\cdot]]$. For example, $\barEp[ v_{i}^{2} ] = n^{-1}\sum_{i=1}^n\Ep[ v_{i}^{2} ]$. For a function $f: \RR \times \RR \times \RR^{p} \to \RR$,  we write $\Gn (f) = n^{-1/2} \sum_{i=1}^{n} (f(y_{i},d_{i},x_{i}) - \Ep [ f(y_{i},d_{i},x_{i}) ])$.
% In what follows, we work with triangular array data $\{\(\omega_{i,n}, i=1,...,n\), n=1,2,3,...\}$
%defined on probability space $(\Omega, \mathcal{A}, \Pr_n)$,
%where $\Pr = \Pr_n$ can change with $n$.    Each  $\omega_{i,n}= (y_{i,n}', x_{i,n}', d_{i,n}')'$
%is a vector with components defined below, and these vectors are i.n.i.d. -- independent across $i$, but not necessarily identically distributed. Thus, all parameters that characterize the distribution of  $\{\omega_{i,n}, i=1,...,n\}$ are
%implicitly indexed by $\Pr_n$ and thus by $n$.  We omit the dependence on these objects from the notation in what follows for notational simplicity.  We use the following empirical process notation: $\En[f] := \En[f(\omega_i)] := \sum_{i=1}^n f(\omega_i)/n,$  and $\Gn(f) := \sum_{i=1}^n ( f(\omega_i)
%- \Ep[f(\omega_i)] )/\sqrt{n}.$
%Since we want to deal with i.n.i.d. data, we also introduce the average expectation operator:
%$\barEp[f] := \Ep \En[f] =  \Ep \En[f(\omega_i)] = \sum_{i=1}^n \Ep[f(\omega_i)]/n$.
We denote the ${l}_1$-norm as $\|\cdot\|_1$, ${l}_2$-norm as $\|\cdot\|$, $l_\infty$-norm as $\|\cdot\|_\infty$, and the ``${l}_0$-norm" as $\|\cdot\|_0$ to denote the number of non-zero components of a vector.   For a sequence $(t_{i})_{i=1}^{n}$, we denote $\| t_{i} \|_{2,n} = \sqrt{ \En [ t_{i}^{2} ]}$.
For example, for a vector $\delta \in \RR^{p}$, $\| x_{i}'\delta \|_{2,n} = \sqrt{ \En [(x_{i}'\delta)^{2}]}$ denotes the prediction norm of $\delta$.
Given a vector $\delta \in \RR^p$, and a set of
indices $T \subseteq \{1,\ldots,p\}$, we denote by $\delta_T \in \RR^p$ the vector such that $(\delta_{T})_{j} = \delta_j$ if $j\in T$ and $(\delta_{T})_{j}=0$ if $j \notin T$. The support of $\delta$ as $\supp(\delta) = \{ j \in \{1,...,p\}: \delta_j \neq 0\} $.
We use the notation $(a)_+ = \max\{a,0\}$, $a \vee b = \max\{ a, b\}$, and $a \wedge b = \min\{ a , b \}$. We also use the notation $a \lesssim b$ to denote $a \leqslant c b$ for some constant $c>0$ that does not depend on $n$; and $a\lesssim_P b$ to denote $a=O_P(b)$. We assume that the quantities such as $p$, $s$, $y_{i}, d_i, x_{i}, \beta_{0}, \theta_{0}, T$ and $T_{\theta_0}$ are all dependent on the sample size $n$, and allow for the case where $p=p_{n} \to \infty$ and $s=s_{n} \to \infty$ as $n \to \infty$. We omit the dependence of these quantities on $n$ for notational convenience.

\section{Generic Setup and Method}\label{Sec:Model}

Consider a generalized linear regression model, where the outcome of interest $y_i$ relates to a scalar main regressor $d_i$ (e.g. a
treatment or a policy variable) and $p$-dimensional controls $x_i$ through a link function $\G$, namely for $i=1,\ldots,n$
\begin{equation}\label{Eq:MainLogisticModel} \Ep[y_i\mid x_i, d_i] = \G( d_i \alpha_0 + x_i'\beta_0 ).\end{equation}
Here $\alpha_0$ is the main target parameter, and $x_i'\beta_0$ is the nuisance regression function. The vector $\beta_0$ is a high-dimensional parameter which is assumed to be sparse, namely $\|\beta_0\|_0 \leq s$.  We require $s$ to be small relative to $n$ in the sense that will
be specified below, in particular
$$
\frac{s^2 \log^2 (p\vee n)}{n} \to 0
$$
is required.  In many settings this condition allows for the estimation of the nuisance function at the rate of $o(n^{-1/4})$.

Let $\{(y_i,d_i,x_i) \ : \ i=1,\ldots,n\}$ be a random sample, independent across $i$, obeying the model (\ref{Eq:MainLogisticModel}) with $\|\beta_0\|_0\leq s$. We aim to perform statistical inference on the coefficient $\alpha_0$ that is robust to moderate model selection mistakes as those are unavoidable if coefficients are near zero. Our proposed methods rely (implicitly or explicitly) on an instrument $z_{0i}=z_0(d_i,x_i)$ such that:
 \begin{eqnarray}
 \label{Eq:Estimating}& \Ep[\{y_i-\G(d_i\alpha_0+x_i'\beta_0)\}z_{0i}]  & = 0, \\
\label{Eq:Estimating1}& \left.\frac{\partial}{\partial\alpha}\Ep[\{y_i-\G(d_i\alpha+x_i'\beta_0)\}z_{0i}] \right|_{\alpha=\alpha_0} & \neq 0, \\
\label{Eq:Estimating2} & \left.\frac{\partial}{\partial\beta}\Ep[\{y_i-\G(d_i\alpha_0+x_i'\beta)\}z_{0i}]\right|_{\beta=\beta_0} & =0.\end{eqnarray} The first and second relations provide an estimating equation for $\alpha_0$. Relation (\ref{Eq:Estimating2}) is key in our analysis and states
 that the estimating equation (\ref{Eq:Estimating}) is insensitive with respect to first order perturbations of the nuisance function $x_i'\beta_0$. We call this orthogonality condition ``immunity." Such immunization ideas can be traced to Neyman's approach to dealing with nuisance parameters, as we discuss in Section \ref{Sec:Neyman}. (Note also that because of (\ref{Eq:MainLogisticModel}), there is also immunity with respect to perturbations on $z_0$.)

  Our methods proceed in three steps:
  \begin{itemize}
\item[1.] The first step computes an estimate for the nuisance function $x_i'\beta_0$.
\item[2.] The second step estimates the instrument $z_{0i}$.
\item[3.] The third step combines these estimates to estimate the parameter of interest $\alpha_0$.
\end{itemize}  Estimation of nuisance
  functions $x_i'\beta_0$ and the instrument $z_{0i}$ has an asymptotically negligible effect, due to the ``immunization properties"
  of the estimating equations.  Several different choices for these procedures and for instruments are possible. Next we provide detailed recommendations for their choices.

In general, we construct a valid, optimal instrument based on the following decomposition for the weighted main regressor:
\begin{equation}\label{Eq:Decomposition}
f_i d_i = f_i x_i'\theta_0 + v_i, \ \ \ \mbox{with} \ \ \Ep\left[f_i v_i x_i\right]=0,
\end{equation}  where \begin{equation}\label{Eq:wi} f_i: = w_i/\sigma_i,  \  \  w_i := G'(d_i\alpha_0+x_i'\beta_0),  \  \ \sigma_i^2:= \text{Var}(y_i|d_i,x_i),
\end{equation}
where $G'(t) = \frac{\partial}{\partial t} \G(t)$. The optimal instrument is given by
\begin{equation}\label{def: optim} z_{0i}:=v_i/\sigma_i.
 \end{equation}
 Here too we shall impose a sparsity condition in (\ref{Eq:Decomposition}), namely that $\|\theta_0\|_0\leq s$.  The use of sparsity in the main equation and this auxiliary equation can be generalized to approximate sparsity, with all results in this paper extending to this case, see Remark \ref{RemarkAlternative}.

The weights $f_i = w_i/\sigma_i$'s are used to achieve the orthogonality condition  (\ref{Eq:Estimating2}):
 \begin{equation}
\left. \frac{\partial}{\partial\beta}\Ep[\{y_i-\G(d_i\alpha_0+x_i'\beta)\}z_{0i}]\right|_{\beta=\beta_0} =
\Ep[ w_iz_{0i} x_i] =  \Ep\left[f_i v_i x_i\right] = 0;
 \end{equation}
 and this condition immunizes the estimation of  the main parameter $\alpha_0$ against crude estimation of the nuisance function $x_i'\beta_0$, in particular via post-selection estimators.  The selection steps make unavoidable moderate model selection mistakes, which translate into vanishing estimation error, which has an asymptotic negligible effect on the estimator based on the sample analog of the equation (\ref{Eq:Estimating}).  The orthogonality condition (\ref{Eq:Estimating2}) is therefore a critical ingredient in achieving asymptotic uniform validity of the coverage of confidence regions.  Among all instruments that provide such immunization, the instrument given in (\ref{def: optim}) minimizes the asymptotic variance of the asymptotically normal and $\sqrt{n}$-consistent estimator based on the estimating equations (\ref{Eq:Estimating}).   Other valid (but sub-optimal) choices of instruments are discussed in Remark \ref{Comment:ValidInstrument}. We will establish results for generalized linear models under high-level conditions.

 %Such decomposition always exists provided that the second moment of the treatment exist. Notably the weights are also a function of the treatment  unless $\alpha_0=0$.

%, i.e. a variable correlated to the treatment and orthogonal to the other controls where orthogonality is with respect the weighting $w_i$, namely $\barEp[w_id_iz_{0i}]\neq 0$ and $\Ep[w_iz_{0i}\mid x_i]=0$.

\subsection{Logistic Case and Specific Estimators}

Next we apply the above principle to the case of logistic regression and propose specific implementations of estimators. In this case the link function $G$ is given by the logistic link function $$\G(t) = \exp(t)/\{1+\exp(t)\},$$ and the following simplification occurs: $w_i$ in (\ref{Eq:wi}) equals the conditional variance of the outcome $\sigma_i^2$, namely $$w_i = \sigma_i^2=\G( d_i\alpha_0 + x_i'\beta_0 )\{1-\G(d_i\alpha_0+x_i'\beta_0)\}, \ \ \mbox{and} \ \  f_i = \sqrt{w_i},$$ so that the decomposition (\ref{Eq:Decomposition}) and the optimal instrument (\ref{def: optim}) become
\begin{equation}\label{Eq:Decomposition2}
\sqrt{w_i} d_i = \sqrt{w_i} x_i'\theta_0 + v_i, \ \ \ \ \Ep\left[\sqrt{w_i}v_i x_i\right]=0 \ \ \ \mbox{and} \ \ \  z_{0i}=v_i/\sigma_i = d_i-x_i'\theta_0.
\end{equation}
We describe two estimators  in Tables \ref{Table:Alg} and \ref{Table:Alg2}.
 In these tables we denote the (negative) log-likelihood function associated with the logistic link function as
\begin{equation} \Lambda(\alpha,\beta) = \En[\Lambda_i(\alpha,\beta)] = \En[\log\{1+\exp(d_i\alpha+x_i'\beta)\}-y_i(d_i\alpha+x_i'\beta)].
\end{equation}
% (Formal analysis of parameter choices and theoretical guarantees of the estimates of each step are provided in the Appendix.)

Table \ref{Table:Alg} displays an estimator based on the optimal instrument. The estimation in Step 1 is based on post-selection logistic regression where the model is selected based on $\ell_1$-penalized logistic regression. Step 2 is based on a post-selection least squares with estimated weights constructed based on Step 1. Note that Step 2 is used to construct the optimal instrument. Step 3 uses an instrumental logistic regression, with estimates of nuisance functions (control function $x_i'\beta_0$ and the instrument $z_{0i}$) obtained in Steps 1 and 2. The use of post-selection estimators in the first two steps instead of penalized estimators was motivated by a better finite sample performance in our experiments.  We also provide two confidence regions for $\alpha_0$ in Table \ref{Table:Alg}. The direct confidence region $\CR_D$ is based on the asymptotic normality of the estimator $\check \alpha$. The indirect confidence region $\CR_I$ is based on the asymptotic $\chi^2(1)$ law of the statistic $nL_n(\alpha_0)$.

\begin{table}[ht]
\hrule
\begin{center}\textbf{Estimators and Confidence Regions based on Optimal Instrument}\end{center}
\begin{enumerate}
\item[{\bf Step 1}] Run Post-Lasso-Logistic of $y_i$ on  $d_i$ and $x_i$:
$$\begin{array}{rl}
 (\hat\alpha,\hat\beta) \in & \arg{\displaystyle\min_{\alpha,\beta}} \ \ \En[\Lambda_i(\alpha,\beta)] + \frac{\lambda_1}{n}\|(\alpha,\beta)\|_1\\
(\widetilde\alpha,\widetilde\beta)\in & \arg{\displaystyle\min_{\alpha,\beta}} \ \ \En[\Lambda_i(\alpha,\beta)] \ : \ \supp(\beta)\subseteq \supp(\hat\beta)\end{array}$$
For $i=1,\ldots,n$, keep the value $x_i'\widetilde \beta$ and weight $$\hat f_i := \hat w_i/\hat \sigma_i, \text{ where }
\hat w_i = \G'(d_i\widetilde\alpha+x_i'\widetilde\beta), \ \ \hat\sigma_i^2 = \widehat{\text{Var}}(y_i|d_i,x_i)= \G(d_i\widetilde\alpha+x_i'\widetilde\beta) \{1-\G(d_i\widetilde\alpha+x_i'\widetilde\beta)\} .$$
\item[{\bf Step 2}] Run  Post-Lasso-OLS of $\hat f_i d_i$ on $\hat f_i x_i$:
$$\begin{array}{rl}
\ \ \ \ \ \ \ \ \ \ \ \ \ \ \hat\theta \in & \arg{\displaystyle\min_{ \theta}} \ \ \En[  \hat f_i^2 (  d_i -  x_i'\theta )^2] + \frac{\lambda_2}{n}\|\widehat \Gamma\theta\|_1\\
 \widetilde\theta \in & \arg{\displaystyle\min_{ \theta}} \ \ \En[ \hat f^2_i ( d_i - x_i'\theta )^2] \ : \ \ \supp(\theta)\subseteq \supp(\hat\theta)\end{array} $$
Keep the residual $\hat v_i:=\hat f_i (d_i-x_i'\widetilde\theta)$ and instrument $\hat z_i:= \hat v_i /\hat \sigma_i$, $i=1,\ldots,n$.
\item[{\bf Step 3}] Run Instrumental Logistic Regression of $y_i - x_i'\widetilde \beta$ on $d_i$
using $\hat z_i$ as the instrument for $d_i$
$$
\check \alpha \in \arg \inf_{\alpha \in \mathcal{A}} L_n(\alpha), \ \ \ \mbox{where} \ \ L_n(\alpha)  = \frac{| \ \En  [ \ \{y_i - \G(d_i\alpha+x_i'\widetilde\beta)\}\hat z_i \ ] \ |^2}{\En[ \ \{y_i - \G(d_i\alpha+x_i'\widetilde\beta)\}^2\hat z_i^2 \ ]}
$$
where $\mathcal{A} = \{ \alpha \in \RR : |\alpha -  \widetilde \alpha|\leq  C/\log n\}$. Define the confidence regions with asymptotic coverage $1-\xi$
$$\begin{array}{l} \CR_D = \{ \alpha \in \RR \ : |\alpha-\check\alpha|\leq \hat\Sigma_n\Phi^{-1}(1-\xi/2)/\sqrt{n}\}\\
%\CR_D = [ \check \alpha - \hat\Sigma_n\Phi^{-1}(1-\xi/2)/\sqrt{n}, \check\alpha + \hat\Sigma_n\Phi^{-1}(1-\xi/2)/\sqrt{n}] \\
\CR_{I} = \{ \alpha \in \mathcal{A} : nL_n(\alpha) \leq (1-\xi){\rm -quantile \  of} \ \chi^2(1)\}.\end{array}$$
  \end{enumerate}
\hrule
\vspace{0.2cm}
\caption{{\small The algorithm has three steps: (1) initial estimation of the regression function via post-selection logistic regression, (2) estimation of instruments which are orthogonal to the weighted controls via a weighted post-selection least squares, and (3) estimation of $\alpha_0$ based on the nuisance estimates obtain in (1) and (2).  Without loss of generality We assume the normalization $\En[x_{ij}^2]=1$ and $\En[d_i^2]=1$, and penalty parameters $\lambda_1 = \frac{1.1}{2}\sqrt{n}\Phi^{-1}(1-0.05/\{n\vee p\log n\})$, $\lambda_2 = 1.1\sqrt{n}2\Phi^{-1}(1-0.05/\{n \vee p\log n\})$ and $\widehat \Gamma$ is defined in the appendix, see (\ref{choice of loadings2}). The estimator of the variance is given by $\hat \Sigma^2_n = \max\{\widehat \Sigma_{1n}^2, \widehat \Sigma_{2n}^2\}$ where $\widehat\Sigma_{1n}^2 = \{\En[\hat w_i d_i \hat z_i]\}^{-1}\En[\{y_i-G(d_i\check\alpha+x_i'\widetilde\beta)\}^2\hat z_i^2]\{\En[\hat w_i d_i\hat z_i]\}^{-1}$ and $\widehat\Sigma_{2n}^2 = \En[\hat v_i^2 ]$.}}\label{Table:Alg} \end{table}

Table \ref{Table:Alg2} describes a second estimator, which builds upon the idea of the double selection method proposed in \cite{BelloniChernozhukovHansen2011} for partial linear mean regression models. The method replaces Step 3 in Table \ref{Table:Alg} with a (weighted) logistic regression of the outcome on the main regressor as well as the union of controls  selected in two selection steps -- Steps 1 and 2. (Note that the algorithm is stated for any generalized linear model in which case Step 3 is a weighted regression where the weights are given by $\hat f_i/\hat \sigma_i$ which equals to $1$ in the case of a logistic link function.) This approach creates an optimal instrument implicitly.  In fact, inspection of the proof shows that the double selection estimator  can be seen as an infinitely iterated version of the previous method. We refer to Section \ref{Sec:Connections} for further connections and discussions.

\begin{table}[ht]
\hrule
\begin{center}\textbf{Estimators and Confidence Region based on Double Selection}\end{center}
\begin{enumerate}
\item[{\bf Step 1}] Run Post-Lasso-Logistic of $y_i$ on  $d_i$ and $x_i$:
$$\begin{array}{rl}
 (\hat\alpha,\hat\beta) \in & \arg{\displaystyle\min_{\alpha,\beta}} \ \ \En[\Lambda_i(\alpha,\beta)] + \frac{\lambda_1}{n}\|(\alpha,\beta)\|_1\\
(\widetilde\alpha,\widetilde\beta)\in & \arg{\displaystyle\min_{\alpha,\beta}} \ \ \En[\Lambda_i(\alpha,\beta)] \ : \ \supp(\beta)\subseteq \supp(\hat\beta)\end{array}$$
For $i=1,\ldots,n$, construct the weights $$\hat f_i := \hat w_i/\hat \sigma_i, \text{ where }
\hat w_i = \G'(d_i\widetilde\alpha+x_i'\widetilde\beta), \ \ \hat\sigma_i^2 = \widehat{\text{Var}}(y_i|d_i,x_i)= \G(d_i\widetilde\alpha+x_i'\widetilde\beta) \{1-\G(d_i\widetilde\alpha+x_i'\widetilde\beta)\} .$$
\item[{\bf Step 2}] Run  Lasso-OLS of $\hat f_i d_i$ on $\hat f_i x_i$:
$$\begin{array}{rl}
 \hat\theta \in & \arg{\displaystyle\min_{ \theta}} \ \ \En[ \hat f^2_i ( d_i - x_i'\theta )^2] + \frac{\lambda_2}{n}\|\widehat\Gamma\theta\|_1 \ \ \ \ \  \end{array} $$
\item[{\bf Step 3}] Run Post-Lasso-Logistic of $y_i$ on  $d_i$ and the covariates selected in Step 1 and 2:
$$\begin{array}{rl}
\ \ \ \ \ \ \ \ \ \ \  \ \ \ \ \ \ (\check\alpha,\check\beta)\in & \arg{\displaystyle\min_{\alpha,\beta}} \ \ \En[ \Lambda_i(\alpha,\beta) \hat f_i/\hat\sigma_i] \ : \ \supp(\beta)\subseteq \supp(\hat\beta) \cup \supp(\hat\theta)\end{array}$$
 Define the confidence region with asymptotic coverage $1-\xi$ as
$$\begin{array}{l} \CR_{DS} = \{ \alpha \in \RR \ : |\alpha-\check\alpha|\leq\hat\Sigma_n\Phi^{-1}(1-\xi/2)/\sqrt{n}\}.
%\CR_{DS} = [ \check \alpha - \hat\sigma_n\Phi^{-1}(1-\xi/2)/\sqrt{n}, \check\alpha + \hat\Sigma_n\Phi^{-1}(1-\xi/2)/\sqrt{n}].
\end{array}$$  \end{enumerate}
\hrule
\vspace{0.2cm}
\caption{{\small The double selection algorithm has three steps: (1) use $\ell_1$-penalized logistic regression to select
covariates; and  use post-selection logistic regression to estimate the weights to be used in the next step, (2)
select covariates based on the weighted post-selection least squares, where the dependent variable is the main regressor and the
independent variables are the rest of the regressors, and (3) run a Logistic regression of the outcome
 on the main regressors and the union of controls in steps (1) and (2). Without loss of generality we assume the normalization $\En[x_{ij}^2]=1$ and $\En[d_i^2]=1$, and penalty parameters $\lambda_1 = \frac{1.1}{2}\sqrt{n}\Phi^{-1}(1-0.05/\{n \vee p\log n\})$, $\lambda_2 = 1.1\sqrt{n}2\Phi^{-1}(1-0.05/\{n \vee p\log n\})$ and $\widehat \Gamma$ is defined in the appendix, see (\ref{choice of loadings2}). The estimator of the variance is given by $\hat \Sigma^2_n = \max\{\hat \Sigma^2_{1n},\hat \Sigma^2_{2n}\}$ where $\hat \Sigma^2_{1n}=\{\En[\check w_i d_i \hat z_i]\}^{-1}\En[\{y_i-G(d_i\check\alpha+x_i'\check\beta)\}^2\hat z_i^2]\{\En[\check w_i d_i\hat z_i]\}^{-1}$, $\Sigma^2_{2n}=\{ \En[\check w_i(d_i,\check x_{i}')'(d_i,\check x_i')]\}^{-1}_{11}$, $\check w_i = \G(d_i\check\alpha + x_i'\check\beta)\{1-\G(d_i\check\alpha+x_i'\check\beta)\}$ and $\check x_i = x_{i,\supp(\check\beta)}$.}}\label{Table:Alg2} \end{table}

\begin{remark}[Other Valid Instruments]\label{Comment:ValidInstrument} An instrument $z_0$ is valid if it has the orthogonality property $$\left. \frac{\partial}{\partial\beta}\Ep[\{y_i-\G(d_i\alpha_0+x_i'\beta)\}z_{0i}]\right|_{\beta=\beta_0} =
\Ep[ w_iz_{0i} x_i]=0$$ and is non-trivial, namely $\barEp[w_id_iz_{0i}] \neq 0$.   A valid, non-trivial instrument is optimal if it minimizes the asymptotic variance of the final estimator of $\alpha_0$.  The algorithm stated in Table \ref{Table:Alg} uses the optimal
instrument $z_{0i} := v_i/\sigma_i$. Estimation of this instrument requires that in Step 2 a Lasso method is applied in the weighted equation (\ref{Eq:Decomposition}).  Since the weights $w_i$'s in the resulting weighted Lasso problem are estimated, with estimation errors depending upon the response variable $d_i$, estimation of the optimal instrument creates interesting technical challenges in the analysis of Lasso or Post-Lasso that are dealt with in the Appendix.  Thus estimation of the optimal instruments poses an interesting problem in its own right. There are other valid instruments that we can rely on, but these instruments are not generally optimal. For example, a valid, yet sub-optimal choice of the instrument is $z_{0i} := (d_i-\Ep[d_i\mid x_i])/w_i$.    The estimation of this instrument is technically simpler, and follows easily from available results. Indeed, assuming $\Ep[d_i\mid x_i] = x_i'\theta_d$, with $\theta_d$ sparse or approximately sparse, we can estimate $z_{0i}$ by estimating $\theta_d$ via standard Lasso of $d_i$ on $x_i$, and estimating $w_i$ using the estimates of the $\ell_1$-penalized logistic regression as in Step 1. Note that since no estimated weights are used in Lasso estimation of $\theta_d$, standard results on the Lasso estimator deliver the required properties. \qed
\end{remark}

\begin{remark}[Alternative Implementations via Approximate Instrumental Regression]\label{RemarkAlternative}
 The instrumental logistic regression can be approximately implemented by a 1-Step estimator from the $\ell_1$-penalized logistic estimator $\hat\alpha$ of the form $ \check \alpha = \hat \alpha + (\En[\hat w_i d_i \hat z_i])^{-1} \En[\{y_i- \G(d_{i} \hat\alpha +x_i'\hat \beta)\}\hat z_i]$.  However, we prefer the exact implementations, since the they perform better in an extensive set of Monte-Carlo experiments. \qed
\end{remark}

\begin{remark}[Data Driven Choice of Penalty Parameters]\label{RemarkDataDriven}
The penalty parameters $\lambda_1$ and $\lambda_2$ as defined in the algorithms above are motivated by self-normalized moderate deviation theory. Other data-driven choices are possible but their theoretical validity is outside the scope of this paper. For example, cross validation typically underpenalize to reduce bias to obtain better estimates but tends to select a substantial larger number of variables. This suggests that cross validation to be more suitable for the algorithm based on optimal instrument than for the algorithm based on double selection. Another approach suggested in Section 4.2 of \cite{chernozhukov2013gaussian}  relies on new  Gaussian approximation results and can be implemented via a multiplier bootstrap procedure.\qed
\end{remark}

\section{Main Theoretical Results}\label{Sec:Main}

\subsection{Logistic Regression under Primitive Assumptions}

%We consider the following quantities associated with the covariates in the sample $\tilde x_i = (d_i,x_i')'$, $i=1,\ldots,n$. We denote the largest value of the components of the controls $x_i$ as $K_x = \max_{i\leq n}\| x_i\|_\infty$ and denote the minimum and maximum $m$-sparse empirical eigenvalues as $$ \semin{m} := \min_{1\leq \|\delta\|_0 \leq m} \frac{\|\tilde x_i'\delta\|_{2,n}^2}{\|\delta\|^2} \ \ \ \mbox{and} \ \ \ \semax{m} := \max_{1\leq \|\delta\|_0 \leq m} \frac{\|\tilde x_i'\delta\|_{2,n}^2}{\|\delta\|^2}.$$
%The following are sufficient primitive conditions.

In this section, we list and discuss primitive conditions that allow us to derive our results in the case of logistic regression. These conditions ensure good properties of $\ell_1$-penalized methods and the associated post-selection estimators. Fix some sequences of constants, $\delta_n\to 0$, and $\Delta_n \to 0$, and constants $0 < c < C <\infty$.

{\bf Condition L.} { (i) Let $\{(y_i,d_i,x_i):i=,\ldots,n\}$ be independent random vectors that obey the model given by (\ref{Eq:MainLogisticModel}) and (\ref{Eq:Decomposition}) with $G$ being the logistic function. There exists $s=s_n$ such that $\|\beta_0\|_0+\|\theta_0\|_0\leq s$, $\|\beta_0\|+\|\theta_0\| \leq C$. (ii) The following moment conditions hold $\barEp[\{(d,x')\xi\}^4]\leq C \|\xi\|^4$, $\barEp[w_i\{(d,x')\xi\}^2]\geq c\|\xi\|^2$. We have that $\min_{j\leq p} \barEp[w_ix_{ij}^2v_i^2] \geq c>0$ and $\max_{j\leq p} \barEp[|\sqrt{w_i}x_{ij}v_i|^3]^{1/3}\log^{1/2}(p\vee n) \leq \delta_n n^{1/6}$. Furthermore, the conditional variance $\sigma_i^2$ satisfy $\min_{i\leq n} \sigma_i^2 \geq c>0$ with probability $1-\Delta_n$. (iii) For $K_q = \Ep[\max_{i\leq n} \|(d_i,z_{0i},x')\|_\infty^q]^{1/q}$, we have $K_1^2s^2\log^2(p\vee n) \leq \delta_n n$ and $K_4^4 s  \log (p\vee n) \log^{3}n\leq \delta_n n$.} % $K_1^4s\log(p\vee n) \leq \delta_n n$,  %$K_4^4\log p \leq \delta_n n$. $s^2\log^2(p\vee n) \leq \delta_n n$ $K_1^2s^2\log(p\vee n) \leq \delta_n n$

%{\bf Condition L.} \textit{(i) Let $(x_i)_{i=1}^n$ denote a sequence of non-stochastic vectors in $\RR^p$ of covariates normalized so that $\En[x_{ij}^2]=1$, $j=1,\ldots, p$, and $\{(y_i,d_i,v_i,w_i):i=1,\ldots,n\}$ be independent random vectors that obey the model given by (\ref{Eq:MainLogisticModel}) and (\ref{Eq:Decomposition}). There exists $s=s_n$ such that $\|\beta_0\|_0 + \|\theta_0 \|_0\leq s$. (ii) The weights $w_i$ satisfy $\min_{i\leq n} w_i \geq c > 0$ with probability $1-\Delta_n$, and the following moment conditions hold
% $ 0 < c \leq \barEp[ v_i^2 \mid x_i ] \leq \max_{i\leq n} \{ \Ep[ v_i^4/w_i^2 \mid x_i ]\}^{1/2} \vee \{ \Ep[ d_i^4 \mid x_i ]\}^{1/2} \leq C,$ $\barEp[ v_i^8 \mid x_i ] \leq C$. (iii) The sparse minimal and maximal eigenvalues are bounded, $c\leq \semin{s \ell_n } \leq \semax{s \ell_n } \leq C$ with probability $1-\Delta_n$. (iv) The sparsity index  $s$ and overall number of controls $p$ obey the following growth conditions $K_x^2s^2\log^3(p\vee n) \leq n\delta_n$ and $K_x^{q}\log(p\vee n) \leq \delta_n n$ for some fixed $q>4$.}

Condition L(i) assumes independence across $i$ and the model described in Section \ref{Sec:Model} and sparsity conditions which makes estimation possible even if $p>n$. Condition L(ii) assumes the conditional variance is bounded away from zero and imposes mild moment conditions. Condition L(iii) imposes growth requirements on the triple $(s,p,n)$ as $n$ grows.  An important consequence of Condition L is to  imply that submatrices of the design matrix are well behaved even though the design matrix cannot have rank $p$ if $p>n$; see \cite{RudelsonVershynin2008,RudelsonZhou2011} for detailed discussion. This ensures that $\ell_1$-penalized estimators are well behaved with suitable choices of penalty parameters under the stated sparsity assumptions.  %We refer the interest reader to Lemma \ref{Thm:EstIVQRgeneric} in the Appendix for the weaker conditions that also imply our results.

Next we state the main inferential results of the paper. It concerns the (uniform) validity of the different  confidence regions for the coefficient $\alpha_0$ based on the optimal instrument and double selection algorithms. %In addition to the primitive conditions listed in the previous section, we also require some growth conditions on how fast the sparsity parameter $s$ and the total number of controls $p$ can growth relative to the sample size $n$.

\begin{theorem}[Robust Estimation and Inference based on the Optimal IV Estimator]\label{theorem:inferenceAlg1}  Consider any triangular array of data $(y_i, d_i, x_i)_{i=1}^n$ that obeys  Condition L for all $n\geq 1$. Then, the estimator  $\check \alpha$ based on the optimal instrument, as defined in Table \ref{Table:Alg}, obeys as $n \to \infty$
$$
\Sigma_n^{-1} \sqrt{n} (\check \alpha - \alpha_0) = Z_n + o_\Pr(1),  \ \   Z_n \rightsquigarrow N(0,1),
$$
where
$$
Z_n := \frac{\Sigma_n}{\sqrt{n}}  \sum_{i=1}^n \{y_i - G(d_i \alpha_0 + x_i'\beta_0)\} z_{0i} \text{ and } \  \Sigma^2_n:=  \barEp[v_i^2]^{-1}.
$$
 Moreover,
$$ nL_n(\alpha_0) = Z_n^2 + o_\Pr(1), \  \ Z_n^2 \rightsquigarrow \chi^2(1). $$ Finally, $\Sigma_n^2$ can be replaced by either $\hat \Sigma^2_{1n} = \{\En[\hat w_i d_i \hat z_i]\}^{-1}\En[\{y_i-G(d_i\check\alpha+x_i'\widetilde\beta)\}^2\hat z_i^2]\{\En[\hat w_i d_i \hat z_i]\}^{-1}$ or by $\hat \Sigma_{2n}^2 = {\En[\hat v_i^2]}^{-1}$ without affecting the result, i.e. $\hat \Sigma_{1n}^2/\Sigma_n^2 = 1 + o_{\Pr}(1)$ and $\hat \Sigma_{2n}^2/\Sigma_n^2 = 1 + o_{\Pr}(1)$.
\end{theorem}

Theorem \ref{theorem:inferenceAlg1} establishes that the IV estimator $\check \alpha$ is $\sqrt{n}$-consistent and asymptotically
normal. Under suitable conditions the large-sample variance coinciding with the semi-parametric efficiency bound for the partially linear logistic regression
model (see Section \ref{Sec:minimax} for an additional discussion).   The studentized estimator converges to the standard normal law,  and the criterion function that this estimator minimized, when evaluated at the true value, converges to the standard chi-squares law with one degree of freedom.  These results justify and imply the validity of the confidence regions $\CR_D$ and $\CR_I$ for $\alpha_0$ proposed in Table \ref{Table:Alg}.  We note that these results are achieved despite possible model selection mistakes in Steps 1 and 2.

The following result derives similar properties for the double selection estimator described in Table \ref{Table:Alg2}.

\begin{theorem}[Robust Estimation and Inference based on Double Selection]\label{theorem:inferenceAlg2}   Consider
any triangular array of data $(y_i, d_i, x_i)_{i=1}^n$ that obeys  Condition L for all $n\geq 1$. Then, the double selection estimator  $\check \alpha$ as defined in Table \ref{Table:Alg2} obeys as $n \to \infty$
$$
\Sigma_n^{-1} \sqrt{n} (\check \alpha - \alpha_0) = Z_n + o_\Pr(1),  \ \   Z_n \rightsquigarrow N(0,1),
$$
where
$$
Z_n := \frac{\Sigma_n}{\sqrt{n}}  \sum_{i=1}^n (y_i - G(d_i \alpha_0 + x_i'\beta_0)) z_{0i} \text{ and } \  \Sigma^2_n:=  \barEp[v_i^2]^{-1}.
$$
Moreover, $\Sigma_n^2$ can be replaced by  $\hat \Sigma^2_{1n} = \{\En[\check w_i d_i \hat z_i]\}^{-1}\En[\{y_i-G(d_i\check\alpha+x_i'\check\beta)\}^2\hat z_i^2]\{\En[\check w_i d_i \hat z_i]\}^{-1}$ or $\hat \Sigma^2_{2n}=\{\En[\check w_i (d_i,\check x_{i}')'(d_i,\check x_{i}')]\}^{-1}_{11}$ without affecting the result, i.e. $\hat\Sigma_{1n}^2/\Sigma_n^2=1+o_{\Pr}(1)$ and $\hat\Sigma_{2n}^2/\Sigma_n^2=1+o_{\Pr}(1)$, where $\check w_i= G(d_i\check\alpha+x_i'\check\beta)\{1-G(d_i\check\alpha+x_i'\check\beta)\}$ and $\check x_i = x_{i,\supp(\check\beta)}$.
\end{theorem}

Theorems \ref{theorem:inferenceAlg1} and \ref{theorem:inferenceAlg2} allow for the data-generating processes to change with $n$, in particular allowing sequences of regression models, with coefficients never perfectly distinguishable from zero, i.e. models where perfect model selection is not possible.  In turn, the results achieved in Theorems \ref{theorem:inferenceAlg1} and \ref{theorem:inferenceAlg2}  are uniformly valid over a large class of sparse models. %, which is referred to as ``honesty" in the statistical literature on construction of confidence intervals.
In what follows, we formalize these assertions as corollaries.

 Let $\mathcal{Q}_{n}$ denote a collection of distributions $Q_n$ for the data $\{ (y_{i},d_{i}, x_i')' \}_{i=1}^{n}$ such that Condition L hold for the given $n$.  This is the collection of all sparse models where the stated above sparsity conditions, moment conditions, and growth conditions hold.
 This collection expressly permits models to have near zero coefficients, and thus does not impose
the separation conditions (which we believe to be unreasonable in many applications). For $Q_n \in \mathcal{Q}_n$, let the notation $\Pr_{Q_{n}}$ mean that under $\Pr_{Q_{n}}$, $\{ (y_{i},d_{i}, x_i')' \}_{i=1}^{n}$ is distributed according to $Q_{n}$.

\begin{corollary}[\textbf{Uniform $\sqrt{n}$-Rate of Consistency and Uniform Normality}]
\label{cor:Uniformity1}Let $\mathcal{Q}_{n}$ be the collection of all distributions of $\{ (y_{i},d_{i},x_i')' \}_{i=1}^{n}$ for which Condition L is  satisfied for the given $n \geq 1$. Then the estimator $\check \alpha$, based either on optimal instrument or double selection, is  $\sqrt{n}$-consistent and asymptotically normal uniformly
over $\mathcal{Q}_n$, namely
 $$
\lim_{n \to \infty} \sup_{Q_n \in \mathcal{Q}_n}\sup_{t\in \RR} |\Pr_{Q_n} ( \Sigma_n^{-1} \sqrt{n} (\check \alpha - \alpha_0)  \leq t) - \Pr( N(0,1) \leq t)| =0
  $$
  Moreover, the result continues to hold if $\Sigma_n^2$ is replaced by any of the estimators $\hat \Sigma_n^2$ specified in the statements of the preceding theorems.
  \end{corollary}

\begin{corollary}[\textbf{Uniformly Valid Confidence Regions}]
\label{cor:Uniformity2}Let $\mathcal{Q}_{n}$ be the collection of all distributions of $\{ (y_{i},d_{i},x_i')' \}_{i=1}^{n}$ for which Condition L is  satisfied for the given $n \geq 1$. Then confidence regions  $\CR \in \{ \CR_D, \CR_I, \CR_{DS}\}$ are asymptotically valid uniformly in $n$,  namely
 $$
\lim_{n \to \infty} \sup_{\xi \in (0,1)} \sup_{Q_n \in \mathcal{Q}_n} |\Pr_{Q_n} ( \alpha_0 \in \CR ) - (1- \xi)| =0
  $$
  \end{corollary}
All of  the results are new under $s \to \infty$ and $p \to \infty$ asymptotics, and
 they are new even under the fixed $s$ and $p$ asymptotics. These results motivates interesting questions on the construction of  confidence regions for many parameters of interest that are simultaneously valid. We refer to \cite{BelloniChernozhukovKato2013a}, \cite{BCFH2013program}, and \cite{BCCW2015} where simultaneous confidence regions are proposed in a variety of settings.

\begin{remark}[\textbf{Generalization to Approximately Sparse Models}] The results can also be shown to hold, with identical conclusions,
 in the class of approximately sparse models, following the analysis of the partially linear mean regression model in \cite{BelloniChernozhukovHansenVolume,BelloniChernozhukovHansen2011}.
For example, if the model satisfies \begin{eqnarray}
 & \Ep[y_i \mid d_i,x_i ]  =  G(\alpha_0 d_i + x_i'\beta_0 +r_{yi}),  &  \\
 &  f_id_i  =  f_i x_i'\theta_0 +r_{di} + v_i, \ \  & \Ep[f_i v_ix_i] =0,
  \end{eqnarray}
 where $\|\beta_0\|_0\leq s$, $\|\theta_0\|_0\leq s$, and the approximation errors $r_{yi}$ and $r_{di}$ are such that
 \begin{equation}  \sqrt{\barEp[r^2_{yi}]} \leq  C \sqrt{ s/n } ,  \ \ \ \sqrt{\barEp[r^2_{di}]} \leq  C \sqrt{ s/n }, \ \ \mbox{and} \ \ |\barEp[f_i v_ir_{yi}]| \leq \delta_n n^{-1/2}
 \end{equation}
 % i.e., the $L^2(\Pr)$ size of the approximation errors does not exceed the size of the estimation error of an oracle estimator.
 We can show that the results in Theorems \ref{theorem:inferenceAlg1} and \ref{theorem:inferenceAlg2} and Corollaries \ref{cor:Uniformity1} and \ref{cor:Uniformity2} continue to hold for this approximately sparse model. This means that the results are robust with respect to moderate violations
 of the sparsity assumption. \qed
 %Simulations results discussed next also provide an additional evidence that these regions are substantial.
\end{remark}

%\begin{remark}[\textbf{Generalization to Using Other Valid Instruments}] In the Appendix, we establish a more general result for
% the IV estimator based on any valid instrument, as defined in Comment \ref{Comment:ValidInstrument}, under high-level conditions. Given any valid instrument $(z_{0i})_{i=1}^n$, we show that
%$$ \{ \barEp[\sqrt{w_i}v_iz_{0i}]^{-1} \barEp[w_iz_{0i}^2]\barEp[\sqrt{w_i}v_iz_{0i}]^{-1}\}^{-1/2}\sqrt{n} (\check \alpha - \alpha_0) \rightsquigarrow N(0,1).$$ Therefore, the choice of instrument can be guided by efficiency considerations. \qed %%%Theorems \ref{theorem:inferenceAlg1} and \ref{theorem:inferenceAlg2} establish that the proposed estimators achieve the semiparametric (minimax) efficiency bound for the partially linear logistic regression (see page 356 in \cite{kosorok:book}).
%\end{remark}

%\begin{remark}[\textbf{The Case of Testing $H_0:\alpha_0=0$}] In some applications, the main goal is on testing if the policy variable $d$ has an impact, i.e. the null hypothesis $H_0:\alpha_0=0$. Under $H_0$, the conditional variance $w_i$ of the outcome no longer varies with the policy variable. Specifically, when $\alpha_0=0$, we have $$\Ep[ \sqrt{w_i} v_i \mid x_i ] = \sqrt{w_i} \Ep[ v_i \mid x_i ]=0, i=1,\ldots,n.$$ In the Logistic model associated with (\ref{Eq:MainLogisticModel}), we have $w_i > 0$ which makes the condition above equivalent to $\Ep[v_i\mid x_i]=0$, $i=1,\ldots,n$. Therefore, one can estimate $\theta_0$ in (\ref{Eq:Decomposition}) using $\hat f_i = 1$. \qed
%\end{remark}

\subsection{Generalized Linear Models under High-Level Assumptions}

In this section we establish $\sqrt{n}$-consistency and asymptotic normality for an estimator $\check \alpha$ of $\alpha_0$ associated with a generalized linear model based on high-level conditions. These high-level conditions cover a variety of different estimators including the estimators described in Tables \ref{Table:Alg} and \ref{Table:Alg2}. In what follows note that the estimated instrument $\hat z_i=\hat z_i(d_i,x_i)$ and the expectations below are evaluated at the given estimates. % Let $(d,x) \in \mathcal{D}\times\mathcal{X}$. In this section for $\tilde h = (\tilde \beta, \tilde z)$, where $\tilde z$ is a function on $(d,x)\mapsto \tilde z(d,x)$ we write $$\psi_{\tilde \alpha,\tilde h}(y_i,d_i,x_i) = \psi_{\tilde \alpha,\tilde \beta,\tilde z}(y_i,d_i,x_i) = \{y_i - \G(x_i'\tilde\beta+d_i\alpha)\}\tilde z(d_i,x_i).$$ %= \{y_i - \G(x_i'\tilde\beta+d_i\alpha)\}\tilde z_i $$ Because of (\ref{Eq:MainLogisticModel}), for $h_0 = (\beta_0, z_0)$ we have $ \Ep[ \psi_{\alpha_0,h_0}(y_i,d_i,x_i)]=0$. The following is a sufficient set of high level conditions.

{\bf Condition IR.} (i) The data $\{y_i,d_i,x_i\}$ independent across $i=1,\ldots,n$, satisfies (\ref{Eq:MainLogisticModel}), $\sigma_i^2 = \text{Var}(y_i\mid d_i,x_i)$, $w_i = G'(d_i\alpha_0+x_i'\beta_0)$, and the link function $G$ such that $\sup_{t\in \RR}|G(t)|\leq C$, $\sup_{t\in \RR}|G'(t)|\leq C$ and $\sup_{t\in \RR}|G''(t)|\leq C$. (ii) The following moment conditions hold $\Ep[ w_iz_{0i}x_i]=0$, $ |\Ep[w_id_iz_{0i}]|\geq c > 0$, $\barEp[\sigma_i^2z_{0i}^2]\geq c > 0$, $\barEp[ z_{0i}^2d_i^2 ] \leq C$, $\barEp[\sigma_i^3z_{0i}^3]\leq C$, and  $\barEp[(x_i'\xi)^4] \leq  C$ for all $\|\xi\|=1$. (iii) For some sequences $\delta_n\to 0$ and $\Delta_n\to 0$ with probability at least $1-\Delta_n$, the estimates $(\check \alpha,\hat \beta,\hat z)$ satisfy
\begin{equation}\label{Eq:HLfirstestimated}
\begin{array}{c}   \|\hat\beta-\beta_0\| \leq \delta_n n^{-1/4},\  \left.\barEp[(\tilde z_i - z_{0i})^2]\right|_{\tilde z = \hat z} \leq \delta_n^2, \ \   \ \|\hat\beta-\beta_0\|\cdot \{\left.\barEp[(\tilde z_i - z_{0i})^2]\right|_{\tilde z = \hat z}\}^{1/2} \leq  \delta_n n^{-1/2}, \end{array}
\end{equation}

\vspace{-0.5cm}

%\begin{equation}\label{Eq:HLestimated}
%\displaystyle   \sup_{\alpha: |\alpha-\alpha_0|\leq \delta_n}\left|(\En-\barEp)\left[\psi_{\alpha,\hat h}(y_i,d_i,x_i)-\psi_{\alpha,h_0}(y_i,d_i,x_i)\right]\right| \leq \delta_n \ n^{-1/2}\\
%\end{equation}
\begin{equation}\label{Eq:HLestimated}
\displaystyle   \sup_{\alpha: |\alpha-\alpha_0|\leq \delta_n}\left|(\En-\barEp)\left[\{y_i-G(d_i\alpha+x_i'\hat\beta)\}\hat z_i - \{y_i-G(d_i\alpha+x_i'\beta_0)\}z_{0i}\right]\right| \leq \delta_n \ n^{-1/2}\\
\end{equation}

\vspace{-0.5cm}

%\begin{equation}\label{Eq:HLhatalpha}
%|\check \alpha-\alpha_0|\leq \delta_n \ \ \ \ \mbox{and} \ \ \ \left| \ \En[\psi_{\check\alpha,\hat h}(y_i,d_i,x_i)] \ \right|\leq \delta_n\ n^{-1/2}.
%\end{equation}

\begin{equation}\label{Eq:HLhatalpha}
|\check \alpha-\alpha_0|\leq \delta_n \ \ \ \ \mbox{and} \ \ \ \left| \ \En[\{y_i-G(d_i\check\alpha+x_i'\hat\beta)\}\hat z_i] \ \right|\leq \delta_n\ n^{-1/2}.
\end{equation}
(iv) with probability $1-\Delta_n$ we have $|\hat w_i|\leq C$, $\|\hat w_i-w_i\|_{2,n}\leq \delta_n$, $\|d_i(\hat w_i-w_i)\|_{2,n}\leq \delta_n$, $\|d_i\hat z_i\|_{2,n}\leq C$, $\|\hat z_i - z_{0i}\|_{2,n}\leq \delta_n$,  and $\|z_{0i}x_i'(\hat\beta-\beta_0)\|_{2,n}\leq \delta_n$.

This set of high-level conditions allow us to cover several generalized models of interest. In particular Condition L and the choices of post-selection methods described in the previous section imply Condition IR.
Next we formally state our main result for generalized linear models.

\begin{theorem}\label{Thm:EstIVQRgeneric}
Under Condition IR(i,ii,iii) we have
$$ \{\barEp[\sigma_i^2z_{0i}^2]\}^{-1/2}\barEp[w_id_iz_{0i}]\sqrt{n}(\check\alpha-\alpha_0) = \frac{ \{\barEp[\sigma_i^2z_{0i}^2]\}^{-1/2}}{\sqrt{n}}\sum_{i=1}^n\{y_i-G(d_i\alpha_0+x_i'\beta_0)\}z_{0i} + o_\Pr(1) $$
$$\mbox{and} \ \ \ \{\barEp[w_id_iz_{0i}]^{-1}\barEp[\sigma_i^2z_{0i}^2]\barEp[w_id_iz_{0i}]^{-1}\}^{-1/2}\sqrt{n}(\check \alpha-\alpha_0)\rightsquigarrow N(0,1).$$
Moreover, if Condition IR(iv) also holds, we have $$nL_n(\alpha_0) \rightsquigarrow \chi^2(1)$$ and the variance estimator is consistent, namely $$\En[\hat w_id_i\hat z_i]^{-1}\En[\{y_i-\G(d_i\check\alpha+x_i'\hat\beta)\}^2\hat z_i^2]\En[\hat w_i d_i \hat z_i]^{-1}=\barEp[w_id_iz_{0i}]^{-1}\barEp[\sigma_i^2z_{0i}^2]\barEp[w_id_iz_{0i}]^{-1}+o_\Pr(1).$$
\end{theorem}

It is important to note that Theorem \ref{Thm:EstIVQRgeneric} is applicable to various estimation methods and we believe it will be of interest even in settings not based on sparsity assumptions.

\section{Monte Carlo}\label{SecMC}

Here we provide a simulation study of the finite sample properties of the proposed estimators and confidence intervals. We compare their performance with the naive post selection estimator, which is defined by applying  the logistic regression performed on the model selected by the $\ell_1$-penalized logistic regression.

Our simulations are based on the model:
$$ \Ep[y\mid d, x] = \G( d\alpha_0 + x'\{c_y\nu_y\} ), \  \ \ \  d = x'\{c_d\nu_d\} + \tilde v,  $$
where  the coefficient vectors $\nu_y$ and $\nu_d$ are set to $$\begin{array}{l}\nu_{y} = (1,1/2,1/3,1/4,1/5,0,0,0,0,0,1,1/2,1/3,1/4,1/5,0,0,\ldots,0)',\\
\nu_{d} = (1,1/2,1/3,1/4,1/5,1/6,1/7,1/8,1/9,1/10,0,0,\ldots,0)',\end{array}$$ $x = (1,z')'$ consists of an intercept and covariates $z \sim N(0,\Theta)$, and the error $\tilde v$ is i.i.d. as $N(0,1)$.  The dimension $p$ of the covariates $x$ is $250$, and the sample size $n$ is $200$. In this setting  the coefficients feature a declining pattern, with the smallest non-zero coefficients being hard to differentiate from zero for the given sample size.  Therefore, we expect that the $\ell_1$-based model selectors will be making selection mistakes on variables with the smaller coefficients. (Additional simulations are provided in the Supplementary Material where we also consider an approximately sparse model for which all 250 coefficients are non-zero.  Those experiments demonstrate that the results are robust with respect to moderate deviations away from exactly sparse models.)

The regressors are correlated with covariance $\Theta_{ij} = \rho^{|i-j|}$ and $\rho = 0.5$. The coefficient $c_d$ is used to control the $R^2$, denoted $R^2_d$, in the equation relating main regressor to the controls, and $c_y$ is used to control the $R^2$, denoted $R^d_y$, for the regression equation: $\tilde y - d\alpha_0 = x'\{c_y\nu_y\}+\epsilon$, where $\epsilon$ is logistic noise with unit variance. In the simulations below we will use different values of $\alpha_0$, $c_y$ and $c_d$, which induce different data-generating processes (dgps). For every replication, we draw new errors $v_{i}$'s and controls $x_i$'s. The regression functions $x'(c_y\nu_y)$ and $x'(c_d\nu_d)$ are sparse. As we vary the coefficients $c_y$ and $c_d$, we induce different amounts of ``signal" strength, making it easier or harder for the Lasso-type methods to detect the controls with non-zero coefficients.
%It should be emphasized, that unless $c_y$ and $c_d$ are very hi
 %typical separation from zero assumptions no longer hold for the relevant regressors.

In Figure \ref{Fig:SimFirst} we consider a dgp with $\alpha_0=0.2$ and $R^2_d = R^2_y=0.75$, induced by setting $c_d = 1$ and $c_y = 0.75$.  We performed 5000 Monte-Carlo simulations. Figure \ref{Fig:SimFirst} summarizes the performance and displays the distribution of the following estimators, which are centered by the true value $\alpha_0$ and studentized by their standard deviation:
\begin{itemize}
\item[1.]  Naive post selection estimator -- estimator of $\alpha_0$ based on logistic regression after the naive selection using $\ell_1$-penalized logistic regression;
 \item[2.] Optimal instrument estimator -- estimator of $\alpha_0$ based on the instrumental logistic regression with the optimal instrument, as defined in Table \ref{Table:Alg};
  \item[3.] Double selection estimator -- estimator of $\alpha_0$ based on the logistic  regression after double selection, as defined in Table \ref{Table:Alg2}.
 \end{itemize}
The optimal IV estimator and the double selection estimator have distribution approximately centered at the true value, with distribution agreeing closely with the standard normal distribution.  They have low biases, low  root mean squared errors, and confidence regions have rejection rates close to the nominal level of .05. This good performance is well aligned with our theoretical results that we have developed in the previous section.  In sharp contrast, the distribution of the naive post selection estimator seems to deviate substantially from the normal distribution. This estimator exhibits large bias and high root mean squared error compared to the former procedures. This occurs because in this dgp, perfect selection is not achieved, and the  resulting ``moderate" selection mistakes create a large omitted variable bias.    Thus, if we use naive post selection estimator with the standard normal distribution for constructing confidence intervals or performing hypothesis testing, we shall end up with rather misleading inference.  This poor performance is well aligned with theoretical predictions of \cite{leeb:potscher:hodges,LeebPotscher2006,LeebPotscher2005} in the context of linear models.

%  \vspace{1cm}
%\hspace{-1cm}
%{\footnotesize \begin{tabular}{l|ccc}
%\hline
% estimator        &  post-naive  &  optimal & post-double  \\
%                  &  selection   &     IV       & selection \\
% \hline
%bias      &   .173  &  .038   &  .024 \\
%variance  &   .041  &  .036  &  .039  \\
%rmse      &   .267  &  .193  &  .199   \\
% rp(0.05) &   .350  & 0.43   & 0.051 \\
%\hline
%\end{tabular}}

\begin{figure}
  \begin{minipage}[t]{0.49\linewidth}
  \vspace{0pt}
   \includegraphics[width=\textwidth]{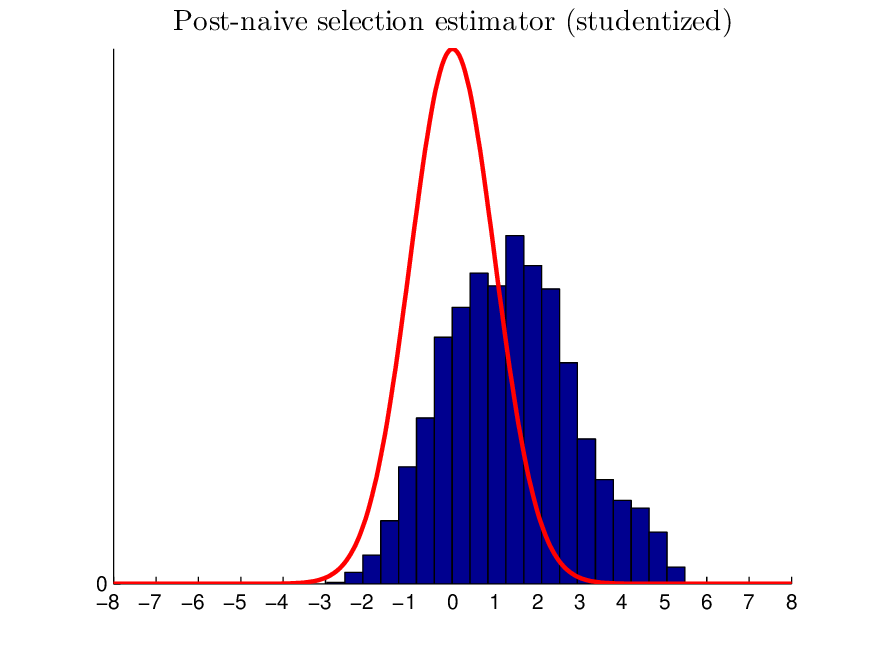}\\
  \end{minipage}
  \begin{minipage}[t]{0.49\linewidth}
  \vspace{1.75cm}

\hspace{-0.75cm}
{\footnotesize \begin{tabular}{l|cccc}
\hline
 estimator        &  bias &  variance & rmse & rp(0.05)\\
 \hline
naive post selection    &  .173  &  .041  &  .267 & .350 \\
optimal IV              &  .038  &  .036  &  .193 & .043 \\
double selection   &  .024  &  .039  &  .199 & .051  \\
\hline
\end{tabular}}
\end{minipage}
\includegraphics[width=0.49\textwidth]{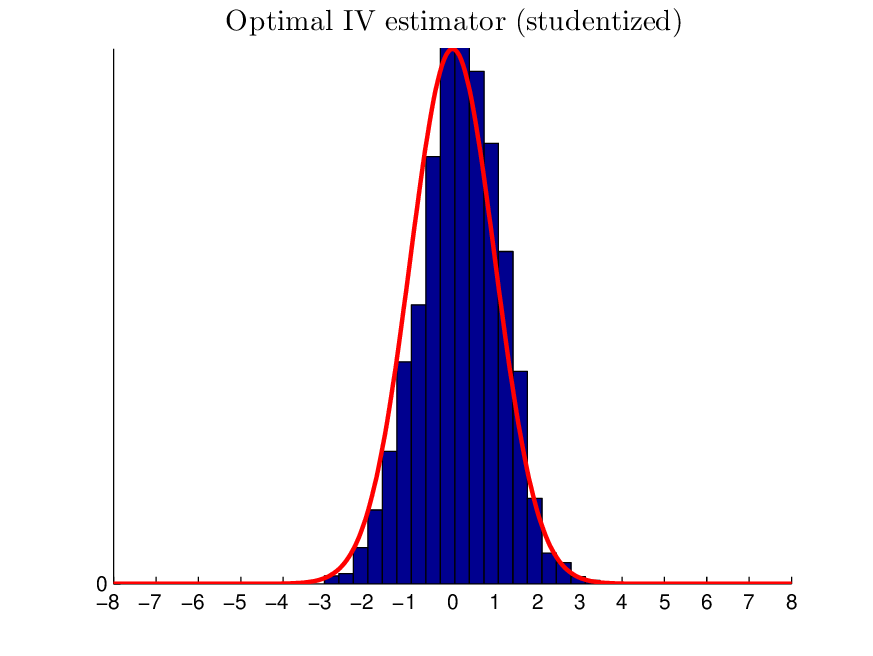}
\includegraphics[width=0.49\textwidth]{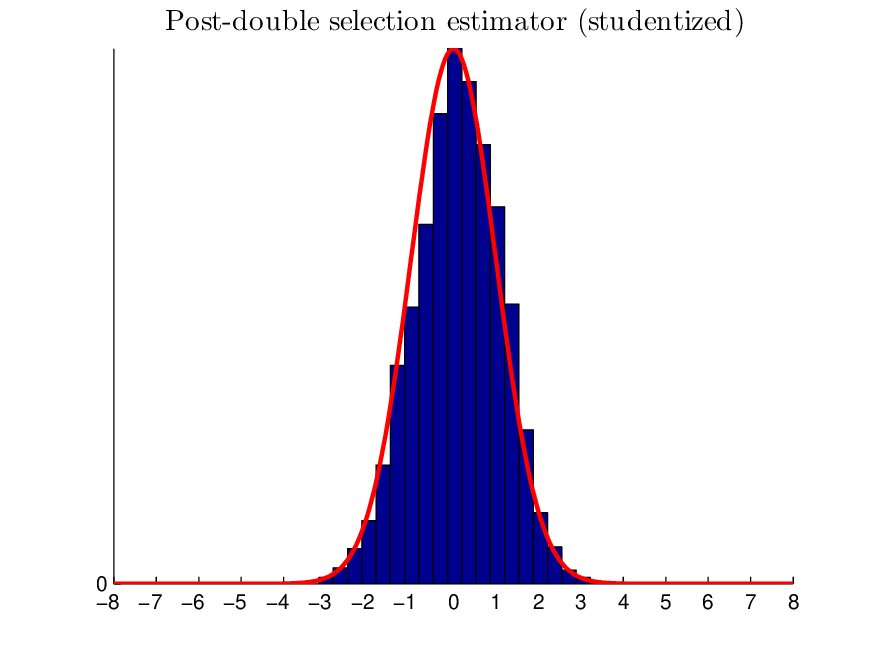}
\caption{\footnotesize The top right panel display bias, variance, RMSE, and rejection frequency for a $.05$-level test. The plots display the distribution of the naive post selection estimator (top left panel) and the two proposed estimators: optimal instruments  (bottom left panel) and double selection (bottom right). %The rejection probabilities at $5\%$ level were: post naive selection $35.0\%$; optimal instrument $4.3\%$; and double selection $5.1\%$.  The results are based on 5000 replications.
}\label{Fig:SimFirst}
\end{figure}

%\begin{table}
%\begin{tabular}{l|cccc}
% estimator        &  bias &  variance & rmse & rp(0.05)\\
% \hline
%post-naive selection    &  0.173  &  0.041  &  0.267 & 0.350 \\
%optimal IV              &  0.038  &  0.036  &  0.193 & 0.043 \\
%post-double selection   &  0.024  &  0.039  &  0.199 & 0.051  \\
%\hline
%\end{tabular}
%\end{table}

%\begin{figure}[ht]
%\includegraphics[width=0.49\textwidth]{HistogramNaive.eps}\\
%\includegraphics[width=0.49\textwidth]{HistogramOptimalInstrument.eps}
%\includegraphics[width=0.49\textwidth]{HistogramDoubleSelection.eps}
%\caption{\footnotesize The figures display the distribution of the naive post-model selection estimator (top panel) and the two proposed estimators optimal instruments (bottom left panel) and double selection (bottom right). The rejection probabilities at $5\%$ level were: post naive selection $35.0\%$; optimal instrument $4.3\%$; and double selection $5.1\%$.  The results are based on 5000 replications.}\label{Fig:SimFirst}
%\end{figure}

%\begin{figure}[h!]
%\includegraphics[width=\textwidth]{All_rp_sd.eps}
%\caption{\footnotesize The figures display the distribution of the standard deviation and rp(0.05) of the standard post-model selection estimator and the proposed confidence regions based on $\CR_D$ and $\CR_I$ (the standard deviation of the latter is measured with respect to its center). There are a total of 500 different designs. The results are based on 500 replications for each design.}\label{Fig:Distribution_RP_SD}
%\end{figure}

We now examine the performance more systematically by varying
\begin{equation}\label{model space}
(R^2_d, R^d_y) \in \{0, .1, .2, .3, .4, .5, .6, .7, .8, .9\}^2 \text { and } \alpha_0 \in \{0, .25, .5\}.
 \end{equation} This gives us 300 different dgps. For each dgp we performed 1000 Monte-Carlo simulations. Figures \ref{Fig:alpha00}-\ref{Fig:alpha05} display the rejection frequencies of confidence regions with (nominal) significance level of .05 and Figure \ref{Fig:RMSE}  displays the root mean squared errors
of the estimators of $\alpha_0$.  The goal of this exercise is to  verify numerically how good the uniformity claims derived in Corollaries \ref{cor:Uniformity1} and \ref{cor:Uniformity2} are, and also confirm that the previous conclusions continue to hold across a wide set of dgp.

 In Figures \ref{Fig:alpha00}-\ref{Fig:alpha05} we consider the rejection (non-coverage) frequencies of confidence regions based on: naive post selection logistic estimator\footnote{This region is given by $\{ |\alpha-\widetilde\alpha|\leq \{\En[\hat w_i(d_i,x_{i\supp(\widetilde\beta)}')'(d_i,x_{i\supp(\widetilde\beta)}')]\}^{-1/2}_{11}\Phi^{-1}(1-\xi/2)/\sqrt{n}\}$ where $(\widetilde\alpha,\widetilde\beta)$ is the naive post selection logistic estimator.}, optimal IV ($\CR_D$ and $\CR_I$) and double selection ($\CR_{DS}$). These figures illustrate the uniformity properties of the confidence regions based on the discussed estimators. The ideal figure would be a flat surface with the rejection frequency of the true value equal to the nominal level of $.05$.    The confidence regions based on the naive post selection perform very poorly, and deviate strongly away from the ideal level of $.05$ throughout large parts of the model space (induced by (\ref{model space})). In contrast, the confidence regions based on optimal IV and double selection seem to be substantially closer to the ideal level, which is in line with our theoretical results in Section \ref{Sec:Main}. The double selection estimator seems to outperform the estimator based on the explicit construction of the optimal instrument (e.g., the rejection rates and the RMSE for the case with $\alpha=.5$, where optimal instrument procedure tends to perform noticeably worse.) %\footnote{In the next section we discuss an interpretation of the double selection procedure as an iterated version of the optimal IV procedure, which might provide some intuition for its better finite-sample performance.  Note that the two procedures are first-order asymptotically equivalent.}
   Thus, based on the theoretical results and on the Monte-Carlo results, we recommend the use of the double selection estimator over the optimal IV estimator and, certainly, over the naive post selection estimator.

%Finally, in Figures \ref{??}-\ref{??} we display the RMSE of the three estimators.  The performance of the post-naive selection estimator is very uneven -- it is not (and can not be) uniformly dominated by our proposed estimators, but it performs rather poorly through most of data-generating processes.  This has a theoretical explanation -- the estimator is ``super-efficient" relative to our model
%(without the separation condition imposed), and this buys very good performance for some data-generating processes, but it must %and does have a rather steep price for other data-generating processes.

%\begin{figure}[ht]
%\includegraphics[width=\textwidth]{RPalpha00new.eps}
%\caption{\footnotesize The figures display the rp(0.05) of the naive post-model selection estimator and the proposed confidence regions based on optimal instrument ($\CR_D$ and $\CR_I$) and double selection ($\CR_{DS}$). There are a total of 100 different designs with $\alpha_0=0$. The results are based on 1000 replications for each design.}\label{Fig:alpha00}
%\end{figure}
%
%\begin{figure}[ht]
%\includegraphics[width=\textwidth]{RPalpha01new.eps}
%\caption{\footnotesize The figures display the rp(0.05) of the naive post-model selection estimator and the proposed confidence regions based on optimal instrument ($\CR_D$ and $\CR_I$) and double selection ($\CR_{DS}$). There are a total of 100 different designs with $\alpha_0=0.1$. The results are based on 1000 replications for each design.
%}\label{Fig:alpha01}
%\end{figure}

\begin{figure}[ht]
\includegraphics[width=0.49\textwidth]{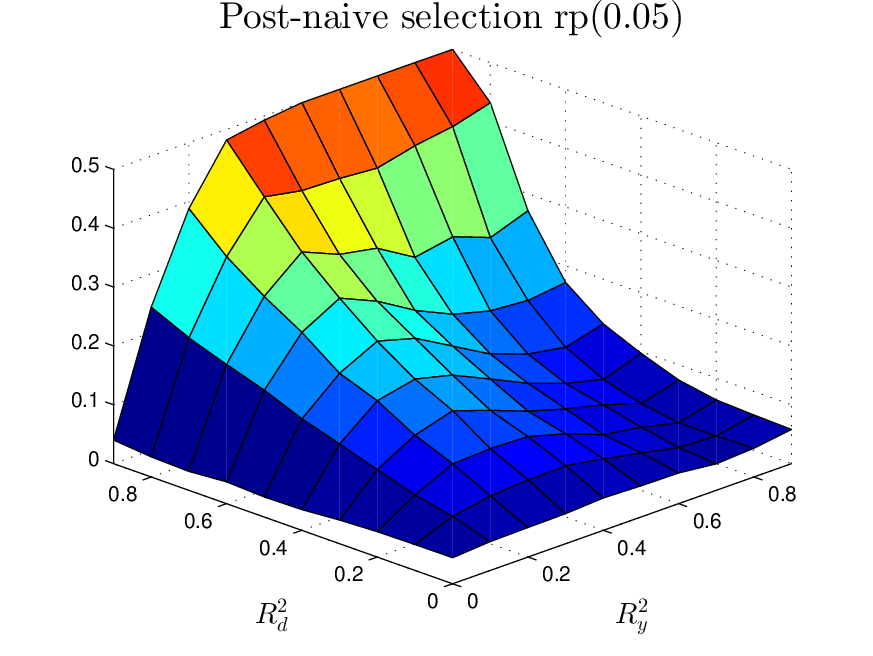}
\includegraphics[width=0.49\textwidth]{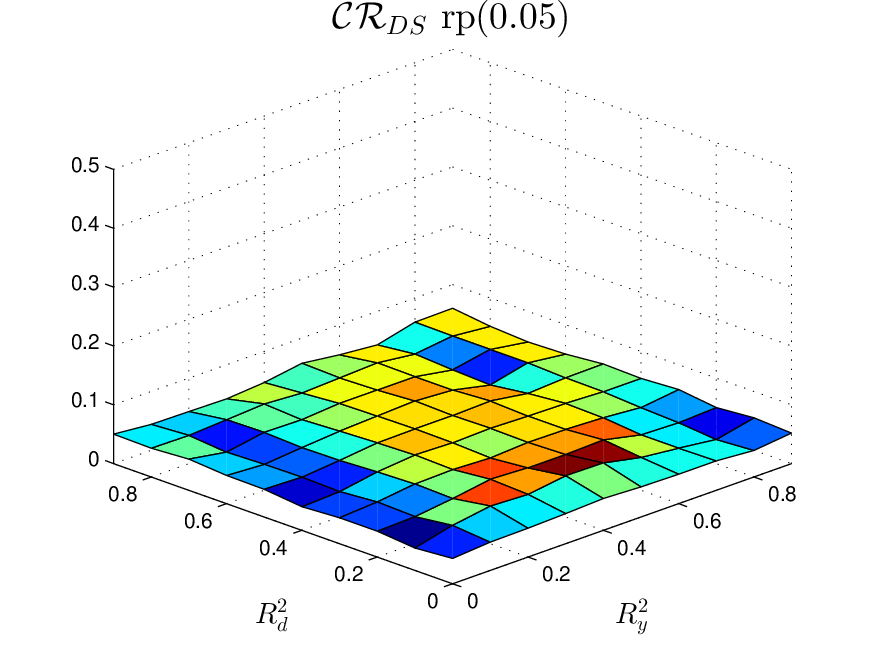}\\
\includegraphics[width=0.49\textwidth]{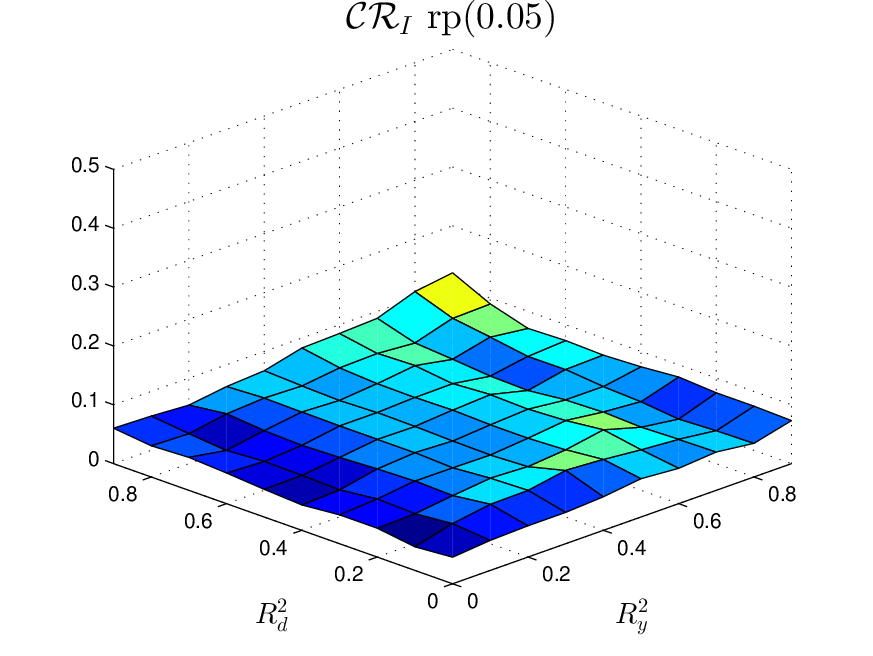}
\includegraphics[width=0.49\textwidth]{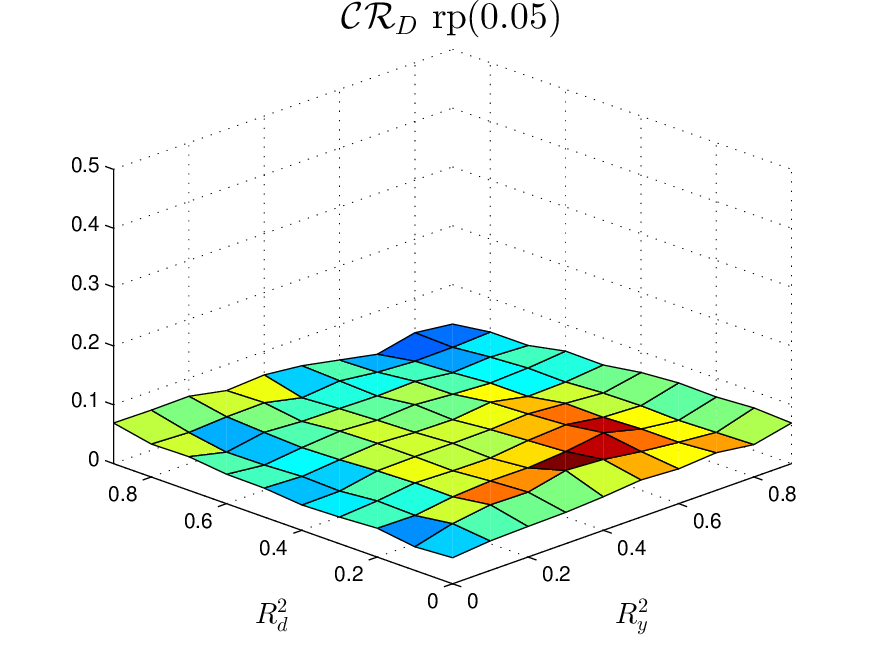}
\caption{\footnotesize The plots display the rejection frequencies at .05 level (rp(0.05)) of the confidence regions based on naive post selection, optimal instrument ($\CR_D$ and $\CR_I$) and double selection ($CR_{DS}$). There are a total of 100 different designs with $\alpha_0=0$. The results are based on 1000 replications for each design.}\label{Fig:alpha00}
\end{figure}

\begin{figure}[ht]
\includegraphics[width=0.49\textwidth]{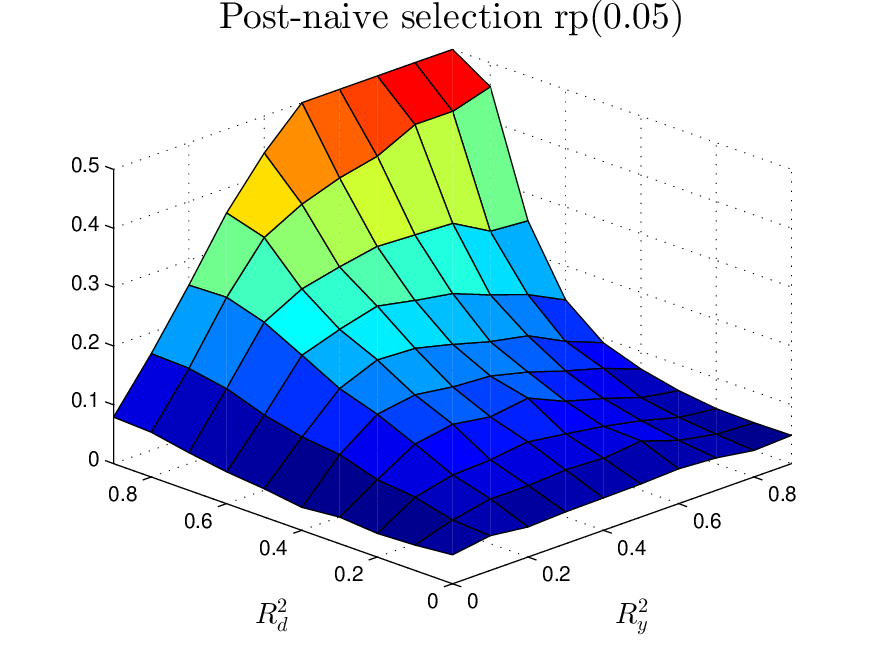}
\includegraphics[width=0.49\textwidth]{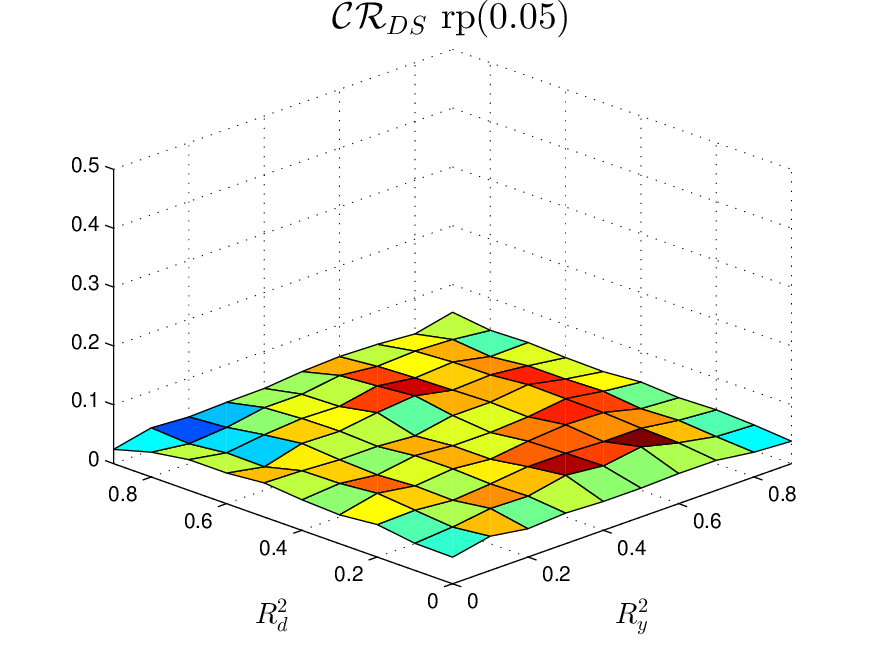}\\
\includegraphics[width=0.49\textwidth]{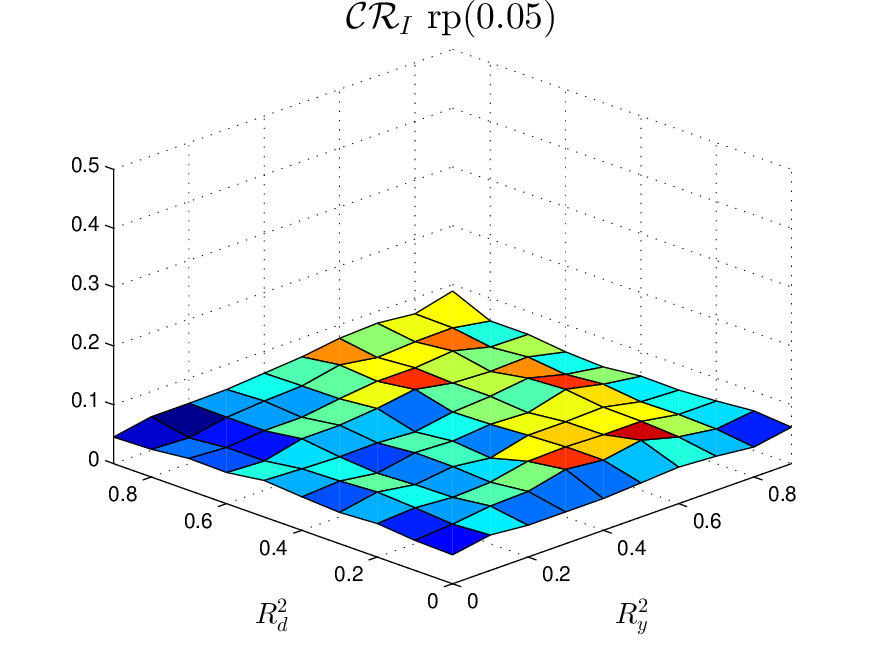}
\includegraphics[width=0.49\textwidth]{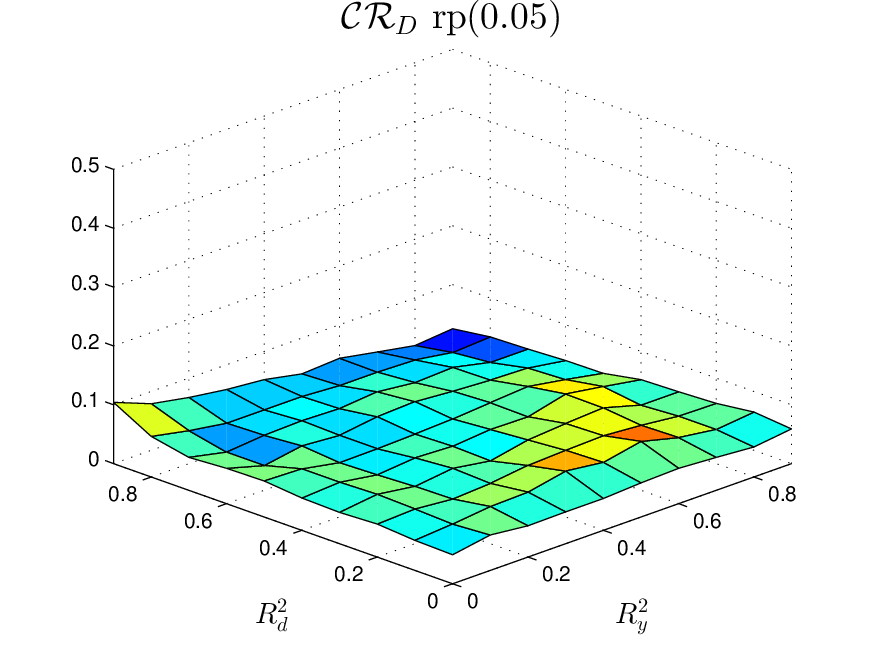}
\caption{\footnotesize The plots display the rejection frequencies at .05 level (rp(0.05)) of the confidence regions based on naive post selection, optimal instrument ($\CR_D$ and $\CR_I$) and double selection ($CR_{DS}$). There are a total of 100 different designs with $\alpha_0=0.25$. The results are based on 1000 replications for each design.}\label{Fig:alpha025}
\end{figure}

\begin{figure}[ht]
\includegraphics[width=0.49\textwidth]{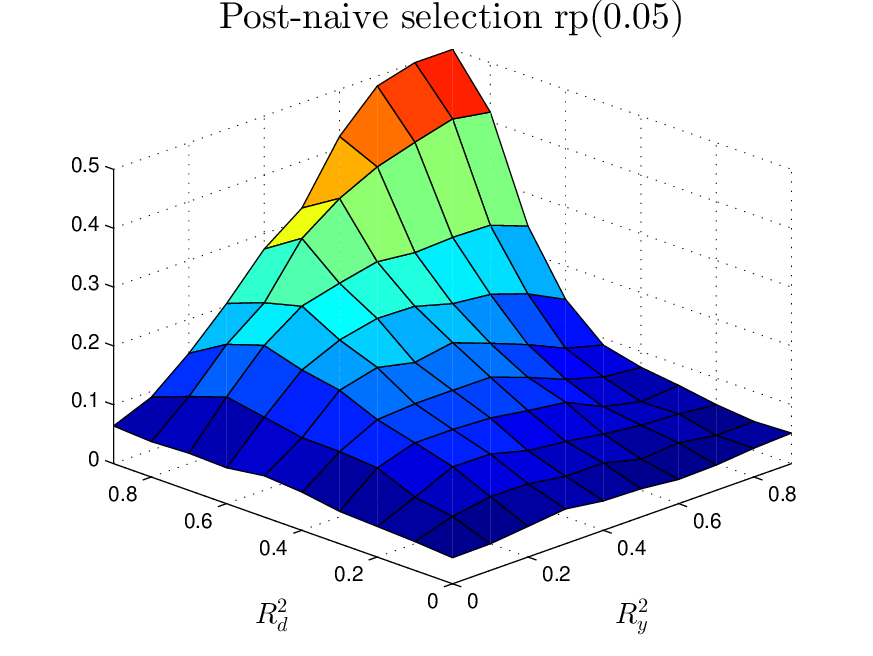}
\includegraphics[width=0.49\textwidth]{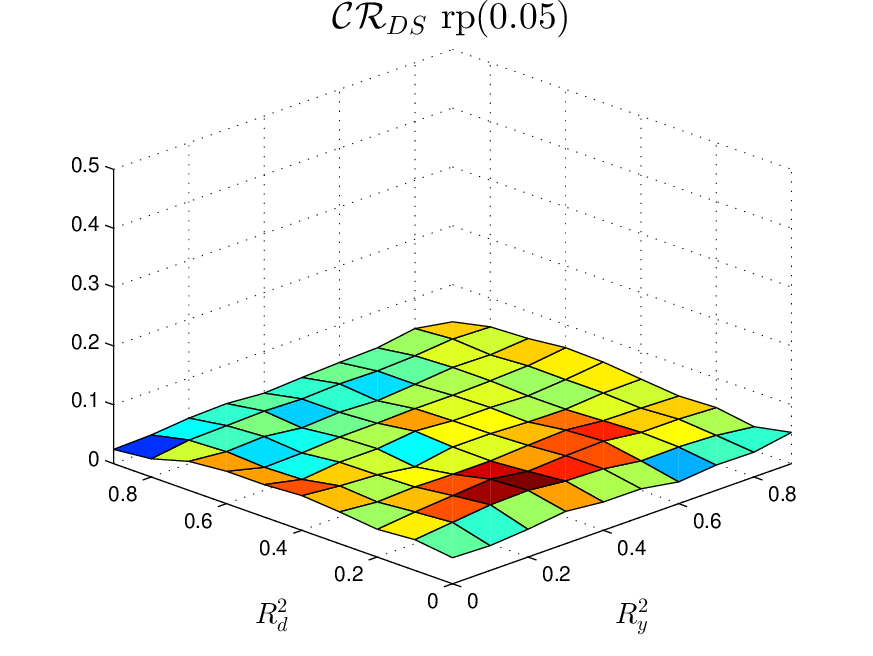}\\
\includegraphics[width=0.49\textwidth]{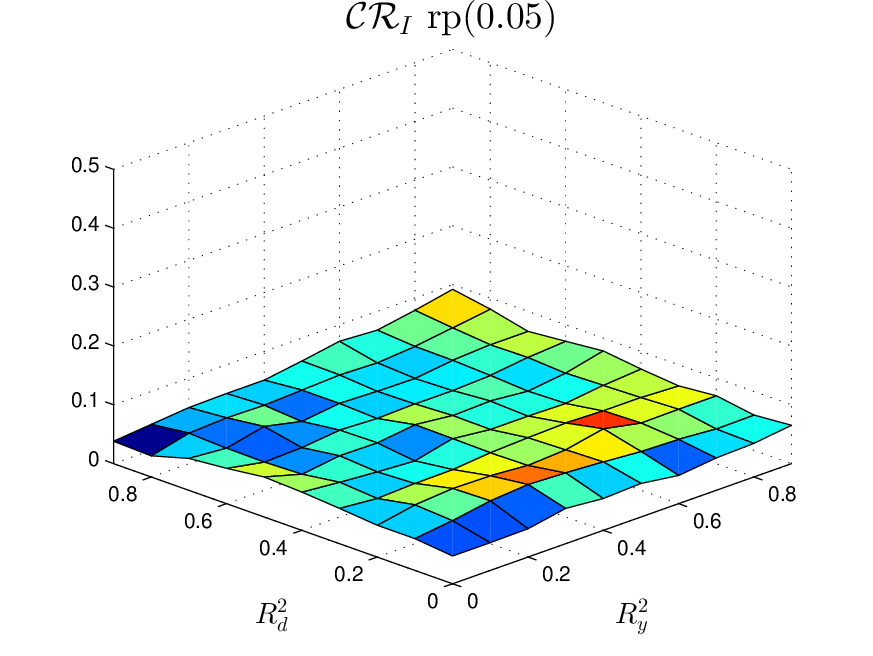}
\includegraphics[width=0.49\textwidth]{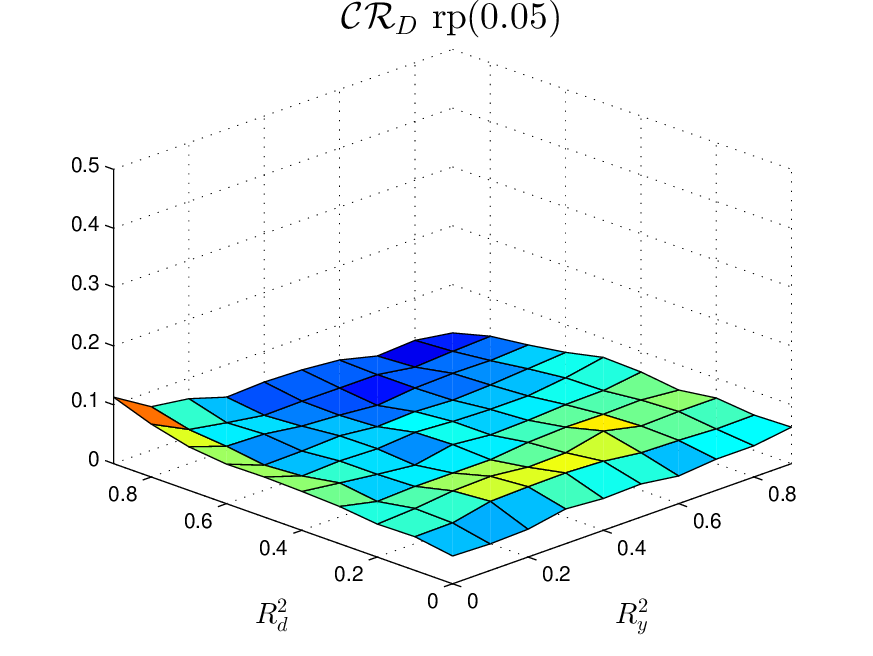}
\caption{\footnotesize The plots display the rejection frequencies at .05 level (rp(0.05)) of the confidence regions based on naive post selection, optimal instrument ($\CR_D$ and $\CR_I$) and double selection ($\CR_{DS}$). There are a total of 100 different designs with $\alpha_0=0.5$. The results are based on 1000 replications for each design.}\label{Fig:alpha05}
\end{figure}

%\begin{figure}[ht]
%\includegraphics[width=0.49\textwidth]{RMSEalpha00naive.eps}
%\includegraphics[width=0.49\textwidth]{RMSEalpha01naive.eps}\\
%\includegraphics[width=0.49\textwidth]{RMSEalpha00IV.eps}
%\includegraphics[width=0.49\textwidth]{RMSEalpha01IV.eps}\\
%\includegraphics[width=0.49\textwidth]{RMSEalpha00DS.eps}
%\includegraphics[width=0.49\textwidth]{RMSEalpha01DS.eps}
%\caption{\footnotesize The figures display the RMSE of the naive post-model selection estimator and the proposed confidence regions based on optimal instrument ($\CR_D$ and $\CR_I$) and double selection ($\CR_{DS}$). The left column refers to $\alpha_0=0$ and the right column refers to $\alpha_0=0.1$. There are a total of 100 different designs for each value of $\alpha_0$. The results are based on 1000 replications for each design.}\label{Fig:RMSE00}
%\end{figure}

\begin{figure}[ht]
\includegraphics[width=0.3\textwidth]{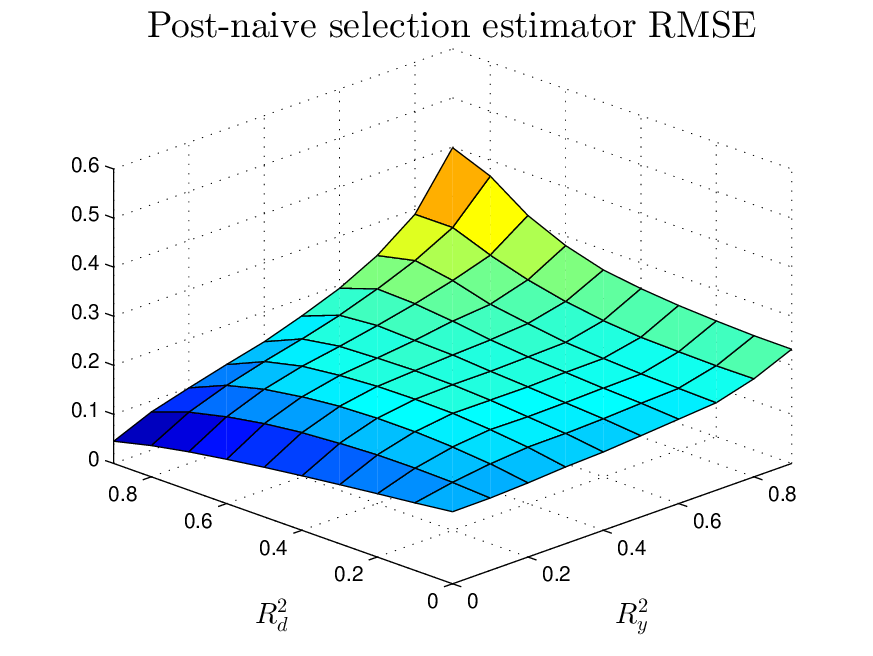}
\includegraphics[width=0.3\textwidth]{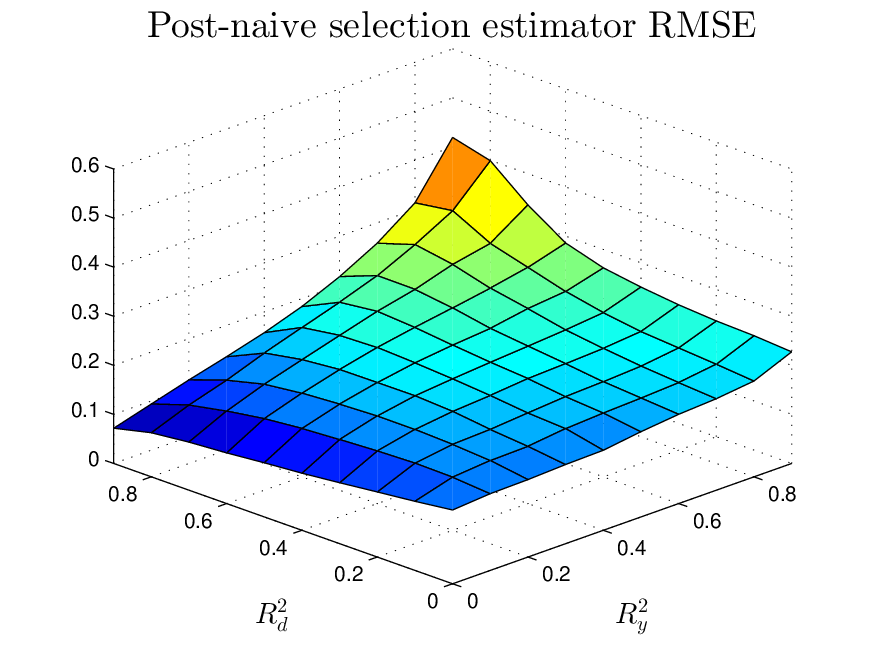}
\includegraphics[width=0.3\textwidth]{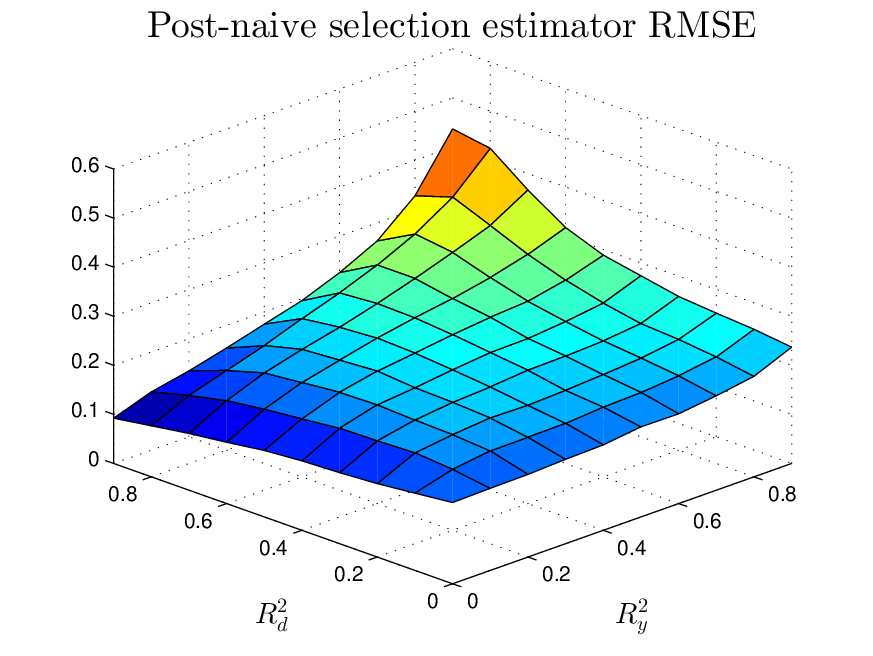}\\
\includegraphics[width=0.3\textwidth]{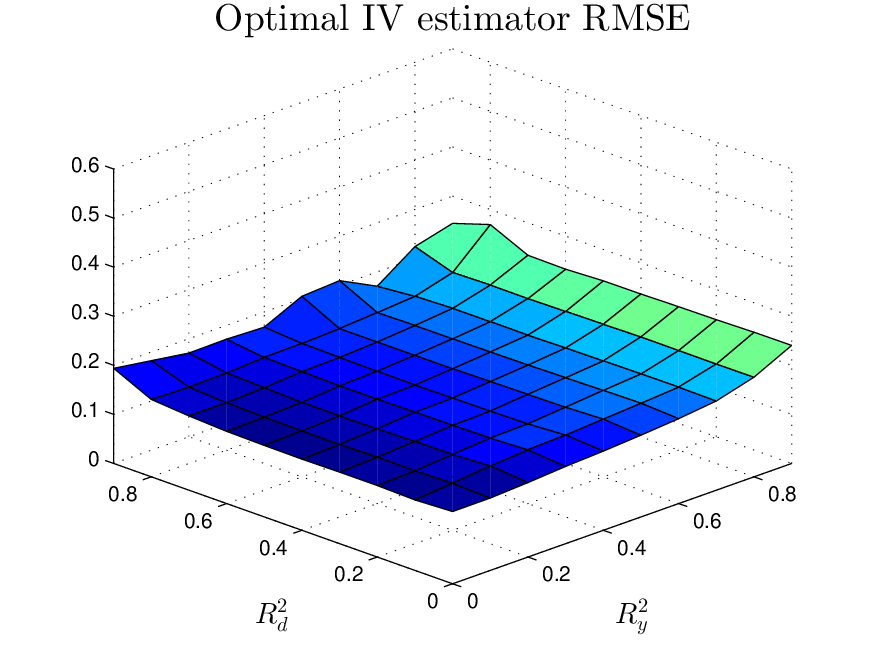}
\includegraphics[width=0.3\textwidth]{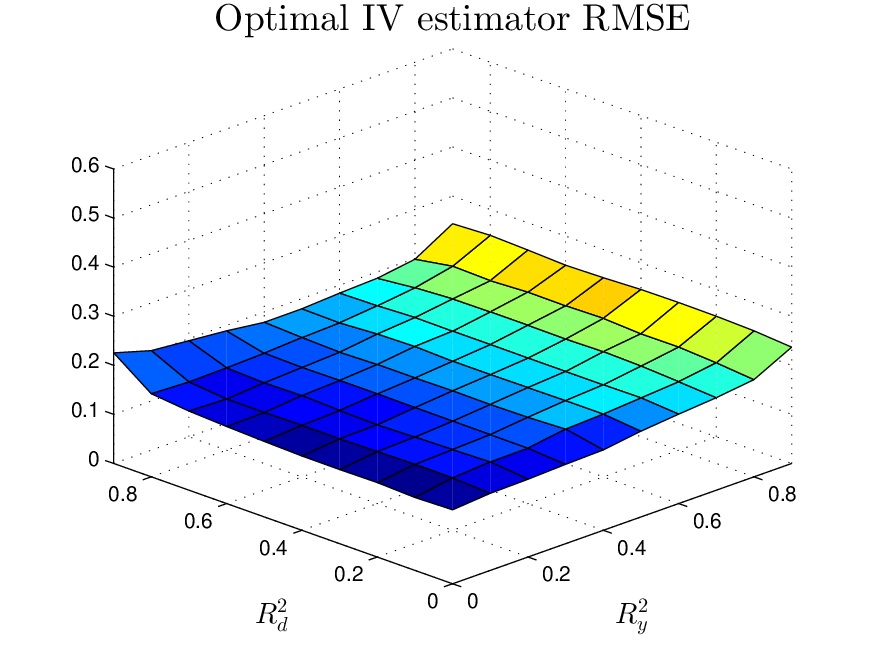}
\includegraphics[width=0.3\textwidth]{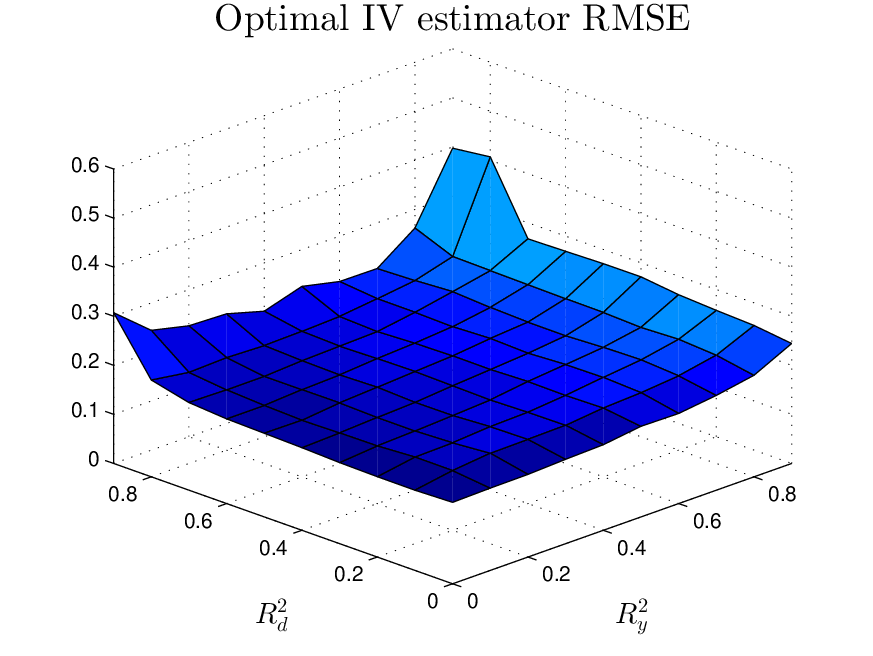}\\
\includegraphics[width=0.3\textwidth]{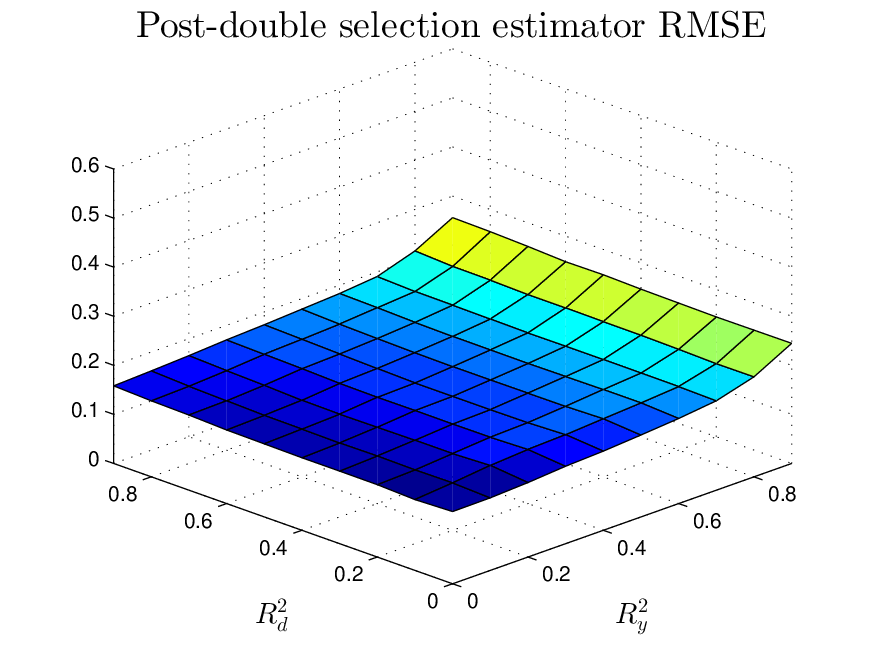}
\includegraphics[width=0.3\textwidth]{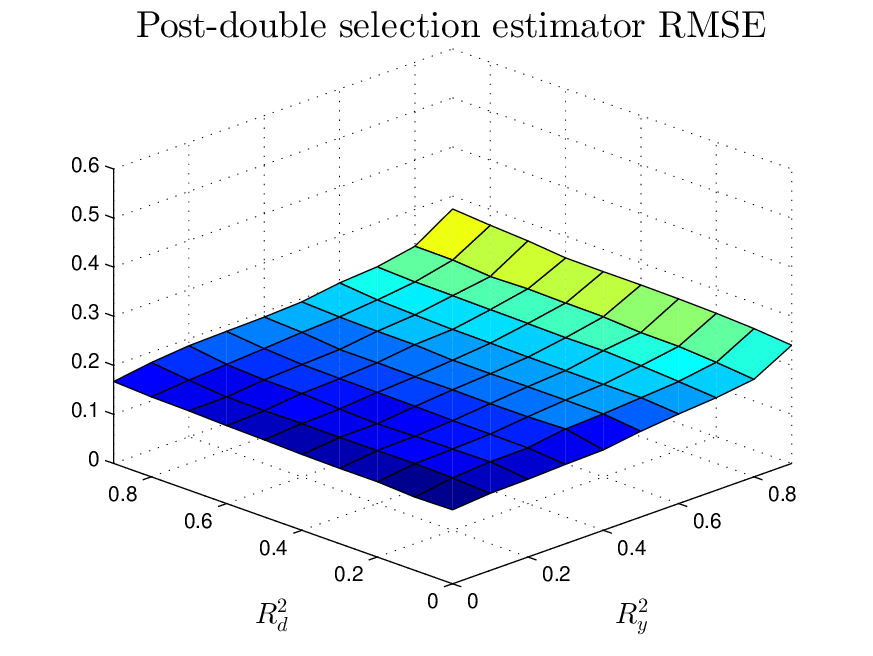}
\includegraphics[width=0.3\textwidth]{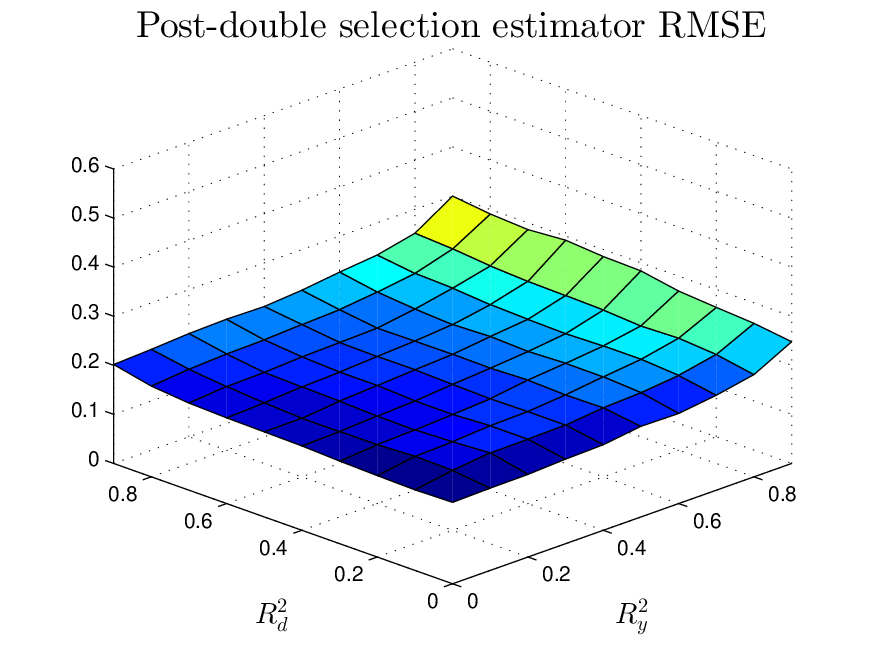}
\caption{\footnotesize The plots display the RMSE of the naive post selection estimator, optimal instrument estimator and double selection estimator. The left column refers to $\alpha_0=0$, the middle column refers to $\alpha_0=0.25$, and the right column refers to $\alpha_0=0.5$. There are a total of 100 different designs for each value of $\alpha_0$. The results are based on 1000 replications for each design.}\label{Fig:RMSE}
\end{figure}

%\section{Empirical Application}

%\cleardoublepage

\section{Discussion}\label{Sec:Disc}

\subsection{Relation between Double Selection and Optimal Instrument}\label{Sec:Connections}

In this section, we provide a more formal connection between the two proposed methods. It turns out that the construction of the double selection estimator implicitly approximates the optimal instrument $z_{0i}=v_i/\sqrt{w_i}$. This occurs because the model selection procedure in Step 2 associated with (\ref{Eq:Decomposition}) allows the estimator to achieve uniformity properties. To see that, using the notation in Table \ref{Table:Alg} where $\hat\beta$, $\hat\theta$ and $\widetilde \theta$ are defined, let $\hat T^* = \supp(\hat\beta)\cup \supp(\hat\theta)$ denote the variables selected in Step 1 and 2. By the first order conditions of the double selection logistic regression of Step 3 in Table \ref{Table:Alg2} we have
$$\En[ \{y_i-\G(d_i\check\alpha  + x_i'\check\beta)\} (d_i, \ x_{i\hat T^*}')']=0$$
which creates an orthogonal relation to any linear combination of $(d_i, \ x_{i\hat T^*}')'$. In particular, by taking the linear combination  $(d_i, \ x_{i\hat T^*}')(1,-\widetilde\theta')'=d_i - x_i'\widetilde\theta =\hat z_i$, we have
$$\En[ \{y_i-\G( d_i\check\alpha + x_i'\check\beta)\} \hat z_i ]=0.$$
Therefore the double selection estimator $\check\alpha$ minimizes
$$\widetilde L_n(\alpha) =  \frac{\| \En[\{ y_i - \G(d_i\alpha + x_i'\check\beta) \} \hat z_i ] \|^2}{\En[\{ y_i - \G(d_i\alpha  + x_i'\check\beta) \}^2 \hat z_i^2 ]}, $$
where $\hat z_i$ is the instrument of the optimal instrument estimator which was {\it implicitly} created. Thus, the double selection estimator can be seen as an iterated version of the method based on instruments where  $\widetilde \beta$ is replaced with $\check \beta$. Although their first order asymptotic properties coincide, in finite sample, the double selection method seems to obtain better estimates leading to a more robust performance.

%approach allows for the model selection in one step potentially mitigating model selection mistakes of the other step. A potential drawback of the double selection is to work with slightly bigger models (and thus potentially slightly more stringent design conditions).

\subsection{Relation to Neyman's $C(\alpha)$ test}\label{Sec:Neyman} Next we discuss connections between the proposed approach and Neyman's $C(\alpha)$ test  \cite{Neyman1959,Neyman1979} (here we draw on the discussion in \cite{BelloniChernozhukovKato2013a}). For the sake of exposition, we assume the instruments are known and i.i.d. observations. As stated in (\ref{Eq:Estimating}) and (\ref{Eq:Estimating2}) we rely on instruments satisfying the two equations:
$$
\Ep[\{y_i-\G(d_i\alpha_0+x_i'\beta_0)\} z_{0i}] =0 \ \ \mbox{and} \ \ \left.\frac{\partial}{\partial \beta} \Ep[\{y_i -\G(d_i\alpha_0+x_i'\beta) \} z_{0i}] \right|_{\beta=\beta_0} = \Ep[w_i z_{0i}x_i]=0.
$$
These  conditions allows us to construct regular, $\sqrt{n}$-consistent estimators of $\alpha_0$, despite the fact that nonregular, non $\sqrt{n}$-consistent estimator for $\beta_0$ are being used to cope with high-dimensionality. In particular, regularized or post model selection estimators can be used as estimators of $\beta_0$. Neyman's $C(\alpha)$ test was motivated by the same idea which motivates the use of the term ``Neymanization"  to describe such procedure.  Although there will be many instruments $z_{0i}$ that can achieve the property stated above, the choice $z_{0i}=v_i/\sigma_i$ proposed in Section \ref{Sec:Model} is optimal as it minimizes the asymptotic variance of the resulting estimators.

Generally, valid  (but not necessarily optimal) instruments can be constructed by generalizing the weighted equation (\ref{Eq:Decomposition}) to
 \begin{equation}\label{define m}
f_i d_i = f_i m_0(x_i) + \tilde v_i,   \  \  \Ep[f_i \tilde v_i x_i]  =0,
 \end{equation}
where $f_i=f(d_i,z_i)$ is a nonnegative weight, and setting the instrument as $z_{0i} := f_i\tilde v_i/w_i$. %For example $f_i =1$ as proposed in Comment \ref{Comment:ValidInstrument} or $f_i =\sqrt{w_i}$ as in (\ref{Eq:Decomposition}).
Because of the zero-mean condition in (\ref{define m}), and provided that $m_0 \in \mathcal{H}$, the function $m_0(x_i)$ in (\ref{define m}) is the solution of the following weighted least squares problem
\begin{equation}
\min_{h \in \mathcal{H}} \barEp\[ f_i^2 \{ d_i -  h(x_i)\}^2 \],
\end{equation}
where $\mathcal{H}$ denotes the set of measurable functions $h$ satisfying
$\Ep[f_i^2 h^2(x_i) ]< \infty$ for each $i$. In the current high-dimensional setting, it is assumed that  $m_0(x_i)$ can be written as
a sparse combination of the controls, namely $m_0(x_i)=x_i'\theta_0$ with $\|\theta_0\|_0 \leq s$,  so that
\begin{equation}\label{define_theta0_f}
f_i d_i = f_i x_i' \theta_0 + \tilde v_i,   \  \  \Ep[f_i \tilde v_i x_i]  =0.
\end{equation} This permits the use of Lasso or Post-Lasso to estimate $\theta_0$ which in turn can be used to construct an estimate of $z_{0i}$. Naturally, if the function $m_0$ satisfies different structured properties that could motivate different estimators (for example, we can use ridge estimators if the $m_0$ is ``dense" with respect to $x$).

Our technical results establish that, uniformly over $\{ \alpha :  \sqrt{n}|\alpha - \alpha_0| \leq C \}$,
\begin{equation}\label{adapt}
 \En[\{y_i -\G(d_i\alpha+x_i'\hat \beta)\} z_{0i}] - \En[\{y_i-\G(d_i\alpha+x_i'\beta_0)\} z_{0i}] =o_\Pr(n^{-1/2}),
\end{equation}
for the estimators $\hat\beta$ proposed in this work. This does not require $\hat \beta$ to converge at root-$n$ rate to $\beta_0$ (which generically is not achievable in the current setting) though we do impose the sparsity condition $s^2 \log^2 (p\vee n)/n \to 0$ to guarantee that  $\|\hat \beta- \beta \| = o_\Pr(n^{-1/4})$.  Equation (\ref{adapt}) implies that the empirical estimating equations behave as if $\beta_0$ was used instead of $\hat\beta$. Hence, for estimation,
we can use the instrumental logistic regression estimator, namely $\check \alpha$ as a minimizer of the statistic
$$
n L_n(\alpha) =  \|  \sqrt{n} \En[\{y_i - \G(d_i\alpha+x_i'\hat \beta)\} z_{0i}]\|^2 \ \ / \ \ \En[\{y_i - \G(d_i\alpha+x_i'\hat \beta)\}^2 z_{0i}^2].
$$
%In turn it is possible to use the statistic above for the construction of confidence regions because $\mathcal{L}_n(\alpha_0)  \rightsquigarrow \chi^2(1)$.

From (\ref{define_theta0_f}) we have that $\theta_0= \barEp[f^2_i x_i x_i' ]^{-}\barEp[f^2_i d_i x_i ]$,   where  $A^-$ denotes a generalized inverse of $A$. Letting $\hat\varepsilon_i(\alpha) = y_i - \G(d_i\alpha+x_i'\hat \beta)$ and
$$
z_{0i} = f_i\tilde v_i/w_i = (f_i^2/w_i) d_i-  (f_i^2/w_i) x_i'  \barEp[f^2_i x_i x_i' ]^{-}\barEp[f^2_i d_i x_i ],$$
 $n L_n(\alpha)$ can be rewritten as a (perhaps) familiar version of Neyman's $C(\alpha)$ statistic
$$
n L_n(\alpha) =  \frac{\|  \sqrt{n} \{ \En[ \hat\varepsilon_i(\alpha) (f_i^2/w_i) d_i] - \En[ \hat\varepsilon_i(\alpha)(f_i^2/w_i) x_i' ] \barEp[f^2_i x_i x_i' ]^{-}\barEp[f^2_i d_i x_i ]  \} \|^2}{\En[\hat\varepsilon_i^2(\alpha) z_{0i}^2]}.
$$
Thus, our IV estimator minimizes a Neyman's $C(\alpha)$ statistic for testing point hypotheses about $\alpha$.  Hence  our construction builds on the classical ideas of Neyman for dealing with (hard-to-estimate) nuisance parameters.

An estimator $\check \alpha$  that minimizes the criterion $n L_n$ up to a $o_\Pr(1)$ term satisfies
$$
\tilde \Sigma_n^{-1} \sqrt{n} (\check \alpha - \alpha_0) \rightsquigarrow  N(0,1),  \  \ \ \tilde\Sigma^2_n = \barEp[ w_i d_i z_{0i} ]^{-2}   \barEp[ \sigma_i^2z_{0i}^2].
$$
It is not difficult to check that using  $f_i = w_i/\sigma_i$ leads to the smallest possible value  $\barEp[ v_i^2]^{-1}$ of $\tilde\Sigma^2_n$.  Therefore, $z_{0i}=v_i/\sigma_i$ is the optimal instrument among all instruments that can be derived by the preceding approach.  Using the optimal instrument translates into more precise estimators, smaller confidence regions, and better power for testing  based on either $\check \alpha$ or $n L_n$.

\subsection{Relation to Minimax Efficiency for Logistic Model}\label{Sec:minimax}

Next we consider a connection to the (local) minimax efficiency analysis from the semiparametric literature, where we follow the discussion in \cite{BelloniChernozhukovKato2013a}. Our model is a special case of a partially linear logistic model, and \cite{kosorok:book} derives an efficient score function for the latter:
$$
S_i= \{y_i - \G(d_i\alpha_0+x_i'\beta_0)\} \{d_i -  m_0^*(x_i)\},
$$
where$$
m^*_0(x_i) =  \frac{\Ep[w_i d_i|x_i]}{\Ep[w_i |x_i]}.
$$
We note that $m_0^*(x_i)$ is $m_0(x_i)$ in (\ref{define m}) induced by the weight $f_i =\sqrt{w_i}$. Thus, the efficient score function can be reexpressed as:
$$
S_i = \{y_i - \G(d_i\alpha_0+x_i'\beta_0)\} v_i/\sqrt{w_i},
$$
where $v_i$ is defined via (\ref{define m}).  Using this score  leads to the same estimating equations as those constructed above using Neymanization (with an optimal instrument).  It follows
that the estimator based on the instrument $z_{0i}=v_i/\sqrt{w_i}$ is efficient in the local minimax sense (see Theorem
18.4 in \cite{kosorok:book}), and inference about $\alpha_0$
based on this estimator provides best minimax power against local alternatives (see Theorem  18.12 in \cite{kosorok:book}).

The preceding claim is formal provided that the least favorable submodels
are permitted as deviations within the overall set of potential models $\mathcal{Q}_n$ (defined similarly to Corollary \ref{cor:Uniformity1}).  Specifically, given a law $Q_n$, there should be a suitable neighborhood $\mathcal{Q}_n^{\delta}$ of $Q_n$ such that
$Q_n \in \mathcal{Q}_n^{\delta} \subset \mathcal{Q}_n$.   For that, we assume $m^*_0(x_i)=x_i'\theta_0$ and consider
a collection of models indexed by $t  = (t_1, t_2)$ satisfying:
\begin{eqnarray}
\Ep[y_i\mid d_i,x_i] & = & \G( d_i \{\alpha_0+ t_1\} + x_i'\{\beta_0 + t_2 \theta_0\}),    \  \  \  \|t\| \leq \delta, \\
\sqrt{w_i} d_i & =&   \sqrt{w_i} x_i'\theta_0 + v_i,  \ \ \Ep [\sqrt{w_i} v_i|x_i] =0,
\end{eqnarray}
where $\|\beta_0\|_0 + \| \theta_0\|_0 \leq s$ and Condition L as in Section \ref{Sec:Main} hold.  By construction, the model associated with $t=0$ generates precisely the model $Q_n$. As $t$ varies within a $\delta$-ball, we generate the set of models $\mathcal{Q}_n^{\delta}$ that contains the least favorable deviations, and which still belong to $\mathcal{Q}_n$. As shown in \cite{kosorok:book}, $S_i$ is the efficient score for such parametric submodel so we cannot have a better regular estimator than the estimator whose influence function is $\Sigma_n S_i$.  Because the set of models $\mathcal{Q}_n$ contains $ \mathcal{Q}_n^{\delta}$, all the formal conclusions about (local minimax) optimality of the proposed estimators hold from theorems cited above (using subsequence arguments to handle models changing with $n$). %Our estimators are regular, since under $\mathcal{Q}_n^t$ with $t = ( O(1/\sqrt{n}), O(1))$, their first order asymptotics do not change, as a consequence of Theorems in Section \ref{Sec:Main}.

\appendix

\section{Proofs of Theorems}\label{Sec:ProofMain}%\section{Instrumental GLM Regression with Nuisance Parameters}

\begin{proof}[Proof of Theorem \ref{theorem:inferenceAlg1}]
We will verify Condition IR and the result follows by Theorem \ref{Thm:EstIVQRgeneric}. We will use the (optimal) instrument $z_{0i}=v_i/\sigma_i$ and recall that in the case of a logistic link $\sigma_i^2 = w_i$. The link function $G$ and its derivatives $|G'|$ and $|G''|$ are uniformly bounded by $1$. The condition of $\Ep[w_iz_{0i}x_i]=0$ holds by (\ref{Eq:Decomposition2}). Since $z_{0i}=d_i-x_i'\theta_0$ and $\Ep[w_iz_{0i}x_i]=0$ we have $\barEp[w_id_iz_{0i}]=\barEp[w_iz_{0i}^2]=\barEp[w_i(d_i-x_i'\theta_0)^2]\geq \tilde c\|(1,\theta_0')\|^2$. The fourth moment condition implies that $\barEp[z_{0i}^2d_i^2]\leq \{\barEp[z_{0i}^4]\barEp[d_i^4]\}^{1/2}\leq C$.
Similarly, $\barEp[\sigma_i^3z_{0i}^3] \leq \barEp[z_{0i}^3]\leq \{\barEp[z_{0i}^4]\}^{3/4}$. Thus Conditions IR(i) and IR(ii) hold.

For $\tilde x_i = (d_i,x_i')'$, $i=1,\ldots,n$, we denote the minimum and maximum $m$-sparse empirical eigenvalues as $$ \semin{m} := \min_{1\leq \|\delta\|_0 \leq m} \frac{\|\tilde x_i'\delta\|_{2,n}^2}{\|\delta\|^2} \ \ \ \mbox{and} \ \ \ \semax{m} := \max_{1\leq \|\delta\|_0 \leq m} \frac{\|\tilde x_i'\delta\|_{2,n}^2}{\|\delta\|^2}.$$
Under the condition $K_2^2 s \log(p\vee n) \log^3 n \leq \delta_n n$, by Lemma \ref{thm:RV34} we have that with probability $1-o(1)$ the sparse eigenvalues of order $k=s\ell_n$ are bounded away from zero and from above by a constant for some $\ell_n\to \infty$. Under the condition that $\min_{i\leq n} w_i \geq c>0$ stated in Condition L(ii), we have that $\kappa_\cc$ defined in (\ref{Def:RestrEig}) is bounded away from zero with probability $1-o(1)$ for $n$ sufficiently large, see \cite{BickelRitovTsybakov2009}.

Step 1 relies on Post-Lasso-Logistic. To apply Lemma \ref{Lemma:LassoLogisticRate} %and\ref{Lemma:LassoLogisticSparsity}
to obtain rates and sparsity bounds, we first verify the side condition $q_{\Delta_\cc} > 3(1+\frac{1}{c})\lambda\sqrt{s}/(n\kappa_\cc)$ where $\Delta_\cc=\{\delta : \|\delta_{T^c}\|_1\leq \cc \|\delta_T\|_1\}$, $\cc = (c+1)/(c-1)$ (see Appendix \ref{Sec:AnalysisAux}). Without loss of generality assume that $T$ contains the treatment $d$ in its support. Thus for $\tilde x_i = (d_i,x_i')'$, $\delta = (\delta_d,\delta_x')'$, we have
$$ \begin{array}{rl}
\inf_{\delta\in \Delta_\cc}\frac{\|\sqrt{w_i}\tilde x_i'\delta\|_{2,n}^{3}}{\En[w_i|\tilde x_i'\delta|^3]} & \geq \inf_{\delta\in \Delta_\cc}\frac{\|\sqrt{w_i}\tilde x_i'\delta\|_{2,n}^{2}\|\delta_T\|\kappa_\cc}{\max_{i\leq n}\|\tilde x_i\|_\infty\|\delta\|_1\|\sqrt{w_i}\tilde x_i'\delta\|_{2,n}^{2}}\\
& \geq \inf_{\delta\in \Delta_\cc}\frac{\|\delta_T\|\kappa_\cc}{\max_{i\leq n}\|\tilde x_i\|_\infty(1+\cc)\|\delta_T\|_1}\\
& \geq \frac{\kappa_\cc}{\max_{i\leq n}\|\tilde x_i\|_\infty(1+\cc)\sqrt{s}} \gtrsim_P \frac{1}{\sqrt{s}K_1}\\
\end{array}$$
by $K_1=\Ep[\max_{i\leq n}\|(d,x')\|_\infty]$ and Markov inequality. Moreover, since $\lambda \lesssim \sqrt{n\log (p\vee n)}$, and $K_1^2s^2\log(p\vee n) \leq \delta_n n$ we have
$$ \frac{\lambda\sqrt{s}}{n\kappa_\cc}\lesssim \frac{1}{\kappa_\cc}\sqrt{\frac{s\log(p\vee n)}{n}}\leq \frac{1}{\kappa_\cc}\frac{\delta_n^{1/2}}{K_1\sqrt{s}} \ll \frac{1}{\sqrt{s}K_1} \lesssim_P \inf_{\delta\in \Delta_\cc}\frac{\|\sqrt{w_i}\tilde x_i'\delta\|_{2,n}^{3}}{\En[w_i|\tilde x_i'\delta|^3]}$$
and the required condition $q_{\Delta_\cc} > 3(1+\frac{1}{c})\lambda\sqrt{s}/(n\kappa_\cc)$ holds with probability  $1-o(1)$ since $\delta_n \to 0$. By Lemma \ref{lemma:PenaltyHoeffding}, setting $\lambda = c\sqrt{2n\log(2(p+1)/\gamma)}$ we have $\lambda/n \geq c\|\nabla\Lambda(\alpha_0,\beta_0)\|_\infty$ with probability $1-\gamma$. Therefore, since $\min_{i\leq n} w_i \geq c$ and $k$-sparse eigenvalues for $k=Cs$ are bounded away from zero, Lemma \ref{Lemma:LassoLogisticRate} yields $\|(\hat\alpha,\hat\beta)-(\alpha_0,\beta_0)\|\lesssim \sqrt{s\log(p\vee n)/n}$, $\|(\hat\alpha,\hat\beta)-(\alpha_0,\beta_0)\|_1\lesssim s\sqrt{\log(p\vee n)/n}$ with probability $1-o(1)$. Furthermore, since $\max_{i\leq n}\|\tilde x_i\|_\infty \|(\hat\alpha,\hat\beta)-(\alpha_0,\beta_0)\|_1 \lesssim_P K_1 \|(\hat\alpha,\hat\beta)-(\alpha_0,\beta_0)\|_1 \lesssim \delta_n^{1/2}\leq 1$ for $n$ sufficiently large, we have $\|\hat\beta\|_0 \lesssim s$ and $\Lambda(\hat\alpha,\hat\beta)-\Lambda(\alpha_0,\beta_0) \lesssim s \log(p\vee n)/n$.

To apply Lemma \ref{Lemma:PostLassoLogisticRate} we need to verify the side condition \begin{equation}\label{VerifyConditionPostLog}q_{A_{\hat s + s}}/6 > \max\left\{ \sqrt{\hat s+s}\|\nabla \Lambda(\eta_0)\|_\infty /\sqrt{\semin{\hat s+s}},\ \ \sqrt{\max\{0,\Lambda(\widetilde \alpha,\widetilde \beta)-\Lambda(\alpha_0,\beta_0)\}}\right\}.\end{equation} Because of the sparsity obtained by Lemma \ref{Lemma:LassoLogisticRate}, it suffices to consider $\hat s \leq Cs$ for some constant $C$. Similarly to the previous argument, for  $\tilde x_i = (d_i,x_i')'$ and $\delta = (\delta_d,\delta_x')'$,
\begin{equation*}
\begin{array}{rl}
{\displaystyle \inf_{\|\delta\|_0\leq s+Cs}}\frac{\|\sqrt{w_i}\tilde x_i'\delta\|_{2,n}^{3}}{\En[w_i|\tilde x_i'\delta|^3]} & \geq {\displaystyle \inf_{\|\delta\|_0\leq s+Cs}}  \frac{\{\semin{s+Cs}\}^{1/2}\|\delta\|\min_{i\leq n} w_i^{1/2}\|\sqrt{w_i}\tilde x_i'\delta\|_{2,n}^{2}}{\max_{i\leq n}\|\tilde x_i\|_\infty\|\delta\|_1\|\sqrt{w_i}\tilde x_i'\delta\|_{2,n}^{2}}\\
& \geq \frac{\{\semin{s+Cs}\}^{1/2}\min_{i\leq n} w_i^{1/2}}{\max_{i\leq n}\|\tilde x_i\|_\infty\sqrt{s+Cs}} \gtrsim_P \frac{1}{K_1\sqrt{s}}\\
\end{array}
\end{equation*}
since the minimum sparse eigenvalue of order $s+Cs$ and $\min_{i\leq n} w_i^{1/2}$ are bounded away from zero. From Lemma \ref{lemma:PenaltyHoeffding} we have $\|\nabla\Lambda(\alpha_0,\beta_0)\|_\infty \lesssim \sqrt{\log(p\vee n)/n}$, and from the definition of the post-selection estimator $\Lambda(\widetilde\alpha,\widetilde\beta)\leq \Lambda(\hat\alpha,\hat\beta)$, so that $$\Lambda(\widetilde\alpha,\widetilde\beta)-\Lambda(\alpha_0,\beta_0) \leq \Lambda(\hat\alpha,\hat\beta)-\Lambda(\alpha_0,\beta_0) \lesssim s \log(p\vee n)/n$$ by Lemma \ref{Lemma:LassoLogisticRate}. Therefore, the right hand side of (\ref{VerifyConditionPostLog}) is bounded above  by $C'\sqrt{s\log(p\vee n)/n}$. Then again the condition $K_{1}^2s^2\log(p\vee n) \leq \delta_n n$ suffices for (\ref{VerifyConditionPostLog}) to hold with probability $1-o(1)$ for $n$ sufficiently large. Thus Lemma \ref{Lemma:PostLassoLogisticRate} yields $\|(\widetilde \alpha, \widetilde \beta) - (\alpha_0,\beta_0)\|\lesssim \sqrt{s\log(p\vee n)/n}$ with probability $1-o(1)$ (and $\|\widetilde\beta\|_0\leq \|\hat\beta\|_0 \lesssim s$ with probability $1-o(1)$).

Step 2 relies on Post-Lasso with estimated weights. Parts (i) and (ii) of  Condition WL are assumed by Conditions L and a suitable choice of the confidence level $\gamma=n^{-1/4}$ to satisfy the growth condition.  Condition WL(iii) follows from Lemma \ref{cor:LoadingConvergence} with $X_{ij}=\sqrt{w_i}x_{ij}v_i$ under the condition that $K_4^4\log p \leq \delta_n n$ (since $|v_i| \leq |z_{0i}|$).

To show Condition WL(iv) note that $\hat f_i = \hat w_i/\hat \sigma_i = \sqrt{\hat w_i}$, $ \hat w_i = G'(d_i\widetilde\alpha+x_i'\widetilde\beta) \leq 1$ and $w_i = G'(d_i\alpha_0+x_i'\beta_0)\leq 1$.
The first part follows from $\En[\hat f_id_i^2] \leq \En[d_i^2] \leq \barEp[d_i^2] + (\En-\barEp)[d_i^2] \lesssim C$ with probability $1-o(1)$ by Chebyshev's inequality and $\barEp[d_i^4]\leq C$.
Since $G'$ is 1-Lipschitz and $0\leq \hat w_i \leq 1$, $|\sqrt{a}-\sqrt{b}|\leq \sqrt{|a-b|}$, the second part of Condition WL(iv) follows from
$$\begin{array}{rl}
 \max_{j\leq p}\En[\{\sqrt{\hat w_i}-\sqrt{w_i}\}^2x_{ij}^2v_i^2] & \leq \max_{i\leq n}|\sqrt{\hat w_i}-\sqrt{w_i}|^2 \max_{j\leq p}\En[x_{ij}^2v_i^2] \\
 & \leq \max_{i\leq n}\|\tilde x_i\|_\infty \|(\widetilde\alpha,\widetilde\beta)-(\widetilde\alpha,\beta_0)\|_1 \max_{j\leq p}\En[x_{ij}^2v_i^2]\\
 &\lesssim_P K_1 \sqrt{s^2\log(p\vee n)/n} \max_{j\leq p}\En[x_{ij}^2v_i^2] \lesssim_P \delta_n^{1/2}
 \end{array}$$
since the rate of $(\widetilde\beta,\widetilde\alpha)$ established before, $K_1^2s^2\log(p\vee n) \leq \delta_n n$, and $\max_{j\leq p}\En[x_{ij}^2v_i^2] \leq \max_{j\leq p}|(\En-\barEp)[x_{ij}^2v_i^2]| + \max_{j\leq p}\barEp[x_{ij}^2v_i^2] \lesssim_P 1$. (Indeed by Lemma \ref{cor:LoadingConvergence} with $X_{ij}=x_{ij}v_i$ under the condition that $K_4^4\log p \leq \delta_n n$, $\barEp[x_{ij}^2v_i^2] \leq \{\barEp[x_{ij}^4]\barEp[(d_i-x_i'\theta_0)^4]\}^{1/2} \leq C$ by the fourth moment condition and $\|\theta_0\|\leq C$.)

To show that last requirement of Condition WL(iv),  $\hat f_i = \hat w_i/\hat \sigma_i = \sqrt{\hat w_i}$ yields $\hat c_f^2=\En[(\hat w_i - w_i)^2v_i^2/\sqrt{w_i}]$. By $G'$ begin 1-Lipschitz, $|\hat w_i-w_i|\leq |x_i'(\widetilde \beta-\beta_0)|+|d_i(\widetilde \alpha - \alpha_0)|$, and $\min_{i\leq n} w_i \geq c > 0$  with probability $1-\Delta_n$,  with the same probability we have
\begin{equation}\label{Eq:cw}\begin{array}{rl}
 \hat c_f^2 & = \En[ (\hat w_i - w_i)^2v_i^2/\sqrt{w_i} ] \leq \frac{2}{\sqrt{c}}\En[\{v_ix_i'(\widetilde \beta-\beta_0)\}^2]+ \frac{2}{\sqrt{c}}|\widetilde\alpha-\alpha_0|^2\En[d_i^2v_i^2] \\
 & \leq  \frac{2}{\sqrt{c}}(\En-\barEp)[\{v_ix_i'(\widetilde \beta-\beta_0)\}^2]+ \frac{2}{\sqrt{c}}\barEp[\{v_ix_i'(\widetilde \beta-\beta_0)\}^2]+ \frac{2}{\sqrt{c}}|\widetilde\alpha-\alpha_0|^2\{\En[d_i^4]\}^{1/2}\{\En[v_i^4]\}^{1/2} \\ \end{array}\end{equation}
Recall that $\|\widetilde\beta\|_0\lesssim s$, $\|\widetilde\beta-\beta_0\|\lesssim \sqrt{s\log p/n}$,  $|\widetilde\alpha-\alpha_0|\lesssim \sqrt{s\log p/n}$ with probability $1-o(1)$. We will apply Lemma \ref{thm:RV34} with $X_i=v_ix_i$. In that case, we have $K=\{\Ep[\max_{i\leq n}\|X_i\|_\infty^2]\}^{1/2} \leq \{\Ep[\max_{i\leq n} \|(v_i,x_i')\|_\infty^4]\}^{1/2} \leq K_4^2$, and $\barEp[(\delta'X_i)^2] = \barEp[v_i^2(x_i'\delta)^2]\leq \{\barEp[v_i^4]\barEp[(x_i'\delta)^4]\}^{1/2} \leq C \|\delta\|^2$ by the fourth moment condition and $|v_i|\leq |d-x_i'\theta_0|$. Therefore, since $\|\widetilde\beta-\beta_0\|_0 \leq 2Cs$ we have
$$\begin{array}{rl}
(\En-\barEp)[\{v_ix_i'(\widetilde \beta-\beta_0)\}^2] & \leq \|\widetilde \beta-\beta_0\|^2 \sup_{\|\delta\|_0\leq 2Cs, \|\delta\| =1} \left| \En\[ (\delta'X_i)^2 - \barEp[(\delta'X_i)^2] \]\right|\\
& \lesssim_P  \|\widetilde \beta-\beta_0\|^2 \left\{ \frac{K_4^4 s \log^{3}n \log (p\vee n)}{n} + \sqrt{\frac{K_4^4 s \log^{3}n \log (p\vee n)}{n}}\right\} \\
& \lesssim \frac{s \log (p\vee n)}{n} \delta_n^{1/2}\end{array}$$
under the assumed condition $K_4^4 s \log^{3}n \log (p\vee n) \leq \delta_n n$ and $\|\widetilde \beta-\beta_0\|^2 \lesssim s\log p / n$
with probability $1-o(1)$. Similarly,  with probability going to one, the last term in (\ref{Eq:cw}) satisfies with probability going to 1
$$ |\widetilde\alpha-\alpha_0|^2\{ |\En[d_i^2v_i^2]-\barEp[d_i^2v_i^2]|+\barEp[d_i^2v_i^2]\} \lesssim s \log (p\vee n) / n.$$
Therefore, $\hat c_f^2 \lesssim s \log (p\vee n) / n$ with probability $1-o(1)$ (which implies $n^2\hat c_f^2/\lambda^2 \lesssim s$). In turn, because the restricted eigenvalue and $\min_{i\leq n} \sqrt{\hat w_i/w_i}$ are bounded away from zero, the required condition is satisfied as $\max_{i\leq n} \|x_i\|_\infty^2 \{\frac{\lambda\sqrt{s}}{n}+\hat c_f\} \lesssim_P K_1^2 \sqrt{\frac{s\log(p\vee n)}{n}} \lesssim \delta_n^{1/2}$ under the assumption $K_1^4s \log(p\vee n) \leq \delta_n n$. Therefore, by Theorem \ref{Thm:RateEstimatedLasso},  we have $\|\widetilde\theta-\theta_0\| \lesssim \sqrt{s\log(p\vee n)/n}$ and $\|\widetilde \theta\|_0 \lesssim Cs$ with probability $1-o(1)$.

%-----
%[CHECK]
%Therefore $|\widetilde\alpha-\alpha_0|\lesssim_P\sqrt{s\log(p\vee n)/n}$ so that $\mathcal{A} = \{\alpha : |\alpha - \widetilde \alpha| \leq C \log^{-1} n\} \supseteq \{\alpha : |\alpha-\alpha_0| \leq n^{-1/2}/\delta_n\}$ under $s \log (p\vee n) \log^2 n \leq \delta_n n$ which is required in ILOG(i). This also ensures the initial rate required for $\check \alpha$ in ILOG(iii) since $\check \alpha \in \mathcal{A}$.
%----

After establishing rates of convergence for $(\widetilde\alpha,\widetilde\beta)$ and $\widetilde\theta$ we proceed to verify Condition IR(iii). Note that $\check\alpha \in \mathcal{A} \subset \{ \alpha : |\alpha-\widetilde \alpha|\leq C \log^{-1} n\} \subset  \{ \alpha : |\alpha-\alpha_0|\leq C \log^{-1} n\}$ so that $|\check \alpha - \alpha_0|\leq C \log^{-1} n$. The choice of instrument is $z_{0i}=v_i/\sqrt{w_i}=d_i-x_i'\theta_0$ and $\hat z_i = d_i - x_i'\widetilde \theta$ so that
\begin{equation}\label{Eq:vi}\hat z_i -  z_{0i} = x_i'\{ \theta_0 - \widetilde\theta\}\end{equation}
The rates established above for $(\widetilde \alpha, \widetilde \beta, \widetilde \theta)$ imply (\ref{Eq:HLfirstestimated}) in Condition IR(ii) since, under $s^2\log^2(p\vee n)\leq \delta_n n$, we have $$\begin{array}{rl}
\|\widetilde \beta-\beta_0\|_{2,n} & \lesssim  \sqrt{\frac{s\log (p\vee n)}{n}} \lesssim \delta_n n^{-1/4}\\
\barEp[(\tilde z_i-z_{0i})^2]|_{\tilde z = \hat z} & = \barEp[ \{x_i'(\theta_0-\widetilde\theta)\}^2] \lesssim  \|\theta_0-\widetilde\theta\|^2\lesssim s\log(p\vee n)/n \lesssim \delta_n\\
\|\widetilde \beta-\beta_0\| \ \barEp[(\tilde z_i-z_{0i})^2]|_{\tilde z = \hat z}& \lesssim \sqrt{\frac{s\log (p\vee n)}{n}}\sqrt{\frac{s\log (p\vee n)}{n}}\leq n^{-1/2}\frac{s\log (p\vee n)}{n^{1/2}}\lesssim \delta_n n^{-1/2}\\
%|\widetilde \alpha - \alpha_0| \ \|x_i'(\widetilde \theta-\theta_0)\|_{2,n}& \lesssim \delta_n n^{-1/2}
\end{array}$$
with probability $1-o(1)$.

Next we verify Condition IR(iv). By definition $|\hat w_i|\leq 1$ and since $G'$ is 1-Lipschitz, we have $\|\hat w_i - w_i\|_{2,n} \leq  \|d_i\|_{2,n}|\widetilde\alpha-\alpha_0|+\|x_i'(\widetilde \theta - \theta_0)\|_{2,n} \lesssim \sqrt{s\log (p\vee n)/n} \lesssim \delta_n^{1/2}$ with probability $1-o(1)$. Moreover, with probability $1-o(1)$
$$ \|d_i(\hat z_i-z_{0i})\|_{2,n}|\check\alpha-\alpha_0| \leq \max_{i\leq n}|d_i| \| x_i'(\theta_0-\widetilde\theta)\|_{2,n}|\check\alpha-\alpha_0|\lesssim K_1 \log n \sqrt{\frac{s \log(p\vee n)}{n}} \frac{1}{\log n} \lesssim \delta_n^{1/2}$$
since $\Pr(\max_{i\leq n}|d_i| > K_1\log n)\leq 1/\log n$, and the condition $K_1^2 s \log(p\vee n) \leq \delta_n n$ holds. Similarly, $\|z_{0i}x_i'(\widetilde\beta-\beta_0)\|_{2,n} \leq \max_{i\leq n}|z_{0i}|\|x_i'(\widetilde\beta-\beta_0)\|_{2,n} \lesssim K_1 \delta_{n}^{-1/4} \sqrt{s \log (p\vee n)/n}\lesssim \delta_n^{1/4}$ with probability $1-o(1)$.

Next we verify Condition IR(iii) part (\ref{Eq:HLestimated}). Let $\hvpi(\alpha) = y_i - G(d_i \alpha + x_i'\widetilde\beta)$, $\vpi(\alpha) = y_i - G(d_i \alpha + x_i'\beta_0)$. Note that
%$$ \sup_{\alpha \in \mathcal{A}}\left|(\En-\barEp)\left[\hvpi(\alpha)\hat v_i - \vpi(\alpha) v_i\right]\right| $$
%$$ \leq \sup_{\alpha \in \mathcal{A}}\left|(\En-\barEp)\left[\varphi_u(y_i - x_i'\hat \beta_g - d_i \alpha)\hat v_i - \varphi_u(y_i -  x_i'\beta_g(\tau) - d_i \alpha) v_i\right]\right| $$ $$+ \sup_{\alpha \in \mathcal{A}}\left|(\En-\barEp)\left[ \varphi_u(y_i - x_i'\beta_g(\tau) - d_i \alpha) v_i -\varphi_u(y_i - g_i - d_i \alpha) v_i\right]\right| $$
{\small \begin{equation}\label{Alg1Eq:I(ii)first} \sup_{\alpha \in \mathcal{A}}\left|(\En-\barEp)\left[\hvpi(\alpha)\hat z_i - \vpi(\alpha) z_{0i}\right]\right| \leq \sup_{\alpha \in \mathcal{A}}\left|(\En-\barEp)\left[\{\hvpi(\alpha)-\vpi(\alpha)\}(\hat z_i -  z_{0i})\right]\right|+\end{equation}
\begin{equation}\label{Alg1Eq:I(ii)seconda} \ \ \ \ \ \ \ \ \ \ \ \ \ \ \ \ \ \ \ \ \ \ \ \ \ \ \ \ \ \ \ \  + \sup_{\alpha \in \mathcal{A}}\left|(\En-\barEp)\left[\vpi(\alpha)(\hat z_i -  z_{0i})\right]\right|+\end{equation}
\begin{equation}\label{Alg1Eq:I(ii)second} \ \ \ \ \ \ \ \ \ \ \ \ \ \ \ \ \ \ \ \ \ \ \ \ \ \ \
\ \ \ \ \ \ \ \  + \sup_{\alpha \in \mathcal{A}}\left|(\En-\barEp)\left[\{\hvpi(\alpha) - \vpi(\alpha) \}z_{0i}\right]\right|. \end{equation}}
To bound (\ref{Alg1Eq:I(ii)first}), since $|\hvpi(\alpha)-\vpi(\alpha)|\leq |x_i'(\widetilde \beta - \beta_0)|$, $|\hat z_i -  z_{0i}|=|x_i'(\widetilde \theta-\theta_0)|$,  $|\hat w_i|\leq 1$, $\|\widetilde\beta\|_0+\|\widetilde\theta\|_0\lesssim s$, $\semax{2Cs}$ is uniformly bounded, we use Cauchy-Schwartz to obtain with probability $1-o(1)$ that
{\small $$\begin{array}{rl}
 \hspace{-1cm}{\rm (\ref{Alg1Eq:I(ii)first})} & \lesssim  \|\widetilde \beta-\beta_0\| \ \| \widetilde \theta - \theta_0\|  \lesssim s \log(p\vee n)/n \lesssim \delta_n n^{-1/2}\end{array}$$}
under the condition $s^2\log^2(p\vee n) \leq \delta_n n$.

To bound (\ref{Alg1Eq:I(ii)seconda}) we use that with probability $1-o(1)$, $\|\widetilde \theta - \theta_0\|_1 \leq Cs\sqrt{\log(p\vee n)/n}$ so that with the same probability
{\small $$\begin{array}{rl}
 {\rm (\ref{Alg1Eq:I(ii)seconda})} & \leq \sup_{\alpha \in \mathcal{A}}\left|(\En-\barEp)\left[\{\vpi(\alpha)-\vpi(\alpha_0)\}x_i'(\theta_0-\widetilde \theta)\right]\right| +  \left|(\En-\barEp)\left[\vpi(\alpha_0)x_i'(\theta_0-\widetilde \theta)\right]\right|\\
 & \lesssim_P {\displaystyle \sup_{\alpha \in \mathcal{A}, \|\delta\|_1 = C s \sqrt{\frac{\log(p\vee n)}{n}}}}\left|(\En-\barEp)\left[\{\vpi(\alpha)-\vpi(\alpha_0)\} x_i'\delta\right]\right| + {\displaystyle \sup_{ \|\delta\|_1\leq C s \sqrt{\frac{\log(p\vee n)}{n}}}} \left|(\En-\barEp)\left[\vpi(\alpha_0) x_i'\delta\right]\right|.\\\end{array}$$}
We will use Lemma \ref{Lemma:ProcessL1Control} for each term. To handle the first term let $W_{ij} = d_ix_{ij}$. Define $r_1=Cs\sqrt{s\log(p\vee n)/n}$, $\mathcal{T}_+ = \{ (\alpha - \alpha_0) \delta \in \RR^p : \alpha \geq \alpha_0, \alpha \in \mathcal{A}, \|\delta\|_1 = r_1 \}$ and $\mathcal{T}_- = \{ (\alpha - \alpha_0) \delta \in \RR^p : \alpha \leq \alpha_0, \alpha \in \mathcal{A}, \|\delta\|_1 = r_1 \}$. For $t\in \mathcal{T}_+$ we define $h_i^+(t) =\{\vpi(\alpha_0+\|t\|_1/r_1)-\vpi(\alpha_0)\}x_i'tr_1/\|t\|_1$. By construction we have $|h_i(t)^+|\leq |t'W_i|=|(\alpha-\alpha_0)d_ix_i'\delta|$.  Similarly, for $t\in \mathcal{T}_-$ we define $h_i^-(t) =\{\vpi(\alpha_0-\|t\|_1/r_1)-\vpi(\alpha_0)\}x_i'tr_1/\|t\|_1$. Note that $\|\mathcal{T}\|_1 \lesssim  s\sqrt{\log(p\vee n)/n}$. So we have
$$ \begin{array}{rl} {\displaystyle \sup_{\alpha \in \mathcal{A}, \|\delta\|_1 = C s \sqrt{\frac{\log(p\vee n)}{n}}}}\left|(\En-\barEp)\left[\{\vpi(\alpha)-\vpi(\alpha_0)\} x_i'\delta\right]\right|\\
  \leq {\displaystyle \sup_{t\in \mathcal{T}_+}}\left|(\En-\barEp)\left[h_i(t)^+\right]\right| + {\displaystyle \sup_{t\in \mathcal{T}_-}}\left|(\En-\barEp)\left[h_i^-(t)\right]\right| \\
 \lesssim  s\sqrt{\log(p\vee n)/n} \sqrt{\frac{\log(p\vee n)}{n}} \lesssim \delta_n^{1/2} n^{-1/2} \end{array}$$
under $s^2 \log^2(p\vee n) \leq \delta_n n$ and $K_4^4\log p \leq \delta_n n$ by Lemma \ref{Lemma:ProcessL1Control} with $K^2 = C\log (p\vee n)$ and $M \lesssim C$ as $\max_{j\leq p} \En[d_i^2x_{ij}^2]$ is bounded with probability $1-o(1)$ under  $K_4^4\log p \leq \delta_n n$ by Lemma \ref{cor:LoadingConvergence}. To bound the second term we again use Lemma \ref{Lemma:ProcessL1Control} with $t=\delta$ and $h_i(t)= \vpi(\alpha_0) x_i'\delta$ which satisfies $|h_i(t)| \leq |t'W_i| = |t'x_i|$. Thus we have with probability $1-o(1)$
{\small $$\begin{array}{rl}
{\rm (\ref{Alg1Eq:I(ii)seconda})} & \lesssim  s \sqrt{\frac{\log(p\vee n)}{n}} \sqrt{\frac{\log (p\vee n)}{n}} +   s \sqrt{\frac{\log(p\vee n)}{n}} \sqrt{\frac{\log (p\vee n)}{n}}\lesssim \delta_n^{1/2} n^{-1/2}\\\end{array}$$}

Next we proceed to bound (\ref{Alg1Eq:I(ii)second}). We will consider the class of functions which pertains to $\{\hvpi(\alpha) - \vpi(\alpha)\}z_{0i}$, namely for some $C$ suitably large
$$ \mathcal{F} = \{ \G(d_i\alpha + x_i'\beta)z_{0i} - \G(d_i\alpha + x_i'\beta_0)z_{0i} : \|\beta\|_0\leq Cs, \ \|\beta - \beta_0\| \leq C\sqrt{s\log p / n}\} $$
Let $h_i(t,\alpha) = \G( d_i\alpha + x_i'(\beta+t)) - \G(d_i\alpha + x_i'\beta_0)$ so that $|h_i(t,\alpha)|\leq |t'x_iz_{0i}|$. Therefore $\|\mathcal{T}\|_1 \lesssim s\sqrt{\log (p\vee n) / n }$ and note that
$ \max_{j\leq p}\En[x_{ij}^2z_{0i}^2] \lesssim C$ with probability $1-o(1)$ under $K_4^4\log p \leq \delta_n n$ and $ \max_{j\leq p}\barEp[x_{ij}^2z_{0i}^2]\leq \max_{j\leq p}\{\barEp[x_{ij}^4]\barEp[z_{0i}^4]\}^{1/2}\leq C$ by the fourth moment condition. By Lemma \ref{Lemma:ProcessL1Control}, with probability $1-o(1)$  we have
$$ (\ref{Alg1Eq:I(ii)second}) \lesssim \sqrt{\frac{\log (p\vee n)}{n}} \|\mathcal{T}\|_1 \lesssim \frac{s \log (p\vee n)}{n}\lesssim \delta_n^{1/2} n^{-1/2} $$
provided $s^2\log^2(p\vee n) \leq \delta_n n$.

Next we verify the requirement (\ref{Eq:HLhatalpha}) in Condition IR(iii). Note that $\mathcal{A}=\{ \alpha : |\alpha-\widetilde \alpha| \leq C \log^{-1} n \} \supseteq \{ \alpha : |\alpha-\alpha_0| \leq (C/2)\log^{-1}n\}$ for $n$ large enough since $|\widetilde \alpha-\alpha_0|\lesssim \sqrt{s\log (p\vee n)/n}$ with probability $1-o(1)$. Also, $\check\alpha \in \mathcal{A}$ implies that $|\check\alpha-\alpha_0|\lesssim \log^{-1}n$ with probability $1-o(1)$ for $n$ sufficiently large. We will show that $\En[\hvpi(\alpha)\hat z_i]$ changes sign over $\alpha \in \mathcal{A}$ with probability $1-o(1)$ which by continuity of $\hvpi(\cdot)$ implies that $\En[\hvpi(\check \alpha)\hat z_i]=0$ with probability $1-o(1)$. Note that for any $\alpha\in\mathcal{A}$
%$$\begin{array}{rl}
% \displaystyle L_n(\check \alpha)& \displaystyle  = \frac{\{\En[\hvpi(\check \alpha)\hat v_i]\}^2}{\En[\hvpi^2(\check \alpha)\hat v_i^2]} = \min_{\alpha\in \mathcal{A}} \frac{\{\En[\hvpi(\alpha)\hat v_i]\}^2}{\En[\hvpi^2(\alpha)\hat v_i^2]}\displaystyle \leq  \frac{4}{\En[\hat v_i^2]}\min_{\alpha\in \mathcal{A}} \{\En[\hvpi(\alpha)\hat v_i]\}^2 \\
% \end{array} $$
%We also have
$$\begin{array}{rl}
 \displaystyle  \En[\hvpi(\alpha)\hat z_i]  & = \overbrace{(\En-\barEp)[\hvpi(\alpha)\hat z_i - \vpi(\alpha) z_{0i}]}^{(1)} + \overbrace{\barEp[\hvpi(\alpha)\hat z_i] - \barEp[\vpi(\alpha) z_{0i}]}^{(2)} + \\
 & + \underbrace{(\En-\barEp)[\vpi(\alpha) z_{0i}]}_{(3)} + \barEp[\vpi(\alpha) z_{0i}]. \\
 \end{array}$$

By Condition IR(iii) part (\ref{Eq:HLestimated}), we have $|(1)| \lesssim \delta_n^{1/2} n^{-1/2}$ with probability $1-o(1)$. With probability $1-o(1)$, by the expansion (\ref{EqBias}), $\barEp[\vpi(\alpha) z_{0i}]= -\barEp[w_id_iz_{0i}](\alpha-\alpha_0)+O(\delta_n|\alpha-\alpha_0|+\delta_nn^{-1/2})$ and $\barEp[\hvpi(\alpha) z_{0i}]= -\barEp[w_id_iz_{0i}](\alpha-\alpha_0)+O(\delta_n|\alpha-\alpha_0|+\delta_nn^{-1/2})$ so that we have $ |(2)| \lesssim \delta_n n^{-1/2} + \delta_n|\alpha-\alpha_0|$ with the same probability. Therefore, with probability $1-o(1)$ we have
 $$|(3)| \leq \sup_{\alpha \in \mathcal{A}}|(\En-\barEp)[\{\vpi(\alpha)-\vpi(\alpha_0)\} z_{0i}]| + |(\En-\barEp)[\vpi(\alpha_0) z_{0i}]| \lesssim \delta_n n^{-1/2}+ n^{-1/2}\log n$$
where the first term is bounded using Lemma \ref{Lemma:ProcessL1Control} and $\barEp[d_i^2z_{0i}^2] = O(1)$, and the second term by Chebyshev's inequality.

Therefore, since $\barEp[\vpi(\alpha) z_{0i}] = (\alpha-\alpha_0)\barEp[v_i^2]+ O(|\alpha-\alpha_0|^2)$, we have with probability $1-o(1)$ that \begin{equation}\label{Eq:Minimizer}\begin{array}{rl}
 \En[\vpi(\alpha) z_{0i}] & = O(n^{-1/2}\log n+\delta_n|\alpha-\alpha_0|) + \barEp[\vpi(\alpha) z_{0i}] \\
 & = O(n^{-1/2}\log n) + (\alpha-\alpha_0)\{\barEp[v_i^2]+ O(\delta_n)\} + O(|\alpha-\alpha_0|^2).\end{array}\end{equation}
Since $\barEp[v_i^2] \geq  c$ and $\delta_n\to 0$, when we evaluate (\ref{Eq:Minimizer}) on the extreme points $\alpha^k, k=1,2,$ of $\mathcal{A}$, we obtain a positive value for one extreme and a negative value for the other extreme for $n$ large enough since $|\alpha^k - \alpha_0|\geq (C/2) \log^{-1} n$, $k=1,2$.
\end{proof}

\begin{proof}[Proof of Theorem \ref{theorem:inferenceAlg2}]
Let $\hat T^* = \supp(\hat\theta)\cup\supp(\hat \beta)$. By the first order condition of the minimization problem in Step 3, and $\hat f_i/\hat\sigma_i=1$ in the logistic case, we have
\begin{equation}\label{OPTdouble}
 \En[ \{y_i-\G(d_i\check\alpha  + x_i'\check\beta)\} (d_i, \ x_{i\hat T^*}')']=0.\end{equation}
Next we will construct a suitable instrument to apply Theorem \ref{Thm:EstIVQRgeneric}. Define $$\hat\theta^*\in \arg\min_\theta \| x_{i}'(\theta-\theta_0)\|_{2,n} \ \ : \ \supp(\theta) \subseteq \hat T^*.$$
We use the optimal instrument $z_{0i} = v_i/\sqrt{w_i}=d_i - x_i'\theta_0$ and the estimated instrument $\hat z_{i} = d_i - x_i'\hat\theta^*$. Note that by (\ref{OPTdouble}), taking the linear combination $(1;-\hat\theta^*)$ of the optimality condition we have $$ \En[ \{y_i-\G(d_i\check\alpha + x_i'\check\beta)\} \hat z_i] = 0.$$ Therefore $\check \alpha$ minimizes the criterion $$L_n(\alpha) = \frac{|\En[\{y_i-\G(d_i\alpha+x_i'\check\beta)\}\hat z_i]|^2}{\En[\{y_i-\G(d_i\alpha+x_i'\check\beta)\}^2\hat z_i^2]},$$ induced by  $\{(x_i'\check\beta,\hat z_i): i=1,\ldots,n\}$, over $\alpha \in \RR$.

Regarding Steps 1 and 2, rates of convergence for $\ell_1$-penalized logistic regression, post selection logistic regression, and Lasso with estimated weights and the associated sparsity bounds are established as in the proof of Theorem \ref{theorem:inferenceAlg1}. Thus we have with probability $1-o(1)$ that $\|\hat \theta\|_0\lesssim s$, $\|\hat\beta\|_0\lesssim s$, $\Lambda(\hat\alpha,\hat\beta)-\Lambda(\alpha_0,\beta_0) \lesssim s\log p/n$ and $\|\hat\theta-\theta_0\| \lesssim \sqrt{s\log(p\vee n)/n}$.

Next we analyze Step 3. The sparsity results above implies that $\hat T^* = \supp(\hat\theta)\cup\supp(\hat \beta)$ satisfies $|\hat T^*|\lesssim s$ with probability $1-o(1)$. Moreover, since $\supp(\hat\beta)\subset \hat T^*$ we have with probability $1-o(1)$ that $$\Lambda(\check\alpha,\check\beta)-\Lambda(\alpha_0,\beta_0)\leq \Lambda(\hat\alpha,\hat\beta)-\Lambda(\alpha_0,\beta_0) \lesssim s\log p/n.$$ Thus, the requirements for Lemma \ref{Lemma:PostLassoLogisticRate} hold as before. Since $k$-sparse eigenvalues are bounded away from zero for $k=s\ell_n$, for some $\ell_n\to\infty$,   Lemma \ref{Lemma:PostLassoLogisticRate} establishes a rate of convergence for post-model selection Logistic regression estimator $\|\check\beta-\beta_0\|\lesssim \sqrt{ s \log p / n }$, $|\check\alpha-\alpha_0|\lesssim \sqrt{ s \log p / n }$, and $\|\check\beta-\beta_0\|_1 \lesssim_P \sqrt{s}\|\check\beta-\beta_0\| \lesssim \sqrt{s}\|\check\beta-\beta_0\| \lesssim s\sqrt{\log p/n}$. Moreover, since $\supp(\hat\theta) \subset \hat T^*$ we have $\| x_{i}'(\hat\theta^*-\theta_0)\|_{2,n} \leq \| x_{i}'(\hat\theta-\theta_0)\|_{2,n}\lesssim \sqrt{s\log(p\vee n)/n}$ and $\|\hat\theta^*-\theta_0\|_1 \lesssim_P \sqrt{s}\|\hat\theta^*-\theta_0\| \lesssim \sqrt{s}\|x_i'(\hat\theta^*-\theta_0)\|_{2,n}/\{\semin{C's}\}^{1/2} \lesssim s\sqrt{\log(p\vee n)/n}$ with probability $1-o(1)$.

\noindent The remaining assumptions in Condition IR can be verified as in the proof of Theorem \ref{theorem:inferenceAlg1}.

Next we show the validity of the calculation of $\widehat \Sigma_{2n}^2=\{ \En[\check w_i (d_i,x_{i\hat T^*}')'(d_i,x_{i\hat T^*}')]\}^{-1}_{11}$. Since $\min_{i\leq n} w_i > c$ with probability $1-\Delta_n$ and $k$-sparse eigenvalues of size $k=s\ell_n$ are bounded away from zero and from above with probability $1-\Delta_n$ by Condition L, and $\max_{i\leq n} |\check w_i - w_i| = o_P(1)$ by the rates above, we have $$ \{ \En[\check w_i (d_i,x_{i\hat T^*}')'(d_i,x_{i\hat T^*}')]\}^{-1}_{11} = \{ \En[ w_i (d_i,x_{i\hat T^*}')'(d_i,x_{i\hat T^*}')]\}^{-1}_{11} + o_P(1).$$
Next note that
 $$\widetilde \Sigma_{2n}=\{ \En[ w_i (d_i,x_{i\hat T^*}')'(d_i,x_{i\hat T^*}')]\}^{-1}_{11} = \{ \En[ w_i d_i^2] - \En[ w_i d_ix_{i\hat T^*}']\{ \En[w_i x_{i\hat T^*}x_{i\hat T^*}']\}^{-1}\En[ w_i x_{i\hat T^*} d_i] \}^{-1}.$$
Note that $\check\theta[\widehat T^*] = \{ \En[ w_i x_{i\hat T^*}x_{i\hat T^*}']\}^{-1}\En[ w_i x_{i\hat T^*} d_i]$ is the least squares estimator of regressing $\sqrt{w_i}d_i$ on $\sqrt{ w_i}x_{i\hat T^*}$. We let $\check \theta$ denote the corresponding $p$-dimensional (sparse) vector. Therefore, using that $\sqrt{w_i}x_i'\theta_0 = \sqrt{w_i}d_i-v_i$ we have
$$\begin{array}{rl}
\widetilde \Sigma_{2n}^{-2} & =  \En[ w_i d_i^2] - \En[ w_i d_ix_i'\check\theta] \\
& =  \En[ w_i d_i^2] - \En[\sqrt{ w_i} d_i \sqrt{ w_i}x_i'\theta_0] - \En[\sqrt{ w_i} d_i \sqrt{w_i}x_i'(\check\theta-\theta_0)] \\
& =\En[ \sqrt{w_i}d_i v_i] - \En[\sqrt{ w_i} d_i \sqrt{w_i}x_i'(\check\theta-\theta_0)]\\
& =\En[v_i^2] + \En[ \sqrt{w_i} v_ix_i'\theta_0] - \En[\sqrt{ w_i} d_i \sqrt{w_i}x_i'(\check\theta-\theta_0)]\\
\end{array}$$
We have that $|\En[ \sqrt{w_i} v_ix_i'\theta_0]|=o_P(\delta_n)$ since $\barEp[ \sqrt{w_i} v_ix_i'\theta_0]=0$ and $\barEp[ (\sqrt{w_i} v_ix_i'\theta_0)^2] \leq \barEp[w_iv_i^2d_i^2] \leq \{\barEp[v_i^4]\barEp[d_i^4]\}^{1/2}\leq C$. Moreover, $|\En[\sqrt{ w_i} d_i \sqrt{w_i}x_i'(\check\theta-\theta_0)]| \leq \|d_i\|_{2,n}\|\sqrt{w_i}x_i'(\check\theta-\theta_0)\|_{2,n} = o_P(\delta_n)$ since $|\widehat T^*|\lesssim_P s$ and $\supp(\check\theta)\subset \widehat T^*$. The result follows.

\end{proof}

\begin{proof}[Proof of Theorem \ref{Thm:EstIVQRgeneric}]
Let $(d,x) \in \mathcal{D}\times\mathcal{X}$. In this section for $\tilde h = (\tilde \beta, \tilde z)$, where $\tilde z$ is a function on $(d,x)\mapsto \tilde z(d,x)$ we write
$$\psi_{\tilde \alpha,\tilde h}(y_i,d_i,x_i) = \psi_{\tilde \alpha,\tilde \beta,\tilde z}(y_i,d_i,x_i) = \{y_i - \G(x_i'\tilde\beta+d_i\alpha)\}\tilde z(d_i,x_i).$$ %= \{y_i - \G(x_i'\tilde\beta+d_i\alpha)\}\tilde z_i $$
Because of (\ref{Eq:MainLogisticModel}), $h_0 = (\beta_0, z_0)$ we have
$$ \Ep[ \psi_{\alpha_0,h_0}(y_i,d_i,x_i)]=0$$

 For a fixed $\tilde \alpha\in\RR$, $\tilde \beta \in \RR^p$, $\tilde z:\mathcal{D}\times\mathcal{X} \to \RR$, and $\tilde h = (\tilde\beta,\tilde z)$, we define
$$\Gamma(\tilde\alpha,\tilde h) := \barEp[ \psi_{\tilde \alpha,\tilde h}(y_i,d_i,x_i)] $$
For notational convenience we let $\tilde z_i=\tilde z(d_i,x_i)$, $h_0 = (\beta_0,z_0)$ and $\hat h = (\hat \beta,\hat z)$. The partial derivative of $\Gamma$ with respect to $\alpha$ at $(\tilde\alpha,\tilde h)$ is denoted by $\Gamma_1(\tilde \alpha,\tilde h)$ and the directional derivative with respect to $[\hat h-h_0]$ at $(\tilde\alpha,\tilde h)$ is denoted as
$$ \Gamma_2(\tilde\alpha,\tilde h)[\hat h- h_0] = \lim_{t\to 0} \frac{\Gamma(\tilde\alpha,\tilde h+t[\hat h- h_0])-\Gamma(\tilde\alpha,\tilde h)}{t}.$$
We assume that the estimated vector $\hat \beta$ and the estimated function $\hat z$ satisfy the following condition.

Steps 1-4 we use Condition IR(i-iii). In Steps 5 and 6 we will also use Condition IR(iv). For notational convenience we let $m_i=(y_i,d_i,x_i)$, and we use $\sup_{t\in \RR}|G(t)|\leq \bar L$, $\sup_{t\in \RR}|G'(t)|\leq \bar L'$ and $\sup_{t\in \RR}|G''(t)|\leq \bar L''$ where $\bar L\vee \bar L'\vee \bar L'' \leq C$.

Step 1. (Main Step for Normality) We have
$$
\begin{array}{rl}
\En[\psi_{\check\alpha,\hat h}(m_i)] & = \En[\psi_{\alpha_0,h_0}(m_i)]+\En[\psi_{\check\alpha,\hat h}(m_i)-\psi_{\alpha_0, h_0}(m_i)]\\
&  =  \En[\psi_{\alpha_0,h_0}(m_i)]+\Gamma(\check\alpha,\hat h) +n^{-1/2}\Gn(\psi_{\check\alpha,\hat h}-\psi_{\check\alpha,h_0}) +n^{-1/2}\Gn(\psi_{\check\alpha,h_0}-\psi_{\alpha_0,h_0})\\
& =: (I) + (II) + (III) + (IV).
\end{array}
$$
\noindent By Condition IR(iii), (\ref{Eq:HLhatalpha}), with probability at least $1-\Delta_n$ we have $|\En[\psi_{\check\alpha,\hat h}(m_i)]|\lesssim \delta_n n^{-1/2}$.

\noindent By Condition IR(iii), (\ref{Eq:HLestimated}), with probability at least $1-\Delta_n$ we have $|(III)| \lesssim \delta_n n^{-1/2}.$

To control $(IV)$ note that by Condition IR(i) $$|\psi_{\alpha,h_0}(m_i)-\psi_{\alpha_0,h_0}(m_i)| \leq |\G(d_i\alpha+x_i'\beta_0) - \G(d_i\alpha_0+x_i'\beta_0)|\  |z_{0i}| \leq \bar L'|\alpha-\alpha_0|\cdot |d_iz_{0i}|.$$ By Condition IR(iii), (\ref{Eq:HLhatalpha}), we have $|\check \alpha - \alpha_0|\leq \delta_n$ so that
$$\begin{array}{rl}
 \En[  \{\psi_{\check \alpha,h_0}(m_i)-\psi_{\alpha_0,h_0}(m_i)\}^2] & \leq \sup_{|\alpha-\alpha_0|\leq \delta_n} \En[  \{\psi_{\alpha,h_0}(m_i)-\psi_{\alpha_0,h_0}(m_i)\}^2] \\
 & \leq (\bar L'\delta_n)^2\En[d_i^2z_{0i}^2] \lesssim_P (\bar L'\delta_n)^2\barEp[d_i^2z_{0i}^2] \end{array}$$ from Markov inequality. By using Lemma \ref{Lemma:ProcessL1Control} with $W_i=d_iz_{0i}$ and $\mathcal{T}=\{\alpha-\alpha_0 \in \RR : |\alpha-\alpha_0|\leq \delta_n\}$,  we have
\begin{equation}\label{Eq:EP}
\begin{array}{rl}
\displaystyle |(IV)| & \displaystyle \lesssim_P \sup_{|\alpha-\alpha_0|\leq \delta_n}\left|n^{-1/2}\Gn(\psi_{\alpha,h_0}-\psi_{\alpha_0,h_0})\right| \\ & \displaystyle \lesssim_P n^{-1/2} \sup_{|\alpha-\alpha_0|\leq \delta_n}|\alpha-\alpha_0| \barEp[d_i^2z_{0i}^2]^{1/2}  \lesssim \delta_n n^{-1/2} \\
\end{array}
\end{equation}
By relation (\ref{EqBias}) in Step 2 below we have $$(II) = \Gamma(\check\alpha,\hat h) = - \barEp[w_id_iz_{0i}](\check\alpha-\alpha_0) + O_P(\delta_n n^{-1/2} + \delta_n|\check\alpha-\alpha_0|).$$

Therefore, combining the relations or $(II), (III), (IV)$ and $\En[\psi_{\check\alpha,\hat h}(m_i)]$ we have
$$ \barEp[w_id_iz_{0i}](\check\alpha-\alpha_0) = \En[\psi_{\alpha_0,h_0}(m_i)] + O_\Pr(\delta_nn^{-1/2}) + O_P(\delta_n)|\check\alpha-\alpha_0|$$ which establish the first assertion since $|\barEp[w_id_iz_{0i}]|\geq c>0$ is bounded away from zero. The second assertion  follows since $\barEp[\psi_{\alpha_0,h_0}(m_i)]=0$ and $\barEp[\sigma_i^3z_{0i}^3]\leq C$, by the Lyapunov CLT  we have $$(I) = \sqrt{n}\En[\psi_{\alpha_0,h_0}(m_i)] \rightsquigarrow N(0,\barEp[\sigma_i^2z_{0i}^2]).$$

Step 2. (Bounding $\Gamma(\alpha,\hat h)$ for $|\alpha-\alpha_0|\leq\delta_n$ which covers $(II)$)
We have
\begin{equation}\label{Eq:BiasIQR}
\begin{array}{rl}
\Gamma(\alpha,\hat h) & = \Gamma(\alpha,h_0) + \Gamma(\alpha,\hat h)- \Gamma(\alpha,h_0)\\
& = \Gamma(\alpha, h_0) + \{ \ \Gamma(\alpha,\hat h) - \Gamma(\alpha,h_0) - \Gamma_2(\alpha,h_0)[\hat h - h_0] \ \}  + \Gamma_2(\alpha,h_0)[\hat h - h_0]\\
\end{array}
\end{equation}
By (\ref{BoundGamma2third})  in Step 3 below we have
$$|  \Gamma(\alpha,\hat h) - \Gamma(\alpha,h_0) - \Gamma_2(\alpha,h_0)[\hat h - h_0] | \lesssim \delta_nn^{-1/2}.$$
By (\ref{BoundGamma2first}) in Step 3 below  we have with probability $1-o(1)$
$$ |\Gamma_2(\alpha,h_0)[\hat h - h_0]| \lesssim |\alpha-\alpha_0| \delta_n.$$

Finally, because $ \Gamma(\alpha_0,h_0) = 0$ and $\Gamma_1(\alpha_0,h_0) = -\barEp[ w_i d_i z_{0i}]$, by Taylor expansion there is some $\tilde \alpha \in [\alpha_0, \alpha]$ such that $$\begin{array}{rl}
\Gamma(\alpha,h_0) & = \Gamma(\alpha_0,h_0) + \Gamma_1(\tilde \alpha,h_0)( \alpha - \alpha_0) = \left\{ \Gamma_1(\alpha_0,h_0) + \eta_n \right\} ( \alpha - \alpha_0)\\
& = -\barEp[ w_i d_i z_{0i}]( \alpha - \alpha_0) + O( \delta_n |\alpha - \alpha_0|)  \end{array}$$
where $|\eta_n| \leq \bar L''\delta_n \barEp[|d_i^2z_{0i}|]\lesssim O(\delta_n)$ by relation (\ref{EqGamma1}) in Step 4.

Using the bounds above into relation (\ref{Eq:BiasIQR}) we have
\begin{equation}\label{EqBias}
\begin{array}{rl}
\Gamma(\alpha,\hat h) & =  -\barEp[ w_i d_i z_{0i}]( \alpha - \alpha_0) + O( \delta_n |\alpha - \alpha_0| + \delta_nn^{-1/2}) \\
\end{array}
\end{equation}

Step 3. (Relations for $\Gamma_2$) The directional derivative $\Gamma_2$ with respect the direction $\hat h - h_0$ at a point $\tilde h = (\tilde \beta,\tilde z)$ is given by $$\begin{array}{rl}
 \Gamma_2(\alpha, \tilde h)[\hat h - h_0]  &  = -\barEp[\G'(d_i\alpha+x_i'\tilde \beta)\tilde z_ix_i'\{\hat \beta - \beta_0\}]  + \barEp[ \{ \G(d_i\alpha_0+x_i'\beta_0) - \G(d_i\alpha +x_i'\tilde \beta)\} \{\hat z_i-z_{0i}\}].\\
 \end{array}$$  Note that when $\Gamma_2$ is evaluated at $(\alpha_0, h_0)$ we have \begin{equation}\label{EqGamma20}\Gamma_2(\alpha_0,h_0)[\hat h - h_0] = -\barEp[w_iz_{0i}x_i'](\widehat\beta-\beta_0) = 0\end{equation} because of the orthogonality condition $\barEp[w_iz_{0i} x_i]=0$ in Condition IR(ii). Therefore the expression for $\Gamma_2$ leads to the following bound
\begin{equation}\label{BoundGamma2first}\begin{array}{l}  \left| \Gamma_{2}(\alpha,h_0)[\hat h - h_0]\right| = \left| \Gamma_{2}(\alpha,h_0)[\hat h - h_0]\right. - \left.  \Gamma_{2}(\alpha_0, h_0)[\hat h - h_0]\right| \\
\leq \bar L' \barEp[|\alpha-\alpha_0| \ |d_i  z_{0i}| \ |x_i'\{\hat \beta - \beta_0\}|] + \bar L'\barEp[ |(\alpha-\alpha_0)d_i| \  |\hat z_i-z_{0i}|]\\
 \leq  |\alpha-\alpha_0| \bar L'\{ \barEp[ \{x_i'\{\hat\beta-\beta_0\}^2]\}^{1/2}\{\barEp[z_{0i}^2d_i^2]\}^{1/2}+|\alpha-\alpha_0|\bar L'  \{\barEp[(\hat z_i-z_{0i})^2]\}^{1/2}\{\barEp[d_i^2]\}^{1/2}\\
 \lesssim |\alpha-\alpha_0|\delta_n\\
   \end{array}\end{equation} since $\barEp[d_i^2] \lesssim C$, $\barEp[z_{0i}^2d_i^2]\leq C$,  $\{ \barEp[ \{x_i'\{\hat\beta-\beta_0\}^2]\}^{1/2}\lesssim \|\hat\beta-\beta_0\| \leq \delta_n n^{-1/4}$ and $\{\barEp[(\hat z_i-z_{0i})^2]\}^{1/2}\leq \delta_n$.

The second directional derivative $\Gamma_{22}$ at $\tilde h = (\tilde \beta, \tilde z)$ with respect to the direction $\hat h - h_0$ can be bounded by
\begin{equation}\label{BoundGamma2second}\begin{array}{l}
 \left|\Gamma_{22}(\alpha, \tilde h)[\hat h - h_0, \hat h - h_0]  \right| \\
 =  \left| -\barEp[\G''(x_i'\tilde \beta + \alpha d_i)\tilde z_i \{x_i'(\hat \beta_0 - \beta_0)\}^2]   - 2\barEp[ \G'(x_i'\tilde\beta+d_i\alpha)\{x_i'(\hat\beta-\beta_0)\}\{\hat z_i-z_i\}]\right|\\
\leq \bar L''\{\barEp[\tilde z_i^2]\}^{1/2} \{\barEp[\{x_i'(\hat\beta-\beta_0)\}^4]\}^{1/2} + 2\bar L' \{\barEp[\{x_i'(\hat\beta-\beta_0)\}^2]\}^{1/2}\{\barEp[(\hat z_i-z_{0i})^2]\}^{1/2}\\
\lesssim \{\barEp[\tilde z_i^2]\}^{1/2}\|\hat\beta-\beta_0\|^2+\|\hat\beta-\beta_0\|\{\barEp[(\hat z_i-z_{0i})^2]\}^{1/2}
 \end{array}\end{equation}
since $\barEp[\{x_i'\xi\}^4]\leq C\|\xi\|^4$. In turn, since $\tilde h \in [ h_0, \hat h]$, $|\tilde z(d_i,x_i)|\leq |z_{0}(d_i,x_i)|+|\hat z(d_i,x_i)-z_{0}(d_i,x_i)|$,  we have that $\{\barEp[\tilde z_i^2]\}^{1/2} \leq \{\barEp[z_{0i}^2]\}^{1/2}+\{\barEp[(\hat z_i-z_{0i})^2]\}^{1/2} \leq C+\delta_n$. Therefore, with probability $1-\Delta_n$
{\small \begin{equation}\label{BoundGamma2third}\begin{array}{rl}
 \left|\Gamma(\alpha,\hat h) - \Gamma(\alpha,h_0)  -  \Gamma_2(\alpha,h_0)\left[\hat h-h_0\right]\right| \leq \sup_{\tilde h \in [ h_0, \hat h ]} \left| \Gamma_{2,2}(\alpha, \tilde h)\left[\hat h - h_0, \hat h-h_0\right] \right| \lesssim \delta_n n^{-1/2}
 \end{array} \end{equation}} by Condition IR(iii).

Step 4. (Relations for $\Gamma_1$) By definition of $\Gamma$, its derivative with respect to $\alpha$ at $(\alpha,\tilde h)$ is
$$ \Gamma_1(\alpha,\tilde h) = -\barEp[ \G'(\alpha d_i+x_i'\tilde \beta)\tilde z_i d_i]. $$
Therefore, when the function above is evaluated at $\alpha=\alpha_0$ and $\tilde h = h_0=(\beta_0,z_0)$, since for $ \G'(x_i'\beta_0 + \alpha_0 d_i)=w_i$,  we have
\begin{equation}\label{EqGamma10} \Gamma_1(\alpha_0,h_0) = -\barEp[ w_i d_i z_{0i}].\end{equation} Moreover, $\Gamma_1$ also satisfies
 \begin{equation}\label{EqGamma1}
 \begin{array}{rl}
 \left| \Gamma_1(\alpha,h_0) - \Gamma_1(\alpha_0,h_0)\right| & = \left| \barEp[ \G'(\alpha d_i+x_i'\beta_0) z_{0i} d_i] - \barEp[ \G'(\alpha_0 d_i+x_i'\beta_0) z_{0i} d_i]\right|\\
 & \leq \bar L ' |\alpha-\alpha_0| \barEp[ | d_i^2z_{0i}|]. \end{array}\end{equation}

Step 5. (Estimation of Variance)
First note that
{\small \begin{equation}\label{Ed:FirstVarianceTerm}\begin{array}{rl}
 |\En[\hat w_id_i\hat z_i]-\barEp[w_id_iz_{0i}] | & = | \En[\hat w_id_i\hat z_i]-\En[w_id_iz_{0i}] |+| \En[ w_id_i z_{0i}]-\barEp[w_id_iz_{0i}] |\\
& \leq | \En[(\hat w_i-w_i)d_i\hat z_i]|+|\En[w_id_i(\hat z_i-z_{0i})] |+ |\En[ w_id_i z_{0i}]-\barEp[w_id_iz_{0i}]|\\
& \leq | \En[(\hat w_i-w_i)d_i(\hat z_i-z_{0i})]|+| \En[(\hat w_i-w_i)d_iz_{0i}]|\\
& +\|w_id_i\|_{2,n}\|\hat z_i-z_{0i}\|_{2,n} + |\ \En[ w_id_i z_{0i}]-\barEp[w_id_iz_{0i}] |\\
& \lesssim_P \|(\hat w_i-w_i)d_i\|_{2,n}\|\hat z_i-z_{0i}\|_{2,n}+\|\hat w_i-w_i\|_{2,n}\|d_iz_{0i}\|_{2,n}\\
&+\|w_id_i\|_{2,n}\|\hat z_i-z_{0i}\|_{2,n} + | \En[ w_id_i z_{0i}]-\barEp[w_id_iz_{0i}] |\\
& \lesssim_P \delta_n\end{array}\end{equation}}
because by Condition IR(ii) we have $\barEp[d_i^4] \leq C$, $\barEp[z_{0i}^4]\leq C$, by Condition IR(iv) $\|\hat z_i-z_{0i}\|_{2,n}+\|\hat w_i-w_i\|_{2,n}\lesssim \delta_n$ with probability $1-\Delta_n$.

%(\ref{Eq:HLestimated2})
%$$\begin{array}{rl}
% \|1\{y_i \leq x_i'\beta_0 + d_i \alpha_0\} - 1\{y_i \leq x_i'\hat\beta + d_i \alpha_0\}\|_{2,n} & \leq \|1\{ |y_i - x_i'\beta_0 - d_i \alpha_0|\leq 2|x_i'(\hat\beta -\beta_0)|\}\|_{2,n} \lesssim_P \delta_n^{1/2}. \\
% \end{array}
% $$

Next we proceed to control the other term of the variance. Since
$|\psi_{\check\alpha,\hat h}(m_i)-\psi_{\alpha_0,\hat h}(m_i)|\leq \bar L|d_i(\check\alpha-\alpha_0)\hat z_i|$ and $|\psi_{\alpha_0,\hat h}(m_i)-\psi_{\alpha_0,h_0}(m_i)| \leq \bar L'|\hat z_i-z_{0i}|+\bar L'|x_i'\{\hat\beta-\beta_0\}z_{0i}|$ we have with probability $1-\Delta_n$
\begin{equation}\label{Eq:Denominator}\begin{array}{rl}
& |\ \|\psi_{\check\alpha,\hat h}(y_i,d_i,x_i)\|_{2,n}  -  \| \psi_{\alpha_0,h_0}(y_i,d_i,x_i)\|_{2,n}  | \\
& \leq \bar L'\|d_i(\check\alpha-\alpha_0)\hat z_i\|_{2,n}+ \bar L\|\hat z_i-z_{0i}\|_{2,n}+\bar L'\|x_i'\{\hat\beta-\beta_0\}z_{0i}\|_{2,n}\\
& \lesssim \delta_n\end{array}\end{equation}
 by Condition IR(iv). Also, by Condition IR(ii), $\barEp[\sigma_i^3z_{0i}^3] \leq C$ we have $|\En[\psi_{\alpha_0,h_0}^2(m_i)] - \barEp[\psi_{\alpha_0,h_0}^2(m_i)]|\lesssim_P \delta_n$ .

Step 6. (Main Step for $\chi^2$) Note that the denominator of $L_n(\alpha_0)$ was analyzed in relation (\ref{Eq:Denominator}) of Step 5. Next consider the numerator of $L_n(\alpha_0)$. Since $\Gamma(\alpha_0,h_0)=\barEp[\psi_{\alpha_0, h_0}(m_i)]=0$, we have $$\En[ \psi_{\alpha_0,\hat h}(m_i)] = (\En-\barEp)[ \psi_{\alpha_0,\hat h}(m_i) - \psi_{\alpha_0, h_0}(m_i)] + \Gamma(\alpha_0,\hat h)+\En[\psi_{\alpha_0, h_0}(m_i)].$$
By Condition IR(iii) and (\ref{EqBias}) with $\alpha=\alpha_0$, it follows that
$$ |(\En-\barEp)[ \psi_{\alpha_0,\hat h}(m_i) - \psi_{\alpha_0, h_0}(m_i)]| \leq \delta_nn^{-1/2} \ \ \mbox{and} \ \ |\Gamma(\alpha_0,\hat h)| \lesssim_P \delta_n n^{-1/2}.$$
Therefore, using that $nA_n^2 = nB_n^2 + n(A_n-B_n)^2+2nB_n(A_n-B_n)$, for $A_n =\En[  \psi_{\alpha_0,\hat h}(m_i)]$  and  $B_n = \En[  \psi_{\alpha_0, h_0}(m_i)] \lesssim_P \{\barEp[\sigma_i^2z_{0i}^2]\}^{1/2}n^{-1/2}$ we have
\begin{eqnarray*}
 nL_n(\alpha_0) & = & \frac{n|\En[  \psi_{\alpha_0,\hat h}(m_i)]|^2}{\En[ \psi_{\alpha_0,\hat h}^2(m_i)]} = \frac{n|\En[  \psi_{\alpha_0,h_0}(m_i)]|^2+O_P(\delta_n)}{\barEp[\sigma_i^2z_{0i}^2]+O_P(\delta_n)}=    \frac{n|\En[ \psi_{\alpha_0, h_0}(m_i)]|^2}{\barEp[\sigma_i^2z_{0i}^2]} +O_P(\delta_n)
  \end{eqnarray*}
since $\barEp[\sigma_i^2z_{0i}^2]$ is bounded away from zero by assumption. The result then follows since $\sqrt{n}\En[  \psi_{\alpha_0,h_0}(m_i)]\rightsquigarrow N(0,\barEp[\sigma_i^2z_{0i}^2])$ and $\Ep[ \psi_{\alpha_0,h_0}^2(m_i)\mid x_i,d_i]=\sigma_i^2z_{0i}^2$.

\end{proof}

\section{Auxiliary Results for Penalized and Post-Model Selection Estimators}\label{Sec:AnalysisAux}

In this section we state relevant theoretical results on the performance of the $\ell_1$-penalized Logistic regression estimators, heteroscedastic Lasso with estimated weights estimators and the associated post-model selection estimators. The analysis of the latter builds upon the analysis of Lasso under heteroscedasticity of \cite{BellChenChernHans:nonGauss} and it was developed in \cite{BelloniChernozhukovKato2013b}. The analysis of the former builds upon the work of \cite{Bach2010} that established rates for $\ell_1$-penalized Logistic regression exploiting self-concordance. The main design condition relies on the restricted eigenvalue proposed in \cite{BickelRitovTsybakov2009}, namely for $\tilde x_i =(d_i, x_i' )'$
\begin{equation}\label{Def:RestrEig} \kappa_{\cc} = \inf_{\|\delta_{T^c}\|_1\leq \cc \|\delta_T\|_1} \| \mu_i\tilde x_i'\delta\|_{2,n}/\|\delta_T\|, \end{equation}
where $\cc = (c+1)/(c-1)$ for the slack constant $c>1$ and $\mu_i$ are problem specific weights. In the original setting of \cite{BickelRitovTsybakov2009} for least squares we have $\mu_i=1$ and it is well known that $\kappa_\cc$ is bounded away from zero if $\cc$ is bounded for any subset $T\subset \{1,\ldots,p\}$ with $|T|\leq s$ if the sparse eigenvalues of order $Cs$ are well behaved (bounded away from zero and from above uniformly) for suitably large constant $C$. When analyzing the logistic regression, the weights will be set to $\mu_i = \sqrt{w_i}$.

\subsection{Results for Lasso and Post Lasso with Estimated Weights}\label{Sec:EstLasso}

In this section we state results obtained in \cite{BelloniChernozhukovKato2013b} for Post-Lasso estimators with estimated weights, namely the model
\begin{equation}\label{Def:ModelEstLasso} f_id_i = f_ix_i'\theta_0 + v_i, \ \ \Ep[ f_iv_i x_i ] = 0\end{equation}
where we observe $\{(d_i,x_i):i=1,\ldots,n\}$, i.n.i.d., and only an estimate $\hat f_i$  of the weights $f_i$. The support $T_{\theta_0}=\supp(\theta_0)$ is unknown but a sparsity condition holds, namely $|T_{\theta_0}| \leq s$. Estimators for $\theta_0$ and $v_i$ can be computed based on Lasso or Post-Lasso, namely
\begin{equation}\label{EstLasso}\hat \theta \in \arg \min_{\theta \in \RR^p} \En[\hat f_i^2 ( d_i - x_i'\theta)^2] + \frac{\lambda}{n}\|\hat \Gamma\theta\|_1,  \ \ \mbox{and} \ \ \hat v_i = \hat f_i (d_i - x_i'\hat \theta ), \end{equation}
\begin{equation}\label{EstPostLasso} \widetilde\theta \in \arg\min_{\theta\in \RR^p} \ \left\{ \ \En[\hat f_i^2(d_i - x_i'\theta)^2] \ : \ \theta_j = 0, \ \mbox{if} \ \hat\theta_{j} = 0 \ \right\}, \ \ \mbox{and} \ \tilde v_i = \hat f_i (d_i - x_i'\widetilde\theta).\end{equation}
where $\lambda$ and $\hat \Gamma$ are the associated penalty level and loadings. We will use penalty loadings $\widehat\Gamma$ that are diagonal matrices defined by the algorithm below.
\begin{algorithm}[Computation of $\widehat \Gamma$]
\item[]Step 1. Compute the Post Lasso estimator $\widetilde \theta$ based on $\lambda = 2c'\sqrt{n}\Phi^{-1}(1-\gamma/2p), \ c'>c>1$ and the following penalty loadings, for $j=1,\ldots,p$ $$ \hat \Gamma_{jj}  =  \max_{i\leq n}\|\hat f_ix_i\|_\infty \sqrt{\En[ (\hat f_id_i- \overline{fd} )^2]}, \ \ \  \mbox{where} \ \  \overline{fd}:= \En[\hat f_id_i].$$
\item[]Step 2. Compute the residuals $\widehat v_i = \hat f_i(d_i - x_i'\widetilde\theta)$ and set $\widehat \Gamma$ as \begin{equation}\label{choice of loadings2} \hat \Gamma_{jj}  =   \sqrt{\En[\hat f_i^2 x^2_{ij} \hat v_i^2]}, \ j=1,\ldots,p. \end{equation}
\end{algorithm}

\cite{BellChenChernHans:nonGauss} established the validity of  using either of the choices in (\ref{choice of loadings2}) in the case the weights $f_i$ are known and equal to one and \cite{BelloniChernozhukovKato2013b} considers the current case with estimated weights $\hat f_i$. Next we provide sufficient high-level conditions to establish rates of convergence and sparsity bounds. As before the sequences $\Delta_n$ and $\delta_n$ go to zero, $C$ is constant independent of $n$.

{\bf Condition WL.} For the model (\ref{Def:ModelEstLasso}), suppose that \\
 (i) $\|\theta_0\|_0\leq s$ where $s\geq 1$;\\ % and the weights satisfy $0 < c \leq f_i \leq C$ uniformly in $n$ with probability $1-\Delta_n$,\\
(ii) $\min_{j\leq p}\barEp[|f_i x_{ij}v_i|^2]>c>0$, $\displaystyle\max_{j\leq p} \{\barEp[|f_ix_{ij}v_i|^3]\}^{1/3}\ \Phi^{-1}(1-\gamma/2p) \leq \delta_n n^{1/6}, \ \ \gamma \leq n^{-1/4} $\\ %$\displaystyle  \max_{j\leq p} \En[x_{ij}^4v_i^2] \vee \barEp[x_{ij}^4v_i^2]\lesssim_P 1 \ \ \mbox{and} \ \
(iii) $\displaystyle \max_{j\leq p} |(\En-\barEp)[f_i^2x_{ij}^2v_i^2]| \leq \delta_n$  with probability $1-\Delta_n$\\
%(v) $\displaystyle  \max_{i\leq n}\|x_i\|_\infty \left( \frac{1}{h\kappa_{\cc}}\sqrt{\frac{s^2\log(n\vee p)}{n}} + h \right) \to_P =0.$
(iv) the estimates $\hat f_i, i=1,\ldots,n$, satisfy with probability $1-\Delta_n$
$$\En[\hat f_i^2d_i^2]\leq C, \ \ \max_{j\leq p} \En[(\hat f_i - f_i)^2x_{ij}^2v_i^2] \leq \delta_n, \ \ \mbox{and} \ \ \  \En\left[ (\hat f_i^2 - f_i^2)^2v_i^2/f_i\right] \leq \hat c_f^2. $$

Condition WL(i) is a standard sparsity assumption and could be relaxed in different directions. Condition WL(ii) has mild moment conditions that are used to apply self-normalized moderate deviation theory to control heteroscedastic non-Gaussian errors similar to \cite{BellChenChernHans:nonGauss} where there are no estimated weights. Condition WL(ii) also has a condition on the dimension $p$ relative to $n$ and bounds how fast the confidence level $1-\gamma$ can converge to 1. Condition WL(iii) provides sufficient conditions for the uniform convergence of cross terms. Condition WL(iv) requires high-level rates of convergence for the estimate $\hat f_i$. In our applications these estimates can be constructed with $\ell_1$-penalized logistic regression estimators studied in Section \ref{Sec:Step1}.

\begin{remark}[Control of $\hat c_f$]
The quantity $\hat c_f$ impacts directly the prediction rate and sparsity results which are needed for the post-model selection estimators. Bounds on $\hat c_f$ will be dependent on regularities conditions. A simple bound is $ \hat c_f^2 \leq \En\left[ (\hat f_i^2 - f_i^2)^2\right] \max_{i\leq n}v_i^2/f_i$. In our analysis we pursued the use of matrix inequalities based on \cite{RudelsonVershynin2008} which seems to lead to sharper results under typical conditions.
\end{remark}

%In this work we consider the following choices for penalty loadings $\widehat \Gamma$ and penalty levels $\lambda$: for $j=1,\ldots,p$,
%\begin{equation}\label{choice of loadings2}\begin{array}{llll}
%&\text{initial}  &   \hat \Gamma_{jj}  =   \sqrt{\En[\hat f_i^2 x^2_{ij} (\hat f_id_i- \overline{fd} )^2]},  & \lambda = 2c'\sqrt{n}\Phi^{-1}(1-\gamma/2p), \ c'>c>1  \\
%& \text{refined} &   \hat \Gamma_{jj}  = \sqrt{\En [\hat f_i^2 x^2_{ij} \widehat v^2_i]},  & \lambda = 2c'\sqrt{n}\Phi^{-1}(1-\gamma/2p), \ c'>c>1
%\end{array}\end{equation}
%where $\overline{fd}:= \En[\hat f_id_i]$ and $\hat v_i$ is an estimate of $v_i$ based on Lasso/Post-Lasso with the initial option (or iterations). The algorithm below summarizes the estimation procedure.
%\begin{algorithm}[Computation of $\widehat \Gamma$]
%\item[]Step 1. Compute the Post Lasso estimator $\widetilde \theta$ based on $\widehat \Gamma$ set as the initial penalty loading (\ref{choice of loadings2}).
%\item[]Step 2. Compute the residuals $\widehat v_i = \hat f_i(d_i - x_i'\widetilde\theta)$ and set $\widehat \Gamma$ as the refined penalty loading (\ref{choice of loadings2}).
%\end{algorithm}

 Next we present results on the performance of the estimators generated by Lasso and Post-Lasso with estimated weights.

\begin{theorem}[Properties of Lasso and Post-Lasso with estimated Weights]\label{Thm:RateEstimatedLasso}
Under Condition WL, setting $\lambda \geq 2c'\sqrt{n}\Phi^{-1}(1-\gamma/2p)$ for $c'>c>1$, and using the penalty loadings $\hat \Gamma$ defined in (\ref{choice of loadings2}), there is an uniformly bounded $\cc$ such that with probability $1-o(1)$
$$ \|\hat f_i x_i'(\hat\theta - \theta_0)\|_{2,n} \lesssim  \frac{\lambda\sqrt{s}}{n\kappa_\cc\min_{i\leq n}\hat f_i/f_i}+ \hat c_f$$
provided that $\max_{i\leq n}\|\hat f_ix_i\|_\infty^2 \{\frac{\lambda\sqrt{s}}{n\kappa_\cc\min_{i\leq n}\hat f_i/f_i}+ \hat c_f\} \leq \delta_n$ with probability $1-o(1)$. Moreover, provided that $\semax{ \{s + n^2\hat c_f^2/\lambda^2\}/\delta_n} \leq C$, $\min_{i\leq n} \hat f_i^2 \geq c/2$, the data-dependent model $\widehat T_{\theta_0}$ selected by a Lasso estimator satisfies with probability $1-o(1)$
\begin{equation}\label{eq: sparsity bound}
\|\widetilde \theta \|_0 = | \widehat T_{\theta_0} | \lesssim s + \frac{n^2\hat c_f^2}{\lambda^2}
 \end{equation}
%and
% \begin{equation}\label{eq: model selection bound}
%\min_{\beta \in \Bbb{R}^p: \ \beta_j = 0 \  \forall j \not \in \widehat T} \sqrt{\En[ f(\tilde z_i) -  \tilde x_i'\beta]^2} \lesssim \sigma \sqrt{ \frac{ s \log (p \vee n)}{n} }.
% \end{equation}
and the Post-Lasso estimator obeys with probability $1-o(1)$
$$ \| x_i'(\widetilde \theta -\theta_0)\|_{2,n} \lesssim \frac{n\hat c_f}{\lambda}\sqrt{\frac{\log p}{n}} + \sqrt{\frac{s\log (p \vee n)}{n}} + \frac{\lambda\sqrt{s}}{n\kappa_\cc} \  \mbox{and} \ \|\widetilde \theta -\theta_0\|_1 \lesssim \left\{\sqrt{s}+\frac{n\hat c_f}{\lambda}\right\}\frac{\| x_i'(\widetilde \theta -\theta_0)\|_{2,n}}{\sqrt{\semin{|\hat T_{\theta_0}|}}}.$$\end{theorem}

Theorem \ref{Thm:RateEstimatedLasso} above establishes the rate of convergence for Lasso and Post-Lasso with estimated weights. This leads to bounds on the error between estimated the instrumented instrument $\hat z_i$ used in Table \ref{Table:Alg} with respect to the associated valid instrument $z_{0i} = v_i/\sqrt{w_i}$ since
\begin{equation}\label{Eq:Identityv}
\hat z_i - z_{0i} = d_i-x_i'\widetilde \theta - \frac{v_i}{\sqrt{w_i}} = d_i-x_i'\widetilde \theta - \{ d_i - x_i'\theta_0\} =  x_i'(\theta_0 - \widetilde \theta).
%\hat v_i - v_i = (\hat w_i -w_i) \frac{v_i}{w_i} + \hat w_i x_i'(\theta_0 - \hat\theta).
\end{equation}
Sparsity properties of the Lasso estimator $\hat\theta$ under estimated weights follows similarly to the standard Lasso analysis derived in \cite{BellChenChernHans:nonGauss}. By combining such sparsity properties and the rates in the prediction norm we can establish rates for the post-model selection estimator under estimated weights.

%\begin{remark}[Penalty Choice in Step 2] In Step 2 of the proposed method in Table \ref{Table:Alg} we have $\rho_n = n^{-1/3}$ so that setting $\lambda _2 = n^{2/3}$ we have $2c'\sqrt{n}\Phi^{-1}(1-\gamma/2p)\|\hat\Gamma\|_\infty \leq \delta_n \lambda_2$. In this case the penalty is dominated by the potential bias in the estimation of the weights. Thus setting the loadings to $1$ with $\lambda_2=n^{2/3}$ is asymptotic valid.
%\end{remark}

\subsection{$\ell_1$-Penalized Logistic Regression}\label{Sec:Step1}

Consider a data generating process such that $$\Ep [ y_i \mid \tilde x_i ] =  \G(\tilde x_i'\eta_0)$$ which is independent across $i$ ($i=1,\ldots,n$). Without loss of generality, we assume that $\|\eta_0\|_0=s\geq 1$, $\En[\tilde x_{ij}^2]=1$ for all  $1 \leq j \leq p$. First we consider the estimation of $\eta_0$ via $\ell_1$-penalized Logistic regression
\begin{equation}
\label{Eq:L1Log}
 \hat \eta \in \arg\min_{\eta} \Lambda(\eta) + \frac{\lambda}{n}\|\eta\|_1. \end{equation}
Following a general principle used in $\ell_1$-penalized estimators as discussed in \cite{BickelRitovTsybakov2009,Bach2010,BC-SparseQR,Unified2012,Wang2012}, under the event that
\begin{equation}
\label{Eq:PenaltyL1Log}
\frac{\lambda}{n} \geq c\|\nabla\Lambda(\eta_0)\|_\infty = c \| \En[ \{y_i - \G(\tilde x_i'\eta_0)\}\tilde x_i]\|_\infty, \ \ \ \mbox{where} \ c > 1,
\end{equation}
the estimator in (\ref{Eq:L1Log}) achieves good theoretical guarantees under mild design conditions. Although $\eta_0$ is unknown, we can set $\lambda$ so that the event in (\ref{Eq:PenaltyL1Log}) holds with high probability. In particular, Remark \ref{Remark:Penalty} based on Lemma \ref{Lemma:PenaltySNDT} shows that it suffices to set $\lambda = \frac{1.1}{2}\sqrt{n}\Phi^{-1}(1-\gamma/[2p])$ where we suggest $\gamma = 0.1/\log n$. Next we present results for the estimator (\ref{Eq:L1Log}). In what follows we consider (\ref{Def:RestrEig}) with $f_i=\sqrt{w_i}$.

\begin{lemma}[Results for $\ell_1$-Penalized Logistic Regression]\label{Lemma:LassoLogisticRate}
Assume $\lambda/n \geq c \|\nabla \Lambda(\eta_0) \|_\infty $, $c>1$ and let $\cc = (c+1)/(c-1)$. Then
$$ \|\sqrt{w_i}\tilde x_i'(\hat\eta-\eta_0)\|_{2,n} \leq 3(1+\mbox{$\frac{1}{c}$})\frac{\lambda\sqrt{s}}{n\kappa_\cc} \ \ \ \mbox{and} \ \ \ \|\hat \eta -\eta_0\|_1 \leq 3\frac{(1+c)(1+\cc)}{c}\frac{\lambda s}{n\kappa_\cc^2}$$
provided that $\inf_{\delta \in \Delta_\cc} \frac{\|\sqrt{w_i}\tilde x_i'\delta\|_{2,n}^3}{\|\sqrt{w_i}|\tilde x_i'\delta|^{3/2}\|_{2,n}^2} > 3(1+\frac{1}{c})\frac{\lambda\sqrt{s}}{n\kappa_\cc}$. Moreover,  we have
$$ |\supp(\hat\eta)|\leq 36s\cc^2\frac{\min_{m\in\mathcal{M}}\semax{m}}{\kappa_\cc^2} \ \ \ \ \mbox{and} \ \ \ \Lambda(\hat\eta)-\Lambda(\eta_0) \leq 3(1+\mbox{$\frac{1}{c}$})\left(\frac{\lambda\sqrt{s}}{n\kappa_\cc}\right)^2$$where $\mathcal{M} = \{ m \in \NN : m > 72\cc^2 s \semax{m}/\kappa_\cc^2 \}$, provided $\max_{i\leq n}\|\tilde x_i\|_\infty\|\hat\eta-\eta_0\|_1\leq 1 $.\end{lemma}

The extra growth condition required for identification is mild. For instance we typically have $\lambda \lesssim \sqrt{\log (n\vee p) / n}$ and, if the weights $w_i$ are bounded away from zero, for many designs of interest we have $\inf_{\delta\in\Delta_\cc} \|\tilde x_i'\delta\|_{2,n}^{3}/\En[|\tilde x_i'\delta|^3]$  bounded away from zero (see \cite{BC-SparseQR}). For more general designs and weights we have $$\inf_{\delta\in\Delta_\cc} \frac{\|\sqrt{w_i}\tilde x_i'\delta\|_{2,n}^{3}}{\En[w_i|\tilde x_i'\delta|^3]}\geq \inf_{\delta\in\Delta_\cc} \frac{\|\sqrt{w_i}\tilde x_i'\delta\|_{2,n}}{\max_{i\leq n}\|\tilde x_i\|_\infty \|\delta\|_1} \geq \frac{\kappa_\cc}{\sqrt{s}(1+\cc)\max_{i\leq n}\|\tilde x_i\|_\infty}$$ which implies the extra growth condition under $K_1^2 s^2 \log (p\vee n) \leq \delta_n \kappa_\cc^2 n$. Under the condition that $s \ell_n$-sparse eigenvalues are bounded away from zero and from above for some $\ell_n\to\infty$, it follows that $Cs$ belongs to $\mathcal{M}$ for $n$ large enough so that $|\supp(\hat\eta)|\lesssim s$ under the conditions above.

In order to alleviate the bias introduced by the $\ell_1$-penalty, we can consider the associated post-model selection estimates. Let $\widehat T^*$ denote a subset of covariates (selected arbitrarily) and define the associated post-model selection estimator
 \begin{equation}
 \label{def:App:postl1qr}
\widetilde \eta \in \arg\min_{\eta} \left\{ \Lambda(\eta)  : \eta_{j} = 0 \ \text{if} \ j \notin \widehat T^* \right\}.
 \end{equation} Typically $\widehat T^*$ can be taken as $\supp(\hat\eta)$. However, we can add additional variables through other procedures. (For example, in Step 1 we always include the treatment $d_i$; in Step 3  of the double selection procedure covariates selected in a different equation are included.) The following result characterizes the performance of the estimator in (\ref{def:App:postl1qr}).
\begin{lemma}[Estimation Error of Post-$\ell_1$-penalized Logistic Regression]\label{Lemma:PostLassoLogisticRate} Let $\widehat s^*= |\widehat T^*|$.  We have
 $$ \|\sqrt{w_i}\tilde x_i'(\tilde \eta - \eta_0)\|_{2,n} \leq \frac{3\sqrt{\hat s^*}\|\nabla \Lambda(\eta_0)\|_\infty}{\sqrt{\semin{\hat s+s}}} + 3\sqrt{\max\{0,\Lambda(\widetilde \eta) - \Lambda(\eta_0)\}} $$
 provided that $$\inf_{\|\delta\|_0\leq \hat s^*+s}\frac{\|\sqrt{w_i}\tilde x_i'\delta\|_{2,n}^3}{\|\sqrt{w_i}|\tilde x_i'\delta|^{3/2}\|_{2,n}^2} >6 \max\left\{\sqrt{\hat s^*+s}\frac{\|\nabla \Lambda(\eta_0)\|_\infty}{\sqrt{\semin{\hat s^*+s}}}, \ \ \sqrt{\max\{0,\Lambda(\widetilde \eta) - \Lambda(\eta_0)\}}\right\}.$$
\end{lemma}

Lemma \ref{Lemma:PostLassoLogisticRate} provides the rate of convergence in the prediction norm for the post model selection estimator despite the possible imperfect model selection. The rates rely on the overall quality of the selected model and the overall number of components $\hat s^*$. Once again, based on the results in Lemma \ref{Lemma:LassoLogisticRate}, the extra growth condition required for identification is mild provided that $\supp(\hat\eta)\subset \widehat T^*$ and $\hat s^*$ is not much larger than $s$.

\begin{remark}\label{Comment:Alpha}
In Step 1 of the algorithms, we use $\ell_1$-penalized Logistic regression with $\tilde x_i = (d_i,x_i')'$, $\hat \delta :=\hat \eta-\eta_0= (\hat\alpha-\alpha_0,\hat \beta'-\beta_0')'$,  and we are interested on rates for $\| x_i'(\hat \beta - \beta_0)\|_{2,n}$ instead of $\|\tilde x_i'\hat \delta\|_{2,n}$. However, it follows that
$$\| x_i'(\hat \beta - \beta_0)\|_{2,n} \leq \|\tilde x_i'\hat\delta\|_{2,n} + |\hat \alpha - \alpha_0| \cdot \|d_i\|_{2,n}.$$
Since $s\geq 1$, without loss of generality we can assume the component associated with the treatment $d_i$ belongs to $T$ (at the cost of increasing the cardinality of $T$ by one which will not affect the rate of convergence). Therefore we have that
$$ |\hat \alpha - \alpha_0| \leq \|\hat \delta_T\| \leq \|\sqrt{w_i}\tilde x_i'\hat\delta\|_{2,n}/\kappa_\cc.$$%\|d_i(\hat\alpha_\tau-\alpha_0) + x_i'(\hat\beta_\tau - \beta_\tau)\|_{2,n}/\kappa_\cc.$$
In most applications of interest $\|d_i\|_{2,n}$ and $1/\kappa_\cc$ are bounded from above with high probability. Similarly, in Step 1 of Algorithm 1 we have that the Post-$\ell_1$-Logistic estimator satisfies
$$\| x_i'(\widetilde \beta - \beta_0)\|_{2,n} \leq  \|\tilde x_i'\widetilde\delta \|_{2,n}\left( 1 + \|d_i\|_{2,n}/\sqrt{\semin{\hat s+s}}\right).$$
\end{remark}

\section{Auxiliary Inequalities}\label{Sec:AuxiliaryInequalities}

\begin{lemma}\label{cor:LoadingConvergence}
Let $X_i\in \RR^p $ be independent random variables and let $K=\Ep[\max_{i\leq n}\|X_i\|_\infty^k]$ for some $k\geq 1$.  Then we have
$$ \Ep[ \max_{1\leq j\leq p} |\En[ |X_{ij}|^k - \barEp[|X_{ij}|^k]]| \lesssim \frac{K \log p}{n} + \sqrt{\frac{K\log p}{n}\max_{j\leq p}\frac{1}{n}\sum_{i=1}^n\Ep[|X_{ij}|^k]}.$$
\end{lemma}

\begin{lemma}[Essentially in Theorem 3.6 of \cite{RudelsonVershynin2008}]\label{thm:RV34}
Let $X_i$, $i=1,\ldots,n$, be independent random vectors in $\RR^p$ be such that $\sqrt{\Ep[ \max_{1\leq i\leq n}\|X_i\|_\infty^2]} \leq K$. Let $$\delta_n:= 2\left( \bar C K \sqrt{k} \log(1+k) \sqrt{\log (p\vee n)} \sqrt{\log n}  \right)/\sqrt{n},$$ where $\bar C$ is the universal constant. Then,
$$ \Ep\left[ \sup_{\|\theta\|_0\leq k, \|\theta\| =1} \left| \En\[ (\theta'X_i)^2 - \Ep[(\theta'X_i)^2] \]\right|\right] \leq \delta_n^2 + \delta_n \sup_{\|\theta\|_0\leq k, \|\theta\| =1} \sqrt{\frac{1}{n}\sum_{i=1}^n\Ep[(\theta'X_i)^2]}. $$
\end{lemma}

%\begin{lemma}\label{cor:LoadingConvergence}
%Fix arbitrary vectors $x_1,\ldots,x_n \in \RR^p$ with $\max_{i\leq n}\|x_i\|_\infty \leq K_{x}$. Let $\zeta_i \ (i=1,\ldots,n)$ be independent random variables such that $\barEp[|\zeta_i|^q]<\infty$ for some $q\geq 4$. Then we have with probability $1-8\tau$
%$$ \max_{1\leq j\leq p} |(\En-\barEp)[x_{ij}^2\zeta_i^2]| \leq 4\sqrt{\frac{\log(2p/\tau)}{n}}K_{x}^2 (\barEp[|\zeta_i|^q]/\tau)^{4/q}$$
%\end{lemma}
%\begin{proof} The result is derived in Lemma 2 of \cite{BelloniChernozhukovKato2013a} which follows from a maximal inequality derived in \cite{BelloniChernozhukovHansen2011}. \end{proof}

%\begin{lemma}[Essentially in Theorem 3.6 of \cite{RudelsonVershynin2008}]\label{thm:RV34}
%Let $z_i$, $i=1,\ldots,n$, be independent random vectors in $\RR^p$ be such that $\sqrt{\Ep[ \max_{1\leq i\leq n}\|z_i\|_\infty^2]} \leq K_z$. Let $\delta_n:= 2\left( \bar C K_z \sqrt{k} \log(1+k) \sqrt{\log (p\vee n)} \sqrt{\log n}  \right)/\sqrt{n}$, where $\bar C$ is the universal constant. Then,
%$$ \Ep\left[ \sup_{\|\alpha\|_0\leq k, \|\alpha\| =1} \left| \En\[ (\alpha'z_i)^2 - \Ep[(\alpha'z_i)^2] \]\right|\right] \leq \delta_n^2 + \delta_n \sup_{\|\alpha\|_0\leq k, \|\alpha\| =1} \sqrt{\barEp[(\alpha'z_i)^2]}. $$
%\end{lemma}

Consider an empirical process $\mathbb{G}_n(f) = n^{-1/2} \sum_{i=1}^n \{f(Z_i) - \Ep[f(Z_i)]\}$ indexed by $\mathcal{F}$, a class of pointwise measurable functions (see \cite{vdV-W} Chapter 2.3) and assume that $0 \in \mathcal{F}$. The random empirical measure for an underlying independent data sequence $\{Z_i,i=1,\ldots,n\}$ is denoted by $\mathbb{P}_n$.

\begin{lemma}\label{Lemma:ProcessL1Control}
For the random process $h_i$ indexed by $\mathcal{T}\subset \RR^{\tilde p}$ and random vector $W_i \in \RR^{\tilde p}$, independent across $i=1,\ldots,n$, let $|h_i(t)|\leq |t'W_i|$, $\bar\sigma^2 := \sup_{t \in \mathcal{T}} \frac{1}{n}\sum_{i=1}^n\Ep[h_i(t)^2]$,  and $\|\mathcal{T}\|_1 = \sup_{t \in \mathcal{T}} \|t\|_1$. Provided that $K^2\|\mathcal{T}\|_1M/4 \geq \bar\sigma^2$, we have
$$ \Ep\left[ \sup_{t \in \mathcal{T}}|\En[ h_i(t)-\Ep[h_i(t)] ]|\right] \leq 4\|\mathcal{T}\|_1 \Ep\left[ \|\En[ \varepsilon_i W_i ]\|_\infty\right] \ \ \mbox{and} $$
$$P\left( \sup_{t \in \mathcal{T}}|\En[\{h_i(t)-\Ep[h_i(t)]\}]|> \frac{K \|\mathcal{T}\|_1 \sqrt{M}}{\sqrt{n}}\right)  \leq   32\tilde p\exp\left(\frac{-K^2}{ 4(16)^2}\right)
 + P\left( \max_{j\leq \tilde p}\En[W_{ij}^2] > M \right).
$$ where $\epsilon_i$, are independent Radamacher random variables.
\end{lemma}
\begin{proof}
To establish the first relation, by symmetrization for expectation Lemma 6.3 in \cite{LedouxTalagrandBook}
 $$ \Ep\left[ \sup_{t \in \mathcal{T}}|(\En-\barEp)[ h_i(t) ]|\right] \leq 2\Ep\left[ \sup_{t \in \mathcal{T}}|\En[\varepsilon_i h_i(t) ]|\right] $$ and Contraction principle Lemma 4.12 in \cite{LedouxTalagrandBook} we have
$$ \Ep\left[ \sup_{t \in \mathcal{T}}|(\En-\barEp)[ h_i(t) ]|\right] \leq 4\Ep\left[ \sup_{t \in \mathcal{T}}|\En[ \varepsilon_i t'W_i ]|\right] \leq 4\sup_{t \in \mathcal{T}}\|t\|_1 \Ep\left[ \|\En[ \varepsilon_i W_i ]\|_\infty\right].$$

To establish the second relation, let $\widetilde K = K\|\mathcal{T}\|_1\sqrt{M}/2$.
By Lemma 2.3.7 in \cite{vdVaartWellner2007}, symmetrization for probabilities, we have
$$ P\left( \sup_{t \in \mathcal{T}}|\Gn( h_i(t) )| > \widetilde K \right) \leq \frac{2}{1-(\bar\sigma^2/\widetilde K^2)}P\left( \sup_{t \in \mathcal{T}}|\Gn( \varepsilon_i h_i(t) )| > \widetilde K/4 \right) \leq 4P\left( \sup_{t \in \mathcal{T}}|\Gn( \varepsilon_i h_i(t) )| > \widetilde K/4 \right)$$ since ${\rm var}(\Gn( h_i(t)))\leq \barEp[h_i(t)^2] \leq \bar\sigma^2\leq \widetilde K^2/4$.
Moreover, letting $\mathcal{W}_M = \{ \max_{j\leq \tilde p}\En[W_{ij}^2] \leq M \}$ we have
$$ P\left( \sup_{t \in \mathcal{T}}|\Gn( \varepsilon_i h_i(t) )| > \widetilde K/4 \right) \leq  P\left( \sup_{t \in \mathcal{T}}|\Gn( \varepsilon_i h_i(t) )| > \widetilde K/4 \mid \mathcal{W}_M \right) + P(\mathcal{W}_M^c)$$
Conditional on $\{W_i\}_{i=1}^n$, also using Contraction principle Lemma 4.12 in \cite{LedouxTalagrandBook} we have
$$ \begin{array}{rl}
\Ep[ \exp( \psi \sup_{t \in \mathcal{T}}|\Gn( \varepsilon_i h_i(t) )| ) ] & \leq \Ep[ \exp( 4\psi \sup_{t \in \mathcal{T}}\|t\|_1\|\Gn( \varepsilon_i W_i )\|_\infty ) ] \\
& \leq \tilde p \cdot {\displaystyle \max_{j\leq \tilde p)}} \Ep[ \exp\{ 4\psi \sup_{t \in \mathcal{T}}\|t\|_1|\Gn( \varepsilon_i W_{ij} ) | \} ] \\
& \leq 2\tilde p \cdot {\displaystyle \exp( 8\psi^2 \sup_{t \in \mathcal{T}}\|t\|_1^2 \max_{j\leq \tilde p}}\En[W_{ij}^2] ) \\
\end{array}$$ Since we have that $P(X > \frac{1}{4}\widetilde K)\leq \min_{\psi\geq 0} \exp(-\psi \frac{1}{4}\widetilde K)\Ep[\exp(\psi X)]$, by choosing the parameter $\psi$ as $\psi = \frac{1}{4}\widetilde K/\{ 16 \|\mathcal{T}\|_1^2 \max_{j\leq \tilde p}\En[W_{ij}^2] \}$ it follows
$$P_\varepsilon\left( \sup_{t \in \mathcal{T}}|\Gn( \varepsilon_i h_i(t) )| > \widetilde K/4 \mid h_i,W_i \right) \leq 8\tilde p\exp( -\widetilde K^2/\{(16)^2\sup_{t \in \mathcal{T}}\|t\|_1^2 \max_{j\leq \tilde p}\En[W_{ij}^2]\}) $$

The result follows by taking the expectation over $(h_i,W_i)\in \mathcal{W}_M=\{ \max_{j\leq \tilde p}\En[W_{ij}^2] \leq M \}$.
\end{proof}

%\begin{lemma}\label{Lemma:MaxIneq2} Suppose that for all  $0 < \varepsilon \leq \varepsilon_0$
%\begin{equation}\label{Eq:J}
% N(\varepsilon, \mathcal{F}, \mathbb{P}_n) \leq ( \omega / \varepsilon )^{m} \text{ and  }   N(\varepsilon, \mathcal{F}^2, \mathbb{P}_n) \leq ( \omega / \varepsilon )^{m},
%\end{equation} for some $\omega$ which can grow with $n$. Then, as $n$ grows we have
%{\small $$
% \sup_{f \in \mathcal{F}} | \Gn (f) | \lesssim_P  \sqrt{m \log(\omega \vee n)} \(\sup_{f \in
% \mathcal{F}} \barEp[f^2]   +  \sqrt{\frac{m \log (n\vee \omega)}{n}}
%  \( \sup_{f \in \mathcal{F}} \En [f^4]\vee\barEp[f^4] \)^{1/2} \)^{1/2}.
%$$}
%\end{lemma}
%\begin{proof}
%The result is derived in \cite{Belloni:Chern:Fern}.
%\end{proof}

\bibliographystyle{plain}
\bibliography{mybibSparseLogistic}

%\cleardoublepage
\pagebreak

\setcounter{page}{1}

\begin{center}
{\sc \LARGE Supplementary Appendix for \\ ``Post-Selection Inference for Generalized Linear Models with Many Controls"}\end{center}

\vspace{0.5cm}

\section{Technical Results and Proofs for Logistic Regression}\label{Sec:ResultsLogistic}

In this section our goal is to establish sparsity and rates of convergence of the Post-Lasso Logistic estimator. Both of these properties  require us to also revisit the analysis of the $\ell_1$-penalize logistic regression (Lasso-Logistic) estimator.  In what follows we use a more compact notation, specifically $\eta = (\alpha,\beta)$, $\tilde x_i = (d_i,x_i')'$, $\eta_0 = (\alpha_0,\beta_0')'$. Thus the Lasso-Logistic estimator is defined as any vector $\hat\eta$ such that
\begin{equation}\label{Eq:LassoLogstic}
\hat \eta \in \arg\min_{\eta} \Lambda(\eta)+\frac{\lambda}{n}\|\eta\|_1.\end{equation}
We will also consider the post-model selection Logistic estimator associated with a support $\hat T^* \subset \{ 1,\ldots,p\}$ defined as
\begin{equation}\label{Eq:PostLassoLogstic}
\widetilde \eta \in \arg\min_{\eta} \Lambda(\eta) \ \ : \ \ \supp(\eta)\subseteq \hat T^*.\end{equation}

\subsection{Design conditions and Relations}
Next we collect relevant quantities associated with the design matrix $\En[\tilde x_i\tilde x_i']$ and the weighted counterpart $\En[w_i\tilde x_i\tilde x_i']$ where $w_i = \G_i(1-\G_i) \in [0,1]$, $\G_i= \G(\tilde x_i'\eta_0)$, $i=1,\ldots,n$, is the conditional variance of the outcome variable $y_i$. The non-weighted quantities are well studied in the literature (namely restricted eigenvalue, minimum and maximal sparse eigenvalues).

\begin{definition} For $T=\supp(\eta_0)$, $|T|\geq 1$, the (logistic) restricted eigenvalue is defined as $$ \kappa_\cc := \min_{\|\delta_{T^c}\|_1 \leq \cc\|\delta_T\|_1} \frac{\|\sqrt{w_i}\tilde x_i'\delta\|_{2,n}}{\|\delta_T\|}$$
\end{definition}

%\begin{definition} We define the following (weighted) minimum and maximum sparse eigenvalues $$ \semin{m} := \min_{1\leq \|\delta\|_0 \leq m} \frac{\|\sqrt{w_i}\tilde x_i'\delta\|_{2,n}^2}{\|\delta\|} \ \ \ \mbox{and} \ \ \ \semax{m} := \max_{1\leq \|\delta\|_0 \leq m} \frac{\|\tilde x_i'\delta\|_{2,n}^2}{\|\delta\|}$$ \end{definition}

\begin{definition} For a subset $A\subset \RR^p$ let the non-linear impact coefficient be defined as
$$ \bar q_A = \inf_{ \delta \in A} \En\[w_i|\tilde x_i'\delta|^2\]^{3/2}/ \ \En\[w_i|\tilde x_i'\delta|^3\].$$
In this work we will apply this for $A = \Delta_\cc$ and $A = \{ \delta \in \RR^p : \|\delta\|_0 \leq Cs \}$.
\end{definition}

The definitions above differ from their counterpart in the analysis of $\ell_1$-penalized least squares estimators by the weighting $0\leq w_i\leq 1$. Thus it will be relevant to understand their relations through the quantities $$\psi_{(r)}(\cc) := \min_{\|\delta_{T^c}\|_1 \leq \cc\|\delta_T\|_1} \frac{\|\sqrt{w_i}\tilde x_i'\delta\|_{2,n}}{\|\tilde x_i'\delta\|_{2,n}} \ \ \mbox{and} \ \ \psi_{(s)}(m) := \min_{1\leq \|\delta\|_0 \leq m} \frac{\|\sqrt{w_i}\tilde x_i'\delta\|_{2,n}}{\|\tilde x_i'\delta\|_{2,n}}$$

Lemma \ref{Lemma:RelatingWeighted} provides three relationships between the weighted versions and the non-weighted versions. Neither dominates the other. Most papers in the literature focus on the first pair of relations which entails assuming that $\min_{i\leq n} w_i$ is bounded away from zero uniformly in $n$. The second and third pairs of relations  allow for better control in the presence of a few small weights. The second pair states that if the average harmonic mean of the weights is bounded the ratio between the weighted and non-weighted quantities is controlled by the intrinsic sparsity.

\begin{lemma}[Relating weighted and non-weighted design quantities]\label{Lemma:RelatingWeighted}
Letting $w_i = \G_i(1-\G_i)$ we have the following inequalities $\psi_{(r)}(\cc) \geq \min_{i\leq n} \sqrt{w_i} \ \ \ \mbox{and} \ \ \ \psi_{(s)}(m) \geq  \min_{i\leq n} \sqrt{w_i}; $
$$\psi_{(r)}(\cc) \geq \frac{\kappa_\cc^u\{\En[1/w_i]\}^{-1/2}}{\sqrt{s}(1+\cc)\max_{i\leq n}\|\tilde x_i\|_\infty} \ \ \ \mbox{and} \ \ \ \psi_{(s)}(m) \geq \frac{\sqrt{\semin{m}}\{\En[1/w_i]\}^{-1/2}}{\sqrt{m}\max_{i\leq n}\|\tilde x_i\|_\infty}. $$
where $\kappa_\cc^u$ is the original (non-weighted) restricted eigenvalue. Moreover, for any $\epsilon \in (0,1]$ we have
$$\psi_{(r)}(\cc) \geq \sqrt{\epsilon}\kappa_\cc^u \left\{ 1- \En[ 1\{w_i\leq \epsilon\}]\frac{s(1+\cc)^2 \max_{i\leq n}\|\tilde x_i\|_\infty^2}{\kappa_\cc^{u2}} \right\}^{1/2} \ \ \ \mbox{and} $$ $$ \psi_{(s)}(m) \geq \sqrt{\epsilon}\sqrt{\semin{m}}\left\{ 1- \En[ 1\{w_i\leq \epsilon\}]\frac{m \max_{i\leq n}\|\tilde x_i\|_\infty^2}{\semin{m}}\right\}^{1/2}. $$
\end{lemma}
\begin{proof}
The first pair of bounds is trivial since $w_i \geq 0$. To show the second pair we have
$$
\begin{array}{rl}
\En[|\tilde x_i'\delta|^2] & = \En[\sqrt{w_i}|\tilde x_i'\delta|\cdot|\tilde x_i'\delta|/\sqrt{w_i}]\\
& \leq \{\En[w_i|\tilde x_i'\delta|^2]\}^{1/2}\cdot\{\En[|\tilde x_i'\delta|^2/w_i]\}^{1/2}\\
& \leq \{\En[w_i|\tilde x_i'\delta|^2]\}^{1/2}\cdot\{\En[1/w_i]\}^{1/2}\|\delta\|_1\max_{i\leq n}\|\tilde x_i\|_\infty\\
\end{array}
$$ Therefore, for $\vartheta_\delta = \|\tilde x_i'\delta\|_{2,n}/\|\delta\|_1$  we have
$$ \begin{array}{rl}
\frac{\|\sqrt{w_i}\tilde x_i'\delta\|_{2,n}}{\|\tilde x_i'\delta\|_{2,n}} & \geq \frac{\|\sqrt{w_i}\tilde x_i'\delta\|_{2,n}}{\|\sqrt{w_i}\tilde x_i'\delta\|_{2,n}^{1/2}\cdot\{\En[1/w_i]\}^{1/4}\|\delta\|_1^{1/2}\max_{i\leq n}\|\tilde x_i\|_\infty^{1/2}} \\
& = \frac{\|\sqrt{w_i}\tilde x_i'\delta\|_{2,n}^{1/2}}{\|\tilde x_i'\delta\|_{2,n}^{1/2}}\frac{\vartheta_\delta^{1/2}}{\max_{i\leq n}\|\tilde x_i\|_\infty^{1/2}}\frac{1}{\{\En[1/w_i]\}^{1/4}} \\
\end{array}$$
By cancelling out $\|\sqrt{w_i}\tilde x_i'\delta\|_{2,n}^{1/2}/\|\tilde x_i'\delta\|_{2,n}^{1/2}$ and squaring both sides we have
$$ \begin{array}{rl}
\frac{\|\sqrt{w_i}\tilde x_i'\delta\|_{2,n}}{\|\tilde x_i'\delta\|_{2,n}} & \geq \vartheta_{\delta} / \max_{i\leq n}\|\tilde x_i\|_\infty.\end{array}$$
The result follows by noting that for $\delta \in \Delta_\cc$ we have $\vartheta_\delta \geq \kappa_\cc^u / \{ (1+\cc)\sqrt{s}\}$ and for any non-zero $\delta$ with $\|\delta\|_0\leq m$ we have $\vartheta_\delta \geq \sqrt{\semin{m}}/\sqrt{m}$.

The third pair follows from noting that
$$\En[w_i|\tilde x_i'\delta|^2] = \En[w_i1\{w_i> \epsilon\}|\tilde x_i'\delta|^2] + \En[w_i1\{w_i \leq \epsilon\}|\tilde x_i'\delta|^2] \geq \epsilon \En[|\tilde x_i'\delta|^2]-\epsilon \En[1\{w_i \leq \epsilon\}|\tilde x_i'\delta|^2]$$
Moreover, by definition of $\vartheta_\delta$ we have $$\En[1\{w_i \leq \epsilon\}|\tilde x_i'\delta|^2] \leq \En[1\{w_i \leq \epsilon\}]\max_{i\leq n}\|\tilde x_i\|_\infty^2\|\delta\|_1^2 \leq \En[1\{w_i \leq \epsilon\}]\max_{i\leq n}\|\tilde x_i\|_\infty^2\frac{\|\tilde x_i'\delta\|_{2,n}^2}{\vartheta_\delta^2}. $$ The result follows.
\end{proof}

\subsection{Identification Lemmas}

In this section we collect new identification results for Logistic regression that might be of independent interest. We build upon the following technical lemma of \cite{Bach2010} which is based on (modified) self-concordant functions. However we will apply it differently than in \cite{Bach2010}. We exploit the separability of the objective function across observations and  make use of  the restricted non-linear impact coefficient \cite{BC-SparseQR}. In turn this allows us to weaken requirements of the analysis when compared to the literature.

\begin{lemma}[Lemma 1 from \cite{Bach2010}]\label{Lemma:SC}
Let $g:\RR \to \RR$ be a convex three times differentiable function such that for all $t \in \RR$, $|g'''(t)| \leq M g''(t)$ for some $M\geq 0$. Then, for all $t \geq 0$ we have
$$ \frac{g''(0)}{M^2} \left\{ \exp(-Mt) + Mt - 1\right\} \leq g(t) - g(0) - g'(0)t \leq \frac{g''(0)}{M^2} \left\{ \exp(Mt) + Mt - 1\right\}.$$
\end{lemma}
\begin{lemma}\label{Lemma:Auxtis} For $t\geq 0$ we have
$ \exp(-t)+t - 1 \geq \frac{1}{2}t^2 - \frac{1}{6}t^3.$
\end{lemma}
\begin{proof}[Proof of Lemma \ref{Lemma:Auxtis}]
For $t\geq 0$, consider the function $f(t) = \exp(-t) + t^3/6 - t^2/2 + t - 1$. The statement is equivalent to $f(t)\geq 0$ for $t\geq 0$. It follows that $f(0)=0$, $f'(0) = 0$, and $f''(t)=\exp(-t)+t-1\geq 0$ so that $f$ is convex. Therefore $f(t) \geq f(0)+tf'(0) = 0$.
\end{proof}

\begin{lemma}[Minoration Lemma]\label{Lemma:Minoration}
We have that
$$\Lambda(\eta_0 + \delta) - \Lambda(\eta_0) -\nabla \Lambda(\eta_0)'\delta \geq \left\{\mbox{$\frac{1}{3}$}\|\sqrt{w_i}\tilde x_i'\delta\|_{2,n}^2\right\} \wedge \left\{ \frac{\bar q_A}{3}\|\sqrt{w_i}\tilde x_i'\delta\|_{2,n}\right\} $$
\end{lemma}
\begin{proof}
 Step 1. (Minoration).   Define the maximal radius over which the following criterion function can be bounded below by a suitable quadratic function
$$ r_A = \sup_{r} \left\{ r \ : \begin{array}{rl}  & \Lambda(\eta_0+ \delta) - \Lambda(\eta_0) -\nabla \Lambda(\eta_0)'\delta \geq \frac{1}{3} \|\sqrt{w_i}\tilde x_i'\delta\|_{2,n}^2, \\
& \ \mbox{for all} \ \delta\in A, \ \|\sqrt{w_i}\tilde x_i'\delta\|_{2,n} \leq r\end{array}\right\}.$$
Step 2 below shows that  $r_{A} \geq \bar q_A$. By construction of $r_A$ and the convexity of $\Lambda(\eta_0+ \delta) - \Lambda(\eta_0) -\nabla \Lambda(\eta_0)'\delta$,
 $$ \begin{array}{lll}
 && \Lambda(\eta_0+ \delta) - \Lambda(\eta_0) -\nabla \Lambda(\eta_0)'\delta \geq  \\
&&     \geq \frac{\|\sqrt{w_i}\tilde x_i'\delta\|_{2,n}^2}{3} \wedge \left\{ \frac{\|\sqrt{w_i}\tilde x_i'\delta\|_{2,n}}{r_A} \cdot {\displaystyle \inf_{\tilde \delta\in A, \|\sqrt{w_i}\tilde x_i'\tilde\delta\|_{2,n} \geq r_A}} \!\! \Lambda(\eta_0+ \tilde \delta) - \Lambda(\eta_0) -\nabla \Lambda(\eta_0)'\tilde \delta \right\}\\
&&  \geq   \frac{\|\sqrt{w_i}\tilde x_i'\delta\|_{2,n}^2}{3} \wedge \left\{ \frac{\|\sqrt{w_i}\tilde x_i'\delta\|_{2,n}}{r_A} \frac{r_A^2}{3}\right\}
 \geq    \frac{\|\sqrt{w_i}\tilde x_i'\delta\|_{2,n}^2}{3} \wedge \left\{ \frac{\bar q_A}{3}\|\sqrt{w_i}\tilde x_i'\delta\|_{2,n}\right\}.
\end{array}$$

Step 2. ($r_A\geq \bar q_A$) Defining $g_i(t) = \log\{1+ \exp( \tilde x_i'\eta_0+t\tilde x_i'\delta)\}$  we have
$$\begin{array}{rl}
 & \Lambda( \eta_0+\delta) - \Lambda(\eta_0) -\nabla \Lambda(\eta_0)'\delta =\\ & = \En\left[ \log \{1+{\rm exp}(\tilde x_i'\{\eta_0+\delta\})\}-y_i \tilde x_i'(\eta_0+\delta)\right] \\
  & \ \ - \En\left[ \log \{1+{\rm exp}(\tilde x_i'\eta_0)-y_i \tilde x_i' \eta_0\}\right] - \En\left[ (\G_i-y_i) \tilde x_i'\delta\right]\\
 & = \En\left[ \log \{1+{\rm exp}(\tilde x_i'\{\eta_0+\delta\})\} - \log \{1+{\rm exp}(\tilde x_i'\eta_0)\} -  \G_i  \tilde x_i'\delta\right]\\
& = \En[ g_i(1) - g_i(0) - 1\cdot g_i'(0) ] \end{array}$$
Note that the function $g_i$ is three times differentiable and satisfies, for $\G_i(t) :=\exp( \tilde x_i'\eta_0 + t \tilde x_i'\delta )/\{1+\exp( \tilde x_i'\eta_0 + t \tilde x_i'\delta )\}$,
%$$ \begin{array}{rl}
% g'_i(t) & =  \frac{h_i\exp( \tilde x_i'\eta_0 + t h_i )}{1+\exp( \tilde x_i'\eta_0 + t h_i )} = h_i\G_i(t), \\
% g''_i(t) & = h_i^2  \frac{\exp( \tilde x_i'\eta_0 + t h_i)\{1-\exp( \tilde x_i'\eta_0 + t h_i )\} }{\{1+\exp( \tilde x_i'\eta_0 + t h_i )\}^2} =  h_i^2 \G_i(t)(1-\G_i(t)) \\
% g'''_i(t)& = h_i^3 \G_i(t)[1-\G_i(t)][1-2\G_i(t)] \\   \end{array}$$
$$  g'_i(t)  =   (\tilde x_i'\delta) \G_i(t), \ \ \
 g''_i(t)  =  (\tilde x_i'\delta)^2 \G_i(t)[1-\G_i(t)], \ \ \  g'''_i(t) = (\tilde x_i'\delta)^3 \G_i(t)[1-\G_i(t)][1-2\G_i(t)].$$
Thus $|g'''_i(t)|\leq |\tilde x_i'\delta| g''_i(t)$. Therefore, by Lemmas \ref{Lemma:SC} and \ref{Lemma:Auxtis} we have
$$\begin{array}{rl}
  g_i(1) - g_i(0) - 1\cdot g_i'(0) &   \geq \frac{(\tilde x_i'\delta)^2w_i}{(\tilde x_i'\delta)^2}\left\{ \exp(-|\tilde x_i'\delta|) + |\tilde x_i'\delta| -1 \right\}  \\
  & \geq w_i \left\{ \frac{|\tilde x_i'\delta|^2}{2} - \frac{|\tilde x_i'\delta|^3}{6} \right\}\end{array}$$
Therefore we have
$$\begin{array}{rl}
\Lambda(\eta_0+\delta) - \Lambda(\eta_0) -\nabla \Lambda(\eta_0)'\delta & \displaystyle \geq \mbox{$\frac{1}{2}$}\En\left[ w_i|\tilde x_i'\delta|^2\right] - \mbox{$\frac{1}{6}$}\En\left[ w_i|\tilde x_i'\delta|^3\right]\\ \end{array}$$

%%%%
%%%%

Note that for any $\delta\in A$ such that $ \|\sqrt{w_i}\tilde x_i'\delta\|_{2,n} \leq  \bar q_A$ we have $$ \|\tilde x_i'\delta\|_{2,n} \leq  \bar q_A \leq \|\sqrt{w_i}\tilde x_i'\delta\|_{2,n}^{3}/\En\[w_i|\tilde x_i'\delta|^3\],$$
so that $\En[w_i|\tilde x_i'\delta|^3] \leq \En[w_i|\tilde x_i'\delta|^2]$. Therefore we have
$$\begin{array}{rl}
\Lambda(\eta_0+\delta) - \Lambda(\eta_0) -\nabla \Lambda(\eta_0)'\delta & \displaystyle \geq \mbox{$\frac{1}{2}$}\En\left[ w_i|\tilde x_i'\delta|^2\right] - \mbox{$\frac{1}{6}$}\En\left[ w_i|\tilde x_i'\delta|^3\right] \\
&\geq \frac{1}{3}\En\left[ w_i|\tilde x_i'\delta|^2\right]\\ \end{array}$$

\end{proof}

\subsection{Penalty Choice and Rate for $\ell_1$-Penalized Logistic Regression}

Next we establish a simple (and known) bound for the choice of the penalty level $\lambda$ within Lasso-Logistic under standard normalization. Refinements are possible under additional mild assumptions on the covariates.

\begin{lemma}[Choice of Penalty, Hoeffding's Inequality]\label{lemma:PenaltyHoeffding}
Assume that $\En[\tilde x_{ij}^2]=1$. Then, for any $\gamma \in (0,1)$ we have $$ P \left( \|\nabla \Lambda(\eta_0) \|_\infty \leq \sqrt{2 \log(2(p+1)/\gamma) / n } \right) \leq \gamma.$$
\end{lemma}
\begin{proof} Let $\G_i = \Ep[y_i\mid \tilde x_i] = \frac{{\rm exp}( \tilde x_i'\eta_0)}{1+{\rm exp}( \tilde x_i'\eta_0)}$, so that $\|\nabla \Lambda(\eta_0) \|_\infty = \left\| \En[\left( y_i - \G_i\right) \tilde x_i] \right\|_\infty$. Then
$$ P( \left\| \En[\left( y_i - \G_i\right) \tilde x_i] \right\|_\infty \geq t ) \leq (p+1)\max_{j\leq p} P( | \En[\left( y_i - \G_i\right) \tilde x_{ij}] | \geq t ) \leq 2(p+1){\rm exp}( - t^2n/ 2  ). $$
\end{proof}

%%% Not Finished %%%%
%%%%%%%%%%%%%%%%%%%%%

\begin{lemma}[Choice of Penalty, Self-Normalized Moderate Deviation Theory]\label{Lemma:PenaltySNDT}
Normalize the covariates so that $\En[\tilde x_{ij}^2]=1$, let $l_j = \sqrt{\En[w_i\tilde x_{ij}^2]}$, and $\hat l_j=\sqrt{\En[\hat w_i\tilde x_{ij}^2]}$. Assume that $K_{\tilde x}^2 \log p \leq n\delta_n \min_j l_j^2$, $\Phi^{-1}(1-2p/\gamma) \leq \delta_n n^{1/3}$, and
$ \|\hat w_i - w_i\|_{2,n} K_{\tilde x} \leq \delta_n \min_j l_j^2.$ Then, setting $\hat \Gamma = \diag(\hat l)$, for any $\gamma \in (0,1)$ and $\mu > 0$, for $n$ sufficiently large we have $$ P \left( \|\hat\Gamma^{-1}\nabla \Lambda(\eta_0) \|_\infty \leq \{1+ \mu \}\Phi^{-1}(1-\gamma/[2p])/\sqrt{n} \right) \leq \gamma + o(1).$$
\end{lemma}
\begin{proof}
Let $\Gamma = \diag(l)$, $\tilde l_j = \sqrt{\En[(y_i-\G_i)^2\tilde x_{ij}^2]}$, and $\widetilde \Gamma = \diag(\tilde l)$. We have
$$\begin{array}{rl}
\|\hat\Gamma^{-1}\nabla \Lambda(\eta_0) \|_\infty & \leq  \|\{\hat\Gamma^{-1}-\Gamma^{-1}+\Gamma^{-1}-\widetilde\Gamma^{-1}\} \widetilde \Gamma  \widetilde \Gamma^{-1} \nabla \Lambda(\eta_0)\|_\infty +  \|\widetilde \Gamma^{-1} \nabla \Lambda(\eta_0)\|_\infty\\
&\leq \{\|\{\hat\Gamma^{-1}-\Gamma^{-1}\}\widetilde\Gamma\|_\infty +\|\{\Gamma^{-1}-\widetilde\Gamma^{-1}\} \widetilde \Gamma  \|_\infty\} \|\widetilde \Gamma^{-1} \nabla \Lambda(\eta_0)\|_\infty +  \|\widetilde \Gamma^{-1} \nabla \Lambda(\eta_0)\|_\infty\\
& \leq \left\{ \max_{j\leq p} \left|\frac{\tilde l_j}{l_j}\frac{l_j-\hat l_j}{\hat l_j}\right| + \max_{j\leq p} \left|\frac{\tilde l_j}{l_j}\frac{\tilde l_j- l_j}{\tilde l_j}\right| +1 \right\} \|\widetilde \Gamma^{-1} \nabla \Lambda(\eta_0)\|_\infty.\end{array}$$
Since $w_i$ and $\hat w_i$ are non-negative we have $$\begin{array}{rl}
\max_{j\leq p}|l_j - \hat l_j| & \leq  \max_{j\leq p} \sqrt{\En[|\hat w_i - w_i|\tilde x_{ij}^2]} \leq \|\hat w_i - w_i\|_{2,n}^{1/2}\max_{j\leq p} \{\En[\tilde x_{ij}^4] \}^{1/4}.\end{array}$$
Also, since $\Ep[(y_i-\G_i)^2\mid \tilde x_i] = w_i$ and for positive number $|\sqrt{a}-\sqrt{b}|\leq \sqrt{|a-b|}$, we have
 $$\begin{array}{rl}
\max_{j\leq p}\left| \tilde l_j - l_j \right| & = \max_{j\leq p}\left| \sqrt{\En[(y_i-\G_i)^2\tilde x_{ij}^2]} - \sqrt{\barEp[w_i\tilde x_{ij}^2]} \right|\\
& \leq  \sqrt{\max_{j\leq p}|(\En-\barEp)[(y_i-\G_i)^2\tilde x_{ij}^2]|}\\
\end{array}$$
By Lemma \ref{cor:LoadingConvergence} we have
$$ \max_{j\leq p}|(\En-\barEp)[(y_i-\G_i)^2\tilde x_{ij}^2] | \lesssim_P \sqrt{\frac{\log p}{n}} \max_{j\leq p} \{\En[\tilde x_{ij}^4] \}^{1/2} $$
Therefore for $n$ large enough we have $\max_{j\leq p}\frac{| \hat l_j - l_j|}{l_j}\vee \frac{|\tilde l_j -l_j|}{l_j} \leq \mu/16$ under the assumed growth conditions with probability $1-o(1)$. In the same event we have
$$ \|\hat\Gamma^{-1}\nabla \Lambda(\eta_0) \|_\infty  \leq \{1 + \mu/2 \} \| \widetilde \Gamma^{-1} \nabla \Lambda(\eta_0)\|_\infty.$$ Finally, by self-normalized moderate deviation theory we have
$$ P( \| \widetilde \Gamma^{-1} \nabla \Lambda(\eta_0)\|_\infty > t) \leq p\max_{j\leq p} P\left( \frac{\En[(y_i-\G_i)\tilde x_{ij}]}{\sqrt{\En[(y_i-\G_i)^2\tilde x_{ij}^2]}} > t \right) \leq 2p\Phi^{-1}(1-\gamma/[2p]) \{ 1 + O(\delta_n)\} $$
\end{proof}

\begin{remark}\label{Remark:Penalty}
Note that we can replace $(\hat w_i)_{i=1}^n$ with $(\bar w_i)_{i=1}^n$ in Lemma \ref{Lemma:PenaltySNDT} if $\bar w_i \geq w_i$ by construction. For instance $w_i \leq \bar w_i := 1/4$. Therefore it is valid to use
$\lambda = \frac{c}{2}\sqrt{n}\Phi^{-1}(1-\gamma/[2p])$ and $\hat l_j = 1$ for $c>1$.
\end{remark}

\begin{lemma}\label{Lemma:LassoLogisticRateRaw}
Assume $\lambda/n \geq c \|\nabla \Lambda(\eta_0) \|_\infty $, $c>1$ and let $\cc = (c+1)/(c-1)$.
Provided that $\bar q_{\Delta_\cc} > 3(1+\frac{1}{c})\lambda\sqrt{s}/(n\kappa_\cc)$
$$ \|\sqrt{w_i}\tilde x_i'(\hat\eta-\eta_0)\|_{2,n} \leq 3(1+\mbox{$\frac{1}{c}$})\frac{\lambda\sqrt{s}}{n\kappa_\cc} \ \ \ \mbox{and} \ \ \ \|\hat \eta -\eta_0\|_1 \leq 3\frac{(1+c)(1+\cc)}{c}\frac{\lambda s}{n\kappa_\cc^2}$$
\end{lemma}
\begin{proof}
Let $\delta = \hat\eta - \eta_0$. By definition of $\hat \eta$ in (\ref{Eq:LassoLogstic}) we have
$
\Lambda(\hat \eta) + \frac{\lambda}{n}\|\hat\eta\|_1  \leq \Lambda(\eta_0) + \frac{\lambda}{n}\|\eta_0\|_1$. Thus,
$$
\begin{array}{rl}
\Lambda(\hat \eta) - \Lambda(\eta_0) & \leq \frac{\lambda}{n}\|\eta_0\|_1 - \frac{\lambda}{n}\|\hat\eta\|_1 \\
%& =  \frac{\lambda}{n}(\|\eta_0\|_1 - \|\hat\eta_T\|_1) - \frac{\lambda}{n} \|\hat\eta_{T^c}\|_1\\
& \leq \frac{\lambda}{n}\|\delta_T\|_1 - \frac{\lambda}{n} \|\delta_{T^c}\|_1\\
\end{array}
$$
However, by convexity of $\Lambda(\cdot)$ and Holder inequality we have $$
\begin{array}{rl}
\Lambda(\hat \eta) - \Lambda(\eta_0) & \geq %\nabla \Lambda(\eta_0)'(\hat\eta - \eta_0) \\ & \geq
-  \|\nabla \Lambda(\eta_0) \|_\infty \|\delta\|_1\\
%& = -  \|\nabla \Lambda(\eta_0) \|_\infty (\| \delta_T\|_1 + \|\delta_{T^c}\|_1)\\
& \geq -  \frac{\lambda}{n}\frac{1}{c} \| \delta_T\|_1 -  \frac{\lambda}{n}\frac{1}{c}\|\delta_{T^c}\|_1\\
\end{array}
$$
Combining these relations we have
$ -  \frac{\lambda}{n}\frac{1}{c} \| \delta_T\|_1 - \frac{\lambda}{n}\frac{1}{c} \|\delta_{T^c}\|_1 \leq  \frac{\lambda}{n}\|\delta_T\|_1 - \frac{\lambda}{n} \|\delta_{T^c}\|_1,$ which leads to $\|\delta_{T^c}\|_1 \leq  \cc \|\delta_T\|_1.$

By Lemma \ref{Lemma:Minoration} with $A = \Delta_\cc$ and the reasoning above we have
$$ \begin{array}{rl}
\frac{1}{3}\|\sqrt{w_i}\tilde x_i'\delta\|_{2,n}^2 \wedge \left\{ \frac{\bar q_A}{3}\|\sqrt{w_i}\tilde x_i'\delta\|_{2,n}\right\} & \leq \Lambda(\hat \eta) - \Lambda(\eta_0) -\nabla \Lambda(\eta_0)'\delta \\
& \leq \frac{\lambda}{n}\|\delta_T\|_1 - \frac{\lambda}{n}\|\delta_{T^c}\|_1 + \|\nabla \Lambda(\eta_0)\|_\infty \|\delta\|_1 \\
& \leq (1+\frac{1}{c})\frac{\lambda}{n}\|\delta_T\|_1  \leq (1+\frac{1}{c})\frac{\lambda\sqrt{s}}{n}\|\delta_T\|\\
& \leq (1+\frac{1}{c})\frac{\lambda\sqrt{s}}{n}\|\sqrt{w_i}\tilde x_i'\delta\|_{2,n}/\kappa_\cc\\
\end{array}$$
Provided that $\bar q_A>3(1+\frac{1}{c})\lambda\sqrt{s}/(\kappa_\cc n)$, so that the minimum on the LHS needs to be the quadratic term, we have
$$ \|\sqrt{w_i}\tilde x_i'\delta\|_{2,n} \leq 3(1+\mbox{$\frac{1}{c}$})\frac{\lambda\sqrt{s}}{n\kappa_\cc}$$
\end{proof}

\subsection{Sparsity of Lasso-Logistic}

We begin by establishing sparsity bounds which do not rely on large penalty choices nor on the irrepresentability condition\footnote{The irrepresentability condition is the assumption that $\|\En[\tilde x_{iT^c}\tilde x_{iT}](\En[\tilde x_{iT}\tilde x_{iT}])^{-1}\sign(\eta_{0T})\|_\infty<1$.} The (data-driven) sparsity is fundamental for the analysis of the rate of convergence of the Post-Lasso-Logistic estimator. The following lemma is useful.

\begin{lemma}\label{Lemma:AuxSparsity}
The logistic link function satisfies $|\G(t+t_0)-\G(t_0)|\leq G'(t_0)\{\exp(|t|) -1\}$. If $|t|\leq 1$ we have $\exp(|t|) -1\leq 2|t|$.
\end{lemma}
\begin{proof}
Note that $|\G''(s)|\leq \G'(s)$ for all $s$. So that $-1\leq \frac{d}{ds}\log(\G'(s)) = \frac{\G''(s)}{\G'(s)} \leq 1$. Suppose $s\geq 0$. Therefore
$$ - s \leq \log(G'(s+t_0)) - \log(G'(t_0)) \leq s. $$
In turn this implies $\G'(t_0)\exp(-s)\leq \G'(s+t_0) \leq \G'(t_0)\exp(s)$. Integrating one more time from $0$ to $t$,
$$ \G'(t_0)\{1-\exp(-t)\}\leq \G(t+t_0) - \G(t_0) \leq \G'(t_0)\{\exp(t)-1\}.$$
The first result follows by noting that $1-\exp(-t) \leq \exp(t)-1$. The second follows by verification.
\end{proof}

\begin{lemma}[Sparsity]\label{Lemma:LassoLogisticSparsity}
Consider $\hat \eta$ as defined in (\ref{Eq:LassoLogstic}), let $\hat s = |\supp(\hat\eta)|$ and suppose $\lambda/n \geq c\|\nabla \Lambda(\eta_0)\|_\infty$. Then
$$ \hat s \leq \frac{c^2(n/\lambda)^2}{(c-1)^2}\semax{\hat s}\|\tilde x_i'(\hat\eta-\eta_0)\|_{2,n}^2.$$
Provided that $\bar q_{\Delta_\cc} > 3(1+\frac{1}{c})\lambda\sqrt{s}/(n\kappa_\cc)$ we have
 $$ \hat s \leq s\cdot \semax{\hat s}\frac{9\cc^2}{\{\psi_{(r)}(\cc)\}^2\kappa_\cc^2}.$$
Moreover, if $\frac{3(1+c)(1+\cc)}{c}\frac{\lambda s}{n\kappa_\cc^2}\max_{i\leq n}\|\tilde x_i\|_\infty  \leq 1$ we have
$$\sqrt{\hat s} \leq  6\cc\frac{\sqrt{\semax{\hat s}}}{\kappa_\cc}\sqrt{s} \ \ \ \mbox{and} \ \ \ \hat s \leq s \cdot 36 \cc^2\min_{m\in \mathcal{M}} \frac{\semax{m}}{\kappa_{\cc}^2}$$ where $\mathcal{M}=\{ m \in \NN : m > 72s\cc^2\semax{m}/\kappa_{\cc}^2  \}$\end{lemma}
\begin{proof}
Let $\hat T = \supp(\hat\eta)$, $\hat s = |\hat T|$, $\delta = \hat\eta-\eta_0$, and $\hat \G_i = \exp(\tilde x_i'\hat\eta)/\{1+\exp(\tilde x_i'\hat\eta)\}$. For any $j\in \hat T$ we have $|\nabla_j\Lambda(\hat\eta)|=|\En[(y_i-\hat \G_i)\tilde x_{ij}]|=\lambda/n$.

The first relation follows from
$$\begin{array}{rl}
\frac{\lambda}{n}\sqrt{\hat s} & = \| \En[(y_i-\hat \G_i)\tilde x_{i\hat T}] \|_2 \\
& \leq \| \En[(y_i- \G_i)\tilde x_{i\hat T}] \|_2 + \| \En[(\hat \G_i-\G_i)\tilde x_{i\hat T}] \|_2\\
& \leq \sqrt{\hat s}\| \En[(y_i- \G_i)\tilde x_{i\hat T}] \|_\infty + \| \En[\tilde x_i'\delta\tilde x_{i\hat T}] \|_2\\
& \leq \frac{\lambda}{cn}\sqrt{\hat s} + \sqrt{\semax{\hat s}}\|\tilde x_i'\delta\|_{2,n}\\
\end{array} $$

The second follows from the first, the definition of $\psi_{(r)}(\cc)$, and Lemma \ref{Lemma:LassoLogisticRateRaw} so that
$$ \hat s \leq \frac{c^2(n/\lambda)^2}{(c-1)^2}\semax{\hat s}\|\tilde x_i'\delta\|_{2,n}^2 \leq \frac{c^2(n/\lambda)^2}{(c-1)^2}\semax{\hat s}\frac{\|\sqrt{w_i}\tilde x_i'\delta\|_{2,n}^2}{\psi_{(r)}(\cc)^2} \leq s \cdot \semax{\hat s}\frac{9\cc^2}{\psi_{(r)}(\cc)^2\kappa_{\cc}^2} $$
The third relation follows from
$$\begin{array}{rl}
\frac{\lambda}{n}\sqrt{\hat s} & = \| \En[(y_i-\hat \G_i)\tilde x_{i\hat T}] \|_2 \\
& \leq \| \En[(y_i- \G_i)\tilde x_{i\hat T}] \|_2 + \| \En[(\hat \G_i-\G_i)\tilde x_{i\hat T}] \|_2\\
& \leq \sqrt{\hat s}\| \En[(y_i- \G_i)\tilde x_{i\hat T}] \|_\infty + \sup_{\|\theta\|_0\leq |\hat T|,\|\theta\|=1} \En[ |\hat\G_i-\G_i| \cdot |\tilde x_i'\theta|] \\
& \leq \frac{\lambda}{cn}\sqrt{\hat s} + 2\sqrt{\semax{\hat s}}\|\sqrt{w_i}\tilde x_i'\delta\|_{2,n}\end{array} $$
where we used Lemma \ref{Lemma:AuxSparsity}  so that $
|\hat \G_i - \G_i| \leq w_i 2|\tilde x_i'\delta|$ since by Lemma \ref{Lemma:LassoLogisticRateRaw} we have $\|\delta\|_1 \leq 3\frac{(1+c)(1+\cc)}{c}\frac{\lambda s}{n\kappa_\cc^2}$ so that $\max_{i\leq n}\|\tilde x_i\|_\infty \|\delta\|_1 \leq 1$ by the assumed condition.

Therefore, by the $\|\cdot\|_{2,n}$ bound in Lemma \ref{Lemma:LassoLogisticRateRaw} we have
$$\begin{array}{rlr}
(1-\mbox{$\frac{1}{c}$})\frac{\lambda}{n}\sqrt{\hat s} & \leq  6\sqrt{\semax{\hat s}}\frac{(1+c)}{c}\frac{\lambda \sqrt{s}}{n\kappa_\cc}\\
\end{array}$$
which implies $\sqrt{\hat s} \leq  6\cc\frac{\sqrt{\semax{\hat s}}}{\kappa_\cc}\sqrt{s}$.

The last relation follows by the previous result and the fact that sparse eigenvalues are sublinear functions.
\end{proof}

\subsection{Post model selection Logistic regression rate}

\begin{lemma}\label{Lemma:PostLassoLogisticRateRaw} Consider $\widetilde \eta$ as defined in (\ref{Eq:PostLassoLogstic}). Let $\hat s^* := |\widehat T^*|$. We have
 $$ \|\sqrt{w_i}\tilde x_i'(\widetilde \eta - \eta_0)\|_{2,n} \leq \sqrt{3}\sqrt{\max\{0,\Lambda(\tilde \eta) - \Lambda(\eta_0)\}} + 3\sqrt{\hat s^*+s}\|\nabla \Lambda(\eta_0)\|_\infty /\sqrt{\semin{\hat s^*+s}}$$
 provided that $\bar q_A/6>\sqrt{\hat s^*+s}\|\nabla \Lambda(\eta_0)\|_\infty /\sqrt{\semin{\hat s^*+s}}$ and $q_A/6>\sqrt{\max\{0,\Lambda(\tilde \eta) - \Lambda(\eta_0)\}}$ for $A = \{ \delta \in \RR^p : \|\delta\|_0 \leq \hat s^* + s\}$.
\end{lemma}
\begin{proof}
Let $\tilde \delta=\tilde \eta - \eta_0$ and $\tilde t_{2,n} = \|\sqrt{w_i}\tilde x_i'\tilde \delta\|_{2,n}$.
By Lemma \ref{Lemma:Minoration} with $A = \{\delta \in \RR^p : \|\delta\|_0 \leq \hat s^*+s\}$, we have
$$ \begin{array}{rl}
\frac{1}{3}\tilde t_{2,n}^2 \wedge \left\{ \frac{\bar q_A}{3}\tilde t_{2,n}\right\} & \leq \Lambda(\tilde \eta) - \Lambda(\eta_0) -\nabla \Lambda(\eta_0)'\tilde\delta \\
& \leq \Lambda(\tilde \eta) - \Lambda(\eta_0) + \|\nabla \Lambda(\eta_0)\|_\infty \|\tilde\delta\|_1 \\
& \leq \max\{0,\Lambda(\tilde \eta) - \Lambda(\eta_0)\} + \tilde t_{2,n}\sqrt{\hat s^*+s}\|\nabla \Lambda(\eta_0)\|_\infty /\sqrt{\semin{\hat s^*+s}} \\
\end{array}$$
Provided that $\bar q_A/6>\sqrt{\hat s^*+s}\|\nabla \Lambda(\eta_0)\|_\infty /\sqrt{\semin{\hat s^*+s}}$ and $\bar q_A/6>\sqrt{\max\{0,\Lambda(\tilde \eta) - \Lambda(\eta_0)\}}$, if the minimum on the LHS is the linear term, we have $\tilde t_{2,n}\leq \sqrt{\max\{0,\Lambda(\tilde \eta) - \Lambda(\eta_0)\}}$ which implies the result. Otherwise, since for positive numbers $a^2 \leq b + ac$ implies $a\leq \sqrt{b} + c$, we have
$$ \tilde t_{2,n} \leq \sqrt{3}\sqrt{\max\{0,\Lambda(\tilde \eta) - \Lambda(\eta_0)\}} + 3\sqrt{\hat s^*+s}\|\nabla \Lambda(\eta_0)\|_\infty /\sqrt{\semin{\hat s^*+s}}.$$
\end{proof}

%\bibliographystyle{plain}
%\bibliography{mybibSparseLogistic}

\section{Additional Monte Carlo}

\subsection{Monte Carlo for Approximately Sparse Models}

In this section we provide further simulations to illustrate the performance of the proposed methods. In particular we illustrate the performance of the method when applied to approximately sparse models. We consider a similar design to the one used in Section \ref{SecMC} of the main text, namely
$$ \Ep[y\mid d, x] = \G( d\alpha_0 + x'\{c_y\nu_y\} ), \  \ \ \  d = x'\{c_d\nu_d\} + v.  $$
However, the vectors $\nu_y$ and $\nu_d$ are set to \begin{equation}\label{DecayCoef}\begin{array}{l}\nu_{yj} = 1/j^2,
\nu_{dj} = 1/j^2,\end{array}\end{equation} so they are approximately sparse. Again we let $x = (1,z')'$ consists of an intercept and covariates $z \sim N(0,\Theta)$, and the error $v$ is i.i.d. as $N(0,1)$. The dimension $p$ of the covariates $x$ is $250$, and the sample size $n$ is $200$. The regressors are correlated with $\Theta_{ij} = \rho^{|i-j|}$ and $\rho = 0.5$. As before the coefficient $c_d$ is used to control the $R^2$ of the reduce form equation, $c_y$ is set similarly and in every repetition, we draw new errors $v_{i}$'s and controls $x_i$'s.
The figures display the results over $100$ different designs where $\alpha_0 = 0.5$ and the values of $c_y$ and $c_d$ are set to achieve $R^2=\{0,0.1,0.2,0.3,0.4,0.5,0.6,0.7,0.8,0.9\}$ for each equation. There were 1000 replications for each of the 100 designs.

Figure \ref{Fig:alpha05QuadraticDecay} reveals that the performance of the method for this approximately sparse design is very similar to the performance obtained with the sparse designs considered in Section \ref{SecMC}. Again the double selection estimator arise as a more reliable estimator.

\bibliographystyle{plain}
\bibliography{mybibSparseLogistic}

\begin{figure}[ht]
\includegraphics[width=\textwidth]{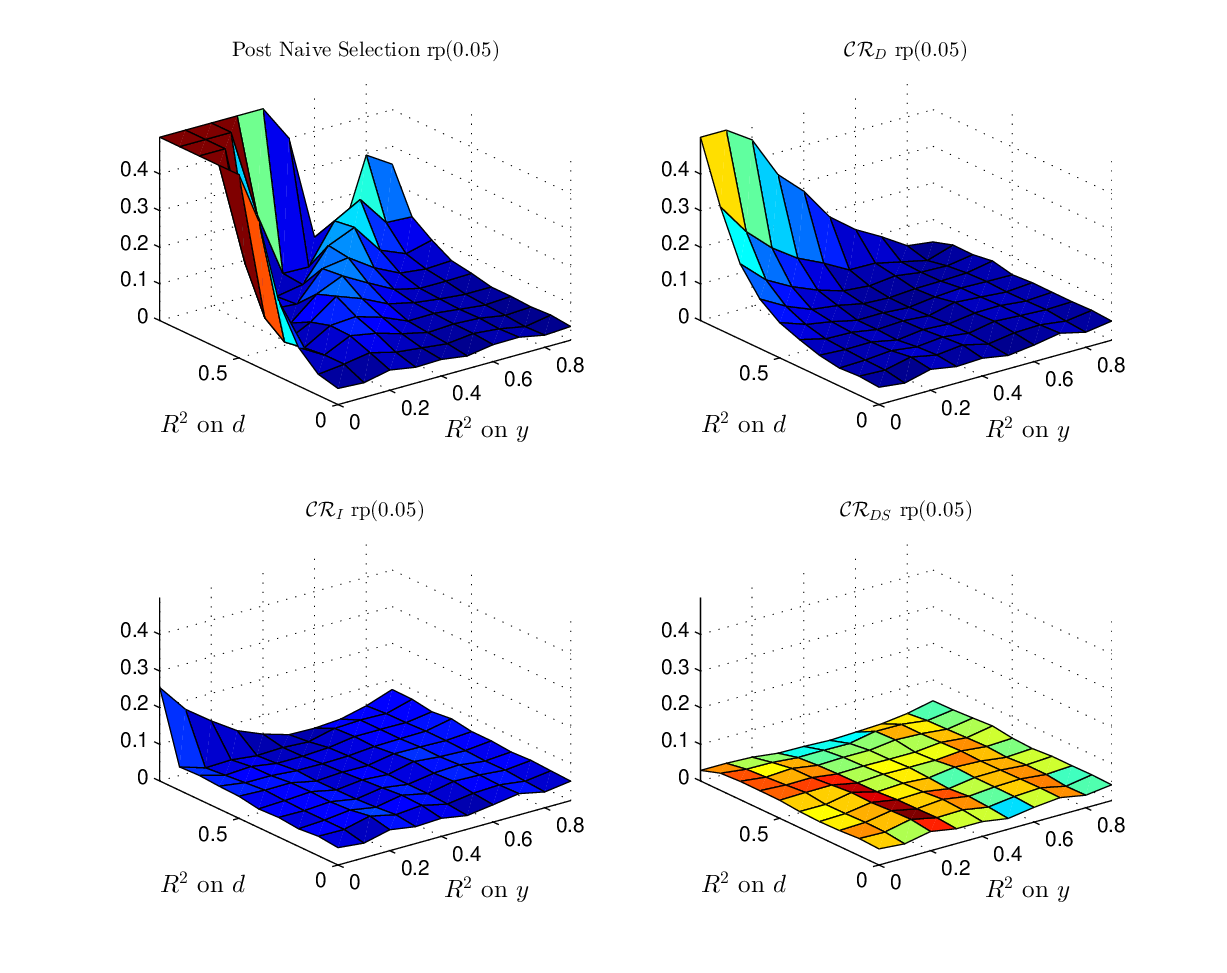}
\caption{\footnotesize For the approximately sparse model defined by (\ref{DecayCoef}), the figures display the rp(0.05) of the naive post selection estimator and the proposed confidence regions based on optimal instrument ($\CR_D$ and $\CR_I$) and double selection ($\CR_{DS}$). There are a total of 100 different designs with $\alpha_0=0.5$. The results are based on 1000 replications for each design.}\label{Fig:alpha05QuadraticDecay}
\end{figure}

\end{document}